\documentclass{amsart}
\usepackage[utf8]{inputenc}
\usepackage{graphicx}
\usepackage{amsmath}
\usepackage{graphicx}
\usepackage{epsfig}
\usepackage{amssymb}
\usepackage{amsthm}
\usepackage{amsfonts}
\usepackage{amsmath}
\usepackage{placeins}
\usepackage{appendix}
\usepackage{xcolor}

\DeclareMathOperator*{\argmin}{argmin}

\usepackage[text={440pt,575pt},headheight=9pt,centering]{geometry}

\usepackage{multicol,multirow}
\usepackage{algorithm,algcompatible}
\usepackage[normalem]{ulem}
\usepackage{enumitem}
\usepackage{float}
\usepackage{booktabs}
\usepackage{natbib}

\newtheorem{theorem}{Theorem}
\newtheorem{lemma}[theorem]{Lemma}
\newtheorem{example}[theorem]{Example}

\newtheorem{proposition}[theorem]{Proposition}
\newtheorem{defn}{Definition}

\newtheorem{assumption}{Assumption}
\newtheorem{remark}{Remark}
\usepackage{algorithm}
\usepackage{algpseudocode}

\newcommand{\G}{\mathcal{G}}

\usepackage[foot]{amsaddr}

\usepackage{hyperref}

\newcommand{\nk}[1]{\textcolor{orange}{[NK: #1]}}
\newcommand{\at}[1]{\textcolor{red}{[AT: #1]}}

\title{Integrating Multi-source Local Information for Global Undirected Graph Reconstruction via Integer Programming}
\date{}
\author[N.~Keret]{Nir Keret$^{1*}$}
\email{nirkeret@uw.edu}
\address{$^{1}$Department of Biostatistics, University of Washington, U.S.}

\author[A.~Taeb]{Armeen Taeb$^{2*}$}
\email{ataeb@uw.edu}
\address{$^{2}$Department of Statistics, University of Washington, U.S.}

\author[A.~Shojaie]{Ali Shojaie$^{1,2}$}
\email{ashojaie@uw.edu}

\address{$^*$ These authors contributed equally to this work.}

\begin{document}

\begin{abstract}
    Graphical models are widely used to represent conditional independence structures, but learning them is challenging when data are fragmented across sources that measure different subsets of variables and may not be directly shareable. We develop a general framework for reconstructing a shared global graph from heterogeneous, partially observed sources. Our framework permits both the observed variable sets and source-specific distributions to differ. Under a source-wise faithfulness condition, local conditional independence graphs correspond to latent projections of the shared global graph, which is identifiable up to the equivalence class induced by the observed variable sets. The framework integrates raw data, estimated source graphs, or a combination of both, accommodating settings with restricted access to individual-level data, such as privacy constraints. We develop mixed-integer optimization formulations and scalable approximate algorithms for several classes of graphical models. Unlike standard high-dimensional analyses, our guarantees are governed by new combinatorial notions of signal strength that capture how information propagates across partially observed sources. Extensive simulations further show that integration can substantially improve estimation of source-specific projected graphs relative to separate source methods. Real-data application illustrates the framework's use with partially-observed sources of fatty acid measurements.
\end{abstract}

\maketitle
\vspace{-0.1in}

\section{Introduction}\label{sec:intro}

Probabilistic graphical models (PGMs) provide a succinct representation of the joint distribution of multiple random variables in a graph where nodes correspond to variables and edges encode conditional dependencies \citep{lauritzen1996graphical}. Such a graph becomes an object of direct scientific interest in many applications: it can reveal regulatory relationships among genes, functional connectivity among brain regions, dependencies among financial assets, or interactions among species in an ecosystem. As such, learning the structure of a PGM from data has become one of the central problems in modern statistics and machine learning \citep{drton2017structure}.

A substantial body of work has developed methods for estimating the structure of undirected graphical models in high-dimensional settings, where the number of variables can exceed the available sample size. Most notable is the setting of Gaussian graphical models (GGMs), where the graph is encoded in the sparsity pattern of the inverse covariance, or precision, matrix \citep{dempster1972covariance}. In this case, the graph structure can be efficiently estimated in high dimensions by using penalized regression via the neighborhood selection approach \citep{meinshausen2006high}, or by directly solving the penalized likelihood problem \citep{yuan2007model}, using, e.g., the graphical lasso algorithm \citep{friedman2008sparse}. Penalized estimation approaches have also been developed for estimation of the graph structure from non-Gaussian observations, including Gaussian copula graphical models \mbox{\citep{liu2012high, xue2012regularized}}, as well as binary and Poisson random variables \citep{ravikumar2010high, yang2013poisson}. More generally, penalized estimation of graphical models has also been extended to models with exponential-family node-conditional distributions \citep{Yang2015GraphicalMV}, mixed graphical models allowing different node-conditional distribution families \citep{chen2015selection}, and models with nonlinear additive conditional means \citep{voorman2014graph}.

The above penalized estimation methods have been successfully adopted in applications and have become standard tools for exploratory network analysis. However, their standard formulations assume that all variables are jointly observed on a common set of samples. There are many important applications in which data relevant to a single underlying graphical model cannot be pooled into one data set, either because they are generated and stored separately or because they cannot be shared across sites. In multi-center clinical and genomic studies, patient-level data are often siloed within individual hospitals or institutions due to privacy regulations and data use agreements, preventing the construction of a single combined data set even when all sites measure the same clinical and molecular variables \citep{tertulino2023privacy}. In neuroscience, large-scale recordings of neural activity are frequently collected across multiple non-simultaneous sessions, multiple animals, or multiple labs, so that any single recording session captures only a subset of the neurons or brain regions of interest \citep{siegle2026large}. In sensor networks, physical or bandwidth constraints mean that different subsets of sensors report to different collection points, and the full sensor array is never observed simultaneously. In each of these settings, data are fragmented: they are distributed across sources that cannot be merged, whether because of legal, logistical, or physical barriers.

This fragmentation motivates a fundamentally different estimation problem: rather than learning a graphical model from a single data set, we need algorithms that can integrate statistical evidence about a shared underlying graph from multiple, separately-held data sources. This problem is particularly challenging in the setting of graphical models: in addition to the usual challenges of estimation across multiple data sources, namely potential differences among distributions and the lack of access to the raw observations, observing different sets of variables across different sites means that conditional independence relationships inferred in each site, and hence the corresponding graph structures, may be drastically different. Several strands of literature address different aspects of this integration problem. 

The first class of methods has focused on integrating multiple data sets that share the same set of observed variables but arise from distinct, possibly heterogeneous, populations. This includes the literature on jointly estimating multiple related GGMs, where each source has its own precision matrix, but these precision matrices are encouraged to have \emph{similar sparsity patterns or parameter values} \citep{guo2011joint, danaher2014joint, saegusa2016joint,phan2017novel,liu2026learning}. A related Bayesian approach shares information about edge inclusion across groups without requiring similar nonzero precision values \citep{peterson2015bayesian}.
Related approaches have also leveraged multi-task learning \citep{honorio2010multi} and population-prior approaches for improving individual brain connectivity estimates using group-level information \citep{varoquaux2010brain}. However, all of these approaches require that the \emph{same set of variables} be observed in every source, while allowing source-specific distributions. They are therefore not directly applicable when different sources observe different, possibly only partially overlapping, subsets of variables. Jewel and Jewel 2.0 \citep{angelini2021jewel,angelini2022jewel} permit different observed variable sets, with the latter also allowing source-specific edges. For a given pair of variables, they encourage the same edge pattern across the different sources. However, edges across sources are not directly comparable, as their presence or absence reflects conditional independence relationships involving different conditioning sets when the observed variable sets differ.


 The second class of methods comes from the distributed and federated learning literature, which was developed primarily to address computational and privacy constraints when a large data set is split across multiple machines, or when data reside on different devices or institutions that cannot share raw data \citep{mcmahan2017communication}. Several algorithms have recently leveraged these ideas for estimation and inference \citep{battey2018distributed}, including estimation of precision matrices \citep{arroyo2016efficient, wang2021efficient, nezakati2023unbalanced}. Complementary strategies based on combining summary statistics across sites have also been proposed \citep{liu2015multivariate, tang2020distributed}. These distributed learning methods are effective at reducing computation and respecting data-access constraints. However, the distributed precision-matrix formulations also assume a common set of observed variables and do not directly address graph reconstruction from different observed subsets. In this setting, the central statistical challenge is no longer simply combining samples across sources with a shared variable set, but inferring a single underlying graph over the \emph{union} of all variables when no single source observes them all. 


The third class of methods allows for partial observations but encourages similarity in distributions across the sources. One such line of work is the ``graph quilting'' framework in \cite{vinci2023graph}. In the Gaussian setting, an $\ell_1$-regularized partially-observed likelihood estimator, combined with thresholding and Schur-complement calculations, can recover edges connecting jointly observed variable pairs and a superset of edges connecting never-jointly-observed pairs. These guarantees hold under certain conditions, with minimality of the superset established under additional assumptions. Related work uses low-rank covariance completion before graph estimation \citep{chang2022low} and extends the framework to Gaussian copula models using rank-based correlation estimates \citep{Chang2023NonparanormalGQ}. Another line of work that handles partial observations is \citep{tozzo2020missing}, which uses the EM algorithm and couples source-specific precision matrices through $\ell_1$ penalties on their differences. The methods described in this class all accommodate multiple data sources with different observed variable subsets and sample sizes. However, these formulations rely on a common covariance or latent correlation matrix across sources. They therefore do not directly address sources whose numerical dependence parameters differ while a common graph is preserved. Moreover, although raw observations need not be shared, the procedures require numerical covariance or rank-correlation summaries and cannot operate on estimated graphs alone. Their developed formulations also focus on Gaussian or Gaussian copula graphical models, leaving integration for broader distributional families unaddressed.

A final class of methods reconstructs global undirected graphs directly from supplied marginal graph structures \citep{kim2009learning,kim2026sequential}.  These methods treat the supplied marginal structures as exact, without accounting for estimation errors that may render them mutually incompatible. Moreover, they do not address the integration of raw data.

The above limitations clearly highlight the need for a general framework that jointly addresses heterogeneous variable sets, source-specific distributions, and different forms of shared information. 

\subsection{Our contributions}\label{sec:contribs}
To address the limitations of existing approaches, this paper develops a general graphical model estimation algorithm by integrating evidence from multiple heterogeneous sources that may differ simultaneously in the variables they observe, the distributions they follow, and the form of information they can share. 

At a high level, our procedure connects the sources through the \emph{latent projection operator}, which maps a candidate global graph to a local graph on each source's observed variable subset and accounts for connections induced by paths through unobserved variables. We estimate the global graph by jointly assessing how well these source-specific projections fit the available data or graph estimates.
The proposed algorithm offers substantial flexibility and, to the best of our knowledge, applies to settings beyond the scope of any single existing procedure by simultaneously addressing the following four challenges. (a) \emph{Heterogeneous variable sets}: we allow each source to observe a different, possibly only partially overlapping, subset of the full variable set, without requiring any source to observe the full variable set. 
(b) \emph{Distribution shift across sources}: we require only that the underlying conditional independence graph is shared across sources, while allowing the full joint distribution to vary; in the Gaussian graphical model setting, for example, this allows each source to have its own distinct precision matrix, as long as all precision matrices share the same support. 
(c) \emph{Flexible forms of information sharing}: our framework can integrate either shared raw data, locally estimated (source-specific) graphs, or a mixture of the two, which makes it directly applicable under privacy or data-access constraints where only summary information, rather than individual-level data, can be transmitted or pooled. 
(d) \emph{Generality across model classes}: the proposed algorithm applies to a broad class of undirected PGMs, rather than being restricted to specific distributions, such as Gaussian. 

To estimate the global graph, we develop a mixed-integer optimization formulation and a scalable coordinate-descent heuristic for the resulting estimation problem. Because fragmented observations need not determine a unique global graph, we also characterize projection-equivalence classes and identifiable graph features. 

Our theoretical guarantees show that recovery is governed by new combinatorial notions of signal strength that account for how a local structural error can force additional errors across sources through latent-projection compatibility. These combinatorial notions, in particular, reveal how strong local signals in some sources can compensate for weak local signals in others, yielding recovery guarantees that differ fundamentally from standard high-dimensional beta-min requirements.

Extensive simulations show that integration can substantially improve recovery of source-specific graphs relative to separate-source estimation. Integrating raw data generally yields greater gains than integrating estimated graphs. Comparisons with graph quilting further demonstrate the importance of allowing source-specific distributions: our method remains stable as the source distributions diverge, whereas methods relying on a common distribution deteriorate.

We illustrate the proposed framework by integrating four partially-overlapping fatty-acid panels from the National Health and Nutrition Examination Survey (NHANES), covering 31 distinct fatty acids. The analysis combines individual-level concentration data from three sources with a graph estimated separately from compositional measurements in the fourth. The resulting network exhibits organization by fatty-acid family and chain length, while equivalence-class and subsampling analyses characterize structural ambiguity and assess the stability of the estimated relationships.

\subsection{Notations}
For an undirected graph $\G$ over vertex set $V$, we denote $E(\G) = \{(i,j) \in V \times V: i \text{ adjacent to }j \text{ in }\G\}$ as its edge set. For a node $i$ in an undirected graph $\widetilde{\G}$, we use $\mathrm{N}(i,\widetilde{\G})$ to denote nodes that are adjacent to $i$ in $\widetilde{\G}$, i.e., $\mathrm{N}(i,\widetilde{\G}):=\{j: (i,j)\in E(\widetilde{\G})\}$. For a positive integer $n$, $[n]$ denote the set $\{1, 2, \dots, n\}$. 

\section{Setup and Identifiability}
\label{sec:setup}
\subsection{Setup} We consider \(K\) data sources, where source \(k\) is associated with a distribution \(P^{(k)}\) on the random vector \(X_{V}\) with $V = \{1,\ldots,p\}$. In each source, we observe data from only a subset of the variables $S_k \subseteq V$ with $V = \cup_{k\in[K]}S_k$. The subsets can be overlapping.

Let $\perp$ denote a specified notion of conditional independence. Depending on the setting, this relation may represent probabilistic conditional independence, zero partial correlation, conditional tail independence in the context of multivariate extremes, or another independence notion of interest.

We assume that the full source distributions share a common dependence structure: for every pair of distinct vertices $i,j\in V$,
$X_i\perp X_j
\mid X_{V\setminus\{i,j\}} \text{ under }P^{(k)}$ either holds for every $k\in[K]$ or fails for every $k\in[K]$. Thus, although the distributions $P^{(k)}$ may differ across sources, they induce the same pairwise conditional independence structure over the full collection of variables. The following definition formalizes the graph structure shared across the data sources.
\begin{defn}[Population global graph]We define the population global graph as $\G^\star = (V,E^\star)$, where:
$$
(i,j)\notin E^\star
\quad\Longleftrightarrow\quad
X_i\perp X_j
\mid X_{V\setminus\{i,j\}}
\text{ under }P^{(k)} ~\forall k\in[K].
$$
\end{defn}
Here, $\G^\star$ represents the global population graph encoding the pairwise conditional independence structure, shared by all source distributions. 

Our inferential goals are twofold. The first is to identify $\G^\star$, or an equivalence class when it is not identifiable, using only information from each local source. The second is to ``pool" the information across the sources to improve the recovery of a given source's local conditional independence graph.   
For both of these goals, we must relate the conditional independence structure observed in each source to the common global graph $\G^\star$. Because marginalizing variables can induce additional dependencies among the observed variables, the relevant local graph is not generally the subgraph of $\G^\star$ induced by $S_k$, but rather its latent projection onto $S_k$. We first introduce the notions of separation and latent projection needed to formalize this relationship.

\begin{defn}[Separation]
Let $\G=(V,E)$ be an undirected graph. For distinct vertices $i,j\in V$, a set
$T\subseteq V\setminus\{i,j\}$ is said to separate $i$ and $j$ in $\G$ if every path between $i$ and $j$ contains at least one vertex in $T$.
\end{defn}

In source $k$, only the variables indexed by $S_k\subseteq V$ are observed. The observed data therefore follow the marginal distribution $P_{S_k}^{(k)}$, obtained from $P^{(k)}$ by marginalizing over $X_{V\setminus S_k}$. Marginalizing unobserved variables can induce additional edges among the observed variables, formalized according to a latent projection, which is described next.

\begin{defn}[Latent projection]
Let $\G=(V,E)$ be an undirected graph and let $S\subseteq V$. The latent projection of $\G$ onto $S$ is the graph
$$
\pi_S(\G)=(S,E_S),
$$
where, for every pair of distinct vertices $i,j\in S$,
$$
(i,j)\in E_S
\quad\Longleftrightarrow\quad
S\setminus\{i,j\}
\text{ does not separate $i$ and $j$ in $\G$}.
$$
Equivalently, $i$ and $j$ are adjacent in $\pi_S(\G)$ if and only if there exists a path between them in $\G$ whose internal vertices all belong to $V\setminus S$.
\end{defn}

The latent projection differs crucially from the subgraph induced by $S$. In particular, two variables may be adjacent in $\pi_S(\G)$ even when they are not adjacent in $\G$, because a path through variables outside $S$ induces an edge after those variables are marginalized out. Thus, unlike an induced subgraph, the latent projection depends on the structure of $\G$ outside $S$. The distinction between the latent projection and the induced subgraph is especially important, as the conditional-independence relations in $S_k$, induced by $P_{S_k}^{(k)}$, are not, in general, represented by the induced subgraph.

\begin{defn}[Source-wise conditional dependence graph]For each source $k \in [K]$, define its conditional dependence graph $\G^{\star(k)}_{S_k}= (S_k,E_k)$ as 
$$
(i,j)\notin E_k
\quad\Longleftrightarrow\quad
X_i\perp X_j
\mid X_{S_k\setminus\{i,j\}}
\text{ under }P_{S_k}^{(k)} ~(\text{or equivalently under }P^{(k)}).
$$
\label{defn:source-wise}
\end{defn}

To ensure that these local graphs are compatible with $\G^\star$, we impose the following condition.

\begin{assumption}[Source-wise pairwise faithfulness]
For every $k\in[K]$ and every pair of distinct vertices $i,j\in S_k$,
$$
X_i\perp X_j
\mid X_{S_k\setminus\{i,j\}}
\text{ under }P_{S_k}^{(k)} ~(\text{or equivalently under }P^{(k)})
~\Longleftrightarrow~ S_k\setminus\{i,j\} \text{ separates } i \text{ and }j \text{ in }\G^\star.
$$
\label{ass:markov_faithfulness}
\end{assumption}

By the definition of latent projection, Assumption~\ref{ass:markov_faithfulness} implies
$$
\G^{\star(k)}_{S_k}=\pi_{S_k}(\G^\star),
\qquad k\in[K].
$$
Thus, the source-specific conditional dependence graphs are compatible in the sense that they are all latent projections of the same global graph $\G^\star$. Violation of Assumption~\ref{ass:markov_faithfulness} results in a lack of compatibility. The following simple example shows one  settings where the compatibility may fail.
\begin{example}[When Assumption~\ref{ass:markov_faithfulness} is violated] As a simple illustration, suppose $K=2$ and each source follows the same Gaussian distribution:
$$P^{(k)} \sim \mathcal{N}(0,\Sigma),\quad \Sigma = \begin{pmatrix}1 & 0.4 & 0 &0 \\ 0.4 & 1& 0.4 & 0.4 \\0 & 0.4 & 1 & 0 \\ 0 & 0.4 & 0 & 1\end{pmatrix}.$$
Since the precision matrix ${\Sigma^\star}^{-1}$ is a fully dense matrix, $\G^\star$ is a fully connected graph. Suppose $S_1 = \{1,2,3\}$ and $S_2 = \{1,3,4\}$. Then, from the precision matrices corresponding to submatrices of $\Sigma$, we have that $\G^{(1)}$ is a fully connected graph, and $\G^{(2)}$ is the fully disconnected graph. We see that $\pi_{S_2}(\G^\star) \neq \G^{(2)}$. In fact, there does not exist a global graph $\G$ over ${V}$ such that $\pi_{S_k}(\G) = \G^{\star(k)}_{S_k},k=1,2$. 
\end{example}
Assumption~\ref{ass:markov_faithfulness} is weaker than the standard faithfulness assumption, which is described next.
\newtheorem{assumptionprime}{Assumption}
\renewcommand{\theassumptionprime}
  {\theassumption\ensuremath{^{\prime}}}
\begin{assumptionprime}[Source-wise faithfulness] For every $k\in[K]$ and every collection of pairwise disjoint sets $A,B,C\subseteq V$,
$$X_A\perp X_B\mid X_C \text{ under }P^{(k)}
\quad\Longleftrightarrow\quad
C\text{ separates $A$ and $B$ in }\G^\star.$$
\label{ass:standard_faithfulness}
\end{assumptionprime}

Assumption~\ref{ass:standard_faithfulness} is a typical condition in graphical modeling \citep{lauritzen1996graphical}. This stronger condition implies Assumption~\ref{ass:markov_faithfulness} by taking $A={i},B={j},C=S_k\setminus\{i,j\}$. The weaker assumption is sufficient when the local information consists only of pairwise conditional independences given all remaining variables observed in each source.

In several widely studied model classes, faithfulness is generic. For example, for a fixed undirected graph, almost every regular Gaussian distribution that is Markov with respect to the graph is also faithful to it, with the unfaithful parameter values forming a Lebesgue measure zero set \citep{sadeghi2017faithfulness,amini2021perfectness}. Intuitively, violations of faithfulness can arise from special configurations of model parameters {or structures}
that induce conditional independences not implied by graphical separation. Our Assumption~\ref{ass:markov_faithfulness} is weaker than the standard faithfulness in Assumption~\ref{ass:standard_faithfulness}, as it requires only a subset of the implications imposed by standard faithfulness. Nevertheless, the measure-zero result concerns exact violations and does not rule out distributions that are nearly unfaithful, for which some nonzero conditional dependences may be arbitrarily small and graph recovery may be difficult in finite samples.

In several distributional families, both the structure of the graph $\G^\star$ and Assumption~\ref{ass:markov_faithfulness} admit algebraic characterizations. We present these characterizations next.
\begin{example}[Gaussian graphical models] Suppose $P^{(k)}$ is a Gaussian distribution with precision matrix $\Omega^{\star{(k)}}$. It is well-known that $X_i \perp X_j \mid X_{V\setminus\{i,j\}}$ if and only if 
$\Omega^{\star(k)}_{ij}=0$, so that edges in the  population global graph $\G^\star$ are encoded by the sparsity pattern of $\Omega^{\star{(k)}}$. Here, $\perp$ represents probabilistic conditional independence. Assumption~\ref{ass:markov_faithfulness} can then be rephrased algebraically as follows: for every $k \in [K]$ and $i,j \in S_k$,  $\Gamma^{\star{(k)}}_{ij} = 0$ if and only if $S_k\setminus\{i,j\} \text{ separates }i \text{ and }j \text{ in }\G^\star$, where $\Gamma^{\star{(k)}}$ is the Schur complement $\Gamma^{\star{(k)}} = \Omega^{\star(k)}_{S_k,S_k} -\Omega^{\star(k)}_{S_k,T_k}(\Omega^{\star(k)}_{T_k,T_k})^{-1}\Omega^{\star(k)}_{T_k,S_k}$ with $T_k:= V \setminus S_k$.  
\label{example:Gaussian}
\end{example}

\begin{example}[Pearson partial correlation graphical models]Suppose \(P^{(k)}\) has finite second moments and a positive-definite covariance matrix \(\Sigma^{\star(k)}\), and let \(\Omega^{\star(k)}=(\Sigma^{\star(k)})^{-1}\). It follows that $X_i \perp X_j|X_{V\setminus\{i,j\}}$ if and only if $\Omega_{ij}^{\star(k)}=0$, so that edges in the population graph $\G^\star$ are encoded by sparsity patten of $\Omega^{\star(k)}$. Here, $\perp$ denotes zero Pearson partial correlation. Assumption~\ref{ass:markov_faithfulness} can be rephrased in a similar algebraic fashion as in Example~\ref{example:Gaussian}.
\label{example:pearson}
\end{example}

\begin{example}[Gaussian copula graphical models]

Suppose $P^{(k)}$ follows a Gaussian copula model: there exist strictly increasing functions $\{f_j^{(k)}:j\in V\}$ such that $Z:=(f_j^{(k)}(X_j),j\in V)\sim N(0,R^{\star(k)})$. Let $\Omega^{\star(k)}=(R^{\star(k)})^{-1}$. Since coordinatewise invertible transformations preserve probabilistic conditional independence, $X_i\perp X_j\mid X_{V\setminus\{i,j\}}$ if and only if $\Omega^{\star(k)}_{ij}=0$, so that the edges of the global graph $\G^\star$ are encoded by the sparsity pattern of $\Omega^{\star(k)}$. Here, $\perp$ represents probabilistic conditional independence. Moreover, Assumption~\ref{ass:markov_faithfulness} admits the same algebraic characterization as in Example~\ref{example:Gaussian}, with the corresponding marginal precision matrix given by the Schur complement of $\Omega^{\star(k)}$ over $T_k:=V\setminus S_k$.
\label{example:copula}
\end{example}

\begin{example}[Elliptical graphical models]
Suppose $P^{(k)}$ follows an elliptical distribution with normalized scatter matrix $R^{\star(k)}$, and let $\Omega^{\star(k)}=(R^{\star(k)})^{-1}$. In this setting, $X_i\perp X_j\mid X_{V\setminus\{i,j\}}$ if and only if $\Omega_{ij}^{\star(k)}=0$, so edges in the global graph $\G^\star$ are encoded by the sparsity pattern of $\Omega^{\star(k)}$. Here, $\perp$ represents zero generalized partial correlation, which is equivalent to zero Pearson partial correlation under finite second moments. Assumption~\ref{ass:markov_faithfulness} can then be expressed algebraically in the same manner as in Example~\ref{example:Gaussian}. 
\label{example:elliptical}
\end{example}

\begin{example}[Extremal graphical models]  Suppose $P^{(k)}$ in the domain of attraction of a H\"usler-Reiss distribution with a conditionally negative-definite variogram matrix $\Gamma^{\star(k)} \in \mathbb{R}^{p \times p}$ and a precision matrix $\Omega^{\star(k)}:= (\Pi(-\Gamma^{\star(k)}/2)\Pi)^+$, where $A^+$ is the Moore--Penrose pseudoinverse of a matrix $A$, and $\Pi$ is the orthogonal projection onto the complement of the all-ones vector \citep{Hentschel2022StatisticalIF,Engelke2018GraphicalMF}. In this setting, $X_i\perp X_j\mid X_{V\setminus\{i,j\}}$ if and only if $\Omega_{ij}^{\star(k)}=0$, so edges in the global graph $\G^\star$ are encoded by the sparsity pattern of $\Omega^{\star(k)}$\citep{Engelke2018GraphicalMF}. Here, $\perp$ represents extremal conditional independence in the limiting H\"usler-Reiss distribution. Appealing to \cite{Engelke2024ExtremalGM}, Assumption~\ref{ass:markov_faithfulness} can be rephrased in a similar algebraic fashion as in Example~\ref{example:Gaussian}.
\label{example:extremal}
\end{example}

\begin{example}[Discrete graphical models]
Suppose each $X_i$ takes values in a finite discrete state space $\mathcal{X}_i$, and $P^{(k)}$ is a strictly positive distribution with probability mass function $p^{(k)}$. Unlike the preceding examples, conditional independence in discrete graphical models is not, in general, characterized by the sparsity pattern of a precision-like matrix. Instead, letting $R = S_k \setminus\{i,j\}$, $X_i \perp X_j | X_{R}$ if and only if $p^{(k)}_{S_k}(x_i,x_j,x_R)p^{(k)}_{S_k}(x_i',x_j',x_R) = p^{(k)}_{S_k}(x_i,x_j',x_R)p^{(k)}_{S_k}(x_i',x_j,x_R)$ for all $x_i,x_i',x_j,x_j',x_R$, where $p^{(k)}_{S_k}$ is the marginalization of $p^{(k)}$ onto $S_k$. Here, $\perp$ denotes probabilistic conditional independence. Equivalently, conditional independence holds if and only if for all $2\times{2}$ minors of the corresponding conditional probability tables vanish. Assumption~\ref{ass:markov_faithfulness} can therefore be expressed algebraically by requiring, for every $k\in[K]$ and $i,j \in S_k$, these minors vanish if and only if $R$ separates $i$ and $j$ in $\G^\star$. 

\label{example:discrete}
\end{example}

\subsection{Identifiability} Given access to the conditional independence relations encoded by \(P^{(k)}_{S_k}\) for \(k\in[K]\), can we recover the global graph \(\G^\star\)? If \(\G^\star\) is not uniquely identifiable -- particularly when the observed variable sets \(S_k\) have insufficient overlap -- what is the corresponding equivalence class of graphs?

Under Assumption~\ref{ass:markov_faithfulness}, we can formalize identifiability purely in graphical terms. Our formalism requires the following definition.

\begin{defn}[Equivalence class] Take any graph $\G$ over vertex set $V$ and subsets $S_k \subseteq V, k \in [K]$. Then, we define the class of graphs equivalent to $\G$ as:
\begin{equation}\mathrm{EC}(\G;\{S_k\}_{k=1}^K) = \{ \text{Undirected graph }\G' \text{ over vertex set }V: \pi_{S_k}(\G') = \pi_{S_k}(\G)~\forall~k\in[K] \}.
\end{equation}
\label{defn:equiv}
\end{defn}
The equivalence class $\mathrm{EC}(\G;\{S_k\}_{k=1}^K)$ consists of all global graphs $\G'$ whose latent projected graphs onto $S_k$ coincide with those of $\G$ for all $k \in [K]$. With Definition~\ref{defn:equiv}, we can formally characterize the equivalence class compatible with the observed local graph structures. The proof is in Appendix~\ref{sec:proof_identify}.

 \begin{figure}[t]
        \centering
        \includegraphics[width=0.9\linewidth]{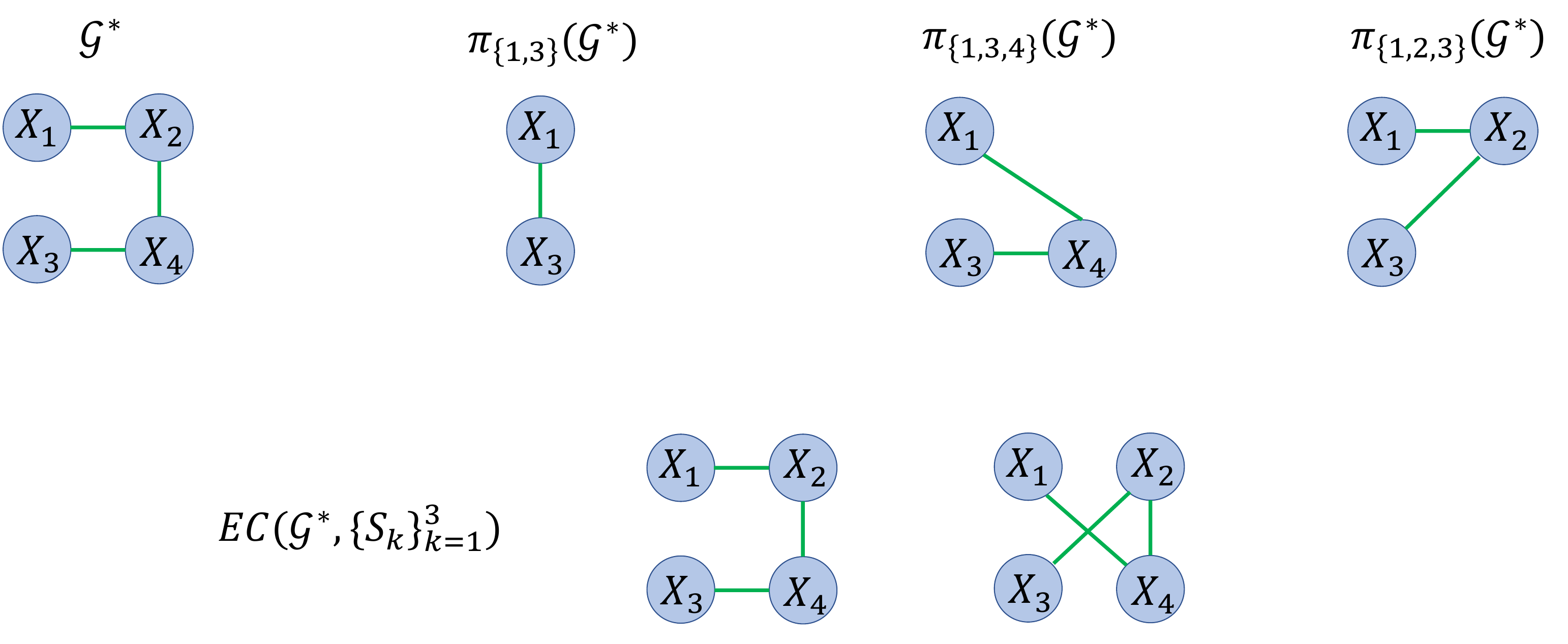}
        \caption{\small \textbf{Top row}. The true graph $\G^\star$ and three latent-projection graphs of $\G^\star$ obtained from three sources with $S_1 = \{1,3\}$, $S_2 = \{1,3,4\}$, and $S_3 = \{1,2,3\}$. \textbf{Bottom row}. The only two graphs in the equivalence class, one being the true graph, whose latent projections match precisely the ones in the top row. The edge between $X_2$ and $X_4$ is identified even though they are never co-observed in the three sources.}
        \label{fig:equivalence_class_illustration}
    \end{figure}
\begin{proposition}[Identifiability] Under Assumption~\ref{ass:markov_faithfulness}, given access to pairwise conditional independences implied by \(P^{(k)}_{S_k}, k\in[K]\), the global graph $\G^\star$ is identifiable up to the equivalence class $\mathrm{EC}(\G^\star;\{S_k\}_{k=1}^K)$.
\label{prop:identifiability}
\end{proposition}

Depending on the degree of overlap in the sets $S_k,k\in[K]$, the size of the equivalence class $\mathrm{EC}(\G^\star;\{S_k\}_{k=1}^K)$ can vary. Moreover, graphs in the equivalence class need not have the same number of edges. 

When the full marginal distributions $P^{(k)}_{S_k}$ are available, rather than only their pairwise conditional independence relations, additional constraints may allow us to rule out certain graphs within the projection-equivalence class. In Gaussian graphical models, covariance rank constraints have been characterized for general undirected graphs \citep{sullivant2010trek}, while more complete algebraic characterizations and associated statistical tests have been developed for Gaussian latent-tree models \citep{shiers2016correlation,leung2018algebraic}. Such constraints could be assessed separately for each source and combined to further refine the equivalence class. Developing a computationally efficient and statistically valid procedure for such refinement, particularly across the general graph structures and distributional families considered here, is beyond the scope of this work and is left for future research.

Even if the graph $\G^\star$ is not fully identifiable, some of its edges may be. We next formalize edge identifiability, under which the presence or absence of an edge is the same across all members of the equivalence class.
\begin{defn}[Edge identifiability] An edge $(i,j) \in V \times V$ is identifiable {with respect to $\{S_k\}_{k=1}^K$} if it is always present or always absent in graphs in the equivalence class of $\G^\star$; that is, either $i$ and $j$ are adjacent in every graph $\G \in \mathrm{EC}(\G^\star;\{S_k\}_{k=1}^K)$ or they are nonadjacent in every graph $\G \in \mathrm{EC}(\G^\star;\{S_k\}_{k=1}^K)$.
\label{defn:edge_identify}
\end{defn}

Since the latent projection operator depends globally on the graph structure, the absence or presence of an edge between a pair of nodes may be identified even if the pair is not co-observed in any source. This phenomenon is highlighted in the following example. 
\begin{example}
  Consider the true global graph $\G^\star$ among four variables $\{X_1,X_2,X_3,X_4\}$ in the top row of Figure~\ref{fig:equivalence_class_illustration}. Suppose we have three sources with $S_1 = \{1,3\}$, $S_2 = \{1,3,4\}$, and $S_3 = \{1,2,3\}$. The latent projected (local) graphs of $\G^\star$ in each of the three sources are displayed in top row of the figure. Under Assumption~\ref{ass:markov_faithfulness}, each graph encodes exactly the conditional independence relationships in the associated marginal distribution. The bottom row of Figure~\ref{fig:equivalence_class_illustration} shows the two graphs in $\mathrm{EC}(\G^\star;\{S_k\}_{k=1}^K)$, highlighting that $\G^\star$ cannot be distinguished from the other global graph using the marginal graphs. Nonetheless, the absence or presence of certain edges is identifiable; in particular, the edge between \(X_1\) and \(X_3\) is absent, whereas the edge between \(X_2\) and \(X_4\) is present. Interestingly, the edge between $X_2$ and $X_4$ is identified even though they are never co-observed in the three sources.  

\label{exam:illus}
\end{example}
\subsection{Summarizing graph equivalence classes without enumeration} For a given graph \(\G\), the equivalence class \(\mathrm{EC}(\G;\{S_k\}_{k=1}^K)\) generally does not admit a compact representation. Although Section~\ref{sec:exact_solvers} provides an integer programming procedure for enumerating its members, the equivalence class may contain many graphs, making exhaustive enumeration computationally costly and the resulting collection difficult to interpret. This motivates seeking informative summaries that can be computed without explicitly enumerating the entire equivalence class. 

Here we list few examples of useful summaries:
\emph{edge identifiability} (see Definition~\ref{defn:edge_identify}): for each pair of nodes \((i,j)\), we can ask whether the presence or absence of an edge is identified; \emph{logical edge implications}: whether the presence of an edge between $(i,j)$ implies the presence of an edge between $(i',j')$, or whether the two edges are always jointly present or absent; \emph{separation statements}: given disjoint variable sets \(A,B,C\subseteq V\), whether in every graph in the equivalence class, $C$ separates $A$ and $B$. Under faithfulness (e.g., Assumption~\ref{ass:standard_faithfulness}), this property would mean that all graphs specify the conditional independence relationship $X_A \perp X_B \mid X_C$; and \emph{bounds on graph attributes}: the minimum or maximum number of edges among graphs in the equivalence class, or the minimum or maximum degree of a given node.

All the above summary statistics can be obtained by solving the following optimization problem, without explicitly enumerating the equivalence class of graphs:

\begin{equation}
\min_{\substack{\text{undirected graph } {\G}'  \text{ over } V}}~~~~f(\G')\quad\text{subject-to:}\quad \G'\in \mathrm{EC}(\G;\{S_k\}_{k=1}^K).
\label{eq:opt_equiv}
\end{equation}
For edge identifiability, we can set \(f(\G')=\mathbb{I}[(i,j)\in E(\G')]\). An optimal value of one implies that \(i\) and \(j\) are adjacent in every graph in the equivalence class; otherwise, there exists at least one graph in the equivalence class in which they are nonadjacent. Taking $f(\G') = \mathbb{I}[(i,j)\not\in E(\G')]$ can similarly verify whether \(i\) and \(j\) are nonadjacent in every graph in the equivalence class. For logical edge implications, let \(f(\G') = \mathbb{I}[(i',j') \in E(\G')]- \mathbb{I}[(i,j) \in E(\G')]\). A non-negative optimal value means that if an edge between $(i,j)$ appears in a graph, the edge between $(i',j')$ also appears. For conditional independence relationships, we set \(f(\G')=\mathbb{I}[A \text{ and } B\text{ are separated in }\G' \text{ by }C]\). Under faithfulness in Assumption~\ref{ass:standard_faithfulness}, an optimal value of one means that all the graphs agree on the conditional independence relation $A\perp{B}\mid{C}$. Finally, for bounds on graph attributes, let \(f(\G')= \# \text{ edges in }\G'\). The optimal value gives a lower bound on the number of edges, and an upper bound can be obtained by considering the negative of $f(\G')$. A lower and upper bound on the degree can be similarly obtained by a suitable function $f(\G')$.  

In all the settings described above, we show in Section \ref{sec:summary_integer} that the optimization problem \eqref{eq:opt_equiv} can be cast as a linear mixed-integer program (MIP). Crucial to this characterization is that the constraint  $\G'\in \mathrm{EC}(\G;\{S_k\}_{k=1}^K)$ can be expressed as linear equalities/inequalities over a collection of binary variables, without explicitly enumerating the equivalence class. Such a formulation enables advanced branch and bound and cutting plane algorithms to be deployed to solve the optimization problem relatively efficiently.   

\section{An estimator for reconstructing a global graph}

\subsection{Preliminaries}

We suppose that each source provides one of two types of local information. For sources in $\mathcal{W}_1\subseteq[K]$, we have access to the local data: source $k\in\mathcal{W}_1$ provides observations drawn from the marginal distribution $P^{(k)}_{S_k}$. For sources in $\mathcal{W}_2=[K]\setminus\mathcal{W}_1$, we instead observe only an estimated local graph. This may arise in federated learning settings, where privacy or data-governance constraints prevent the sharing of raw data and only a learned graph can be released. Another example is evidence synthesis across published studies, where the original data are unavailable but estimated graphical structures are reported, potentially for different subsets of variables. Finally, it may also arise when fitting a graph to a particular source requires a specialized procedure with a highly nonconvex loss that is not amenable to convex mixed-integer programming (e.g., real-data analysis in Section~\ref{sec:realdata}). In such cases, the source can first estimate its local graph using an appropriate method and then pass the resulting graph to the integrative procedure. Our goal is to combine these two forms of local information to estimate a global graph.

For each $k\in\mathcal{W}_2$, let $\widehat{\G}^{(k)}_{S_k}$ denote the graph estimate on vertex set $S_k$ supplied by source $k$. We assume this estimate is obtained from $n_k$ independent and identically distributed observations drawn from the marginal distribution $P^{(k)}_{S_k}$. For each $k \in \mathcal{W}_1$, let $\mathcal{D}_k$ denote the corresponding local dataset of size $n_k:=|\mathcal{D}_k|$.
To incorporate these local datasets into an estimate of the global graph, we require a loss function that evaluates how well a candidate graph $\widetilde{\G}$ on $S_k$ fits $\mathcal{D}_k$, augmented by a sparsity-promoting penalty proportional to the number of free parameters. We denote this loss function by: 

\begin{equation}\mathrm{Loss}_k(\widetilde{\G};\mathcal{D}_k)= \min_{\theta^{(k)}\in\Theta_k({\widetilde{\G}})}\ell_k(\theta^{(k)};\mathcal{D}_k)+  \lambda_k\cdot\mathrm{dim}(\Theta_k({\widetilde{\G}}))
\label{eqn:loss_func}
\end{equation}
where $\ell_k(\theta^{(k)};\mathcal{D}_k)$ measures the lack of fit of the distribution indexed by $\theta^{(k)}$ to the local data, and $\Theta_k({\widetilde{\G}})$ denotes the admissible parameter space associated with the chosen loss function and the graph $\widetilde{\G}$. Thus, \(\mathrm{Loss}_k(\widetilde{\G};\mathcal{D}_k)\) is the minimum achievable loss over $\Theta_k({\widetilde{\G}})$, augmented by a sparsity-promoting penalty proportional to the dimension of $\Theta_k({\widetilde{\G}})$; here, $\mathrm{dim}(\Theta_k({\widetilde{\G}}))$ is the dimension of the parameter space associated with $\widetilde{\G}$. Although the penalty is constant in the optimization over \(\theta^{(k)}\) for a fixed graph, it varies across candidate graphs and discourages overfitting. We discuss strategies for choosing the regularization parameter in Section \ref{sec:simulations}. In Section~\ref{sec:disccusion_consistent}, for Examples~\ref{example:Gaussian}-\ref{example:discrete}, we describe choices for the loss $\ell_k(\cdot)$, the corresponding parameter space $\Theta_k({\widetilde{\G}})$, and its dimension. 

To ensure that our procedure is statistically accurate, both the estimators used to obtain the local graphs $\widehat{\G}_{S_k}^{(k)}$ for $k\in\mathcal{W}_2$ and the graph estimators obtained by minimizing $\mathrm{Loss}_k(\widetilde{\G};\mathcal{D}_k)$ for $k\in\mathcal{W}_1$ must be consistent. 

\begin{assumption}[Consistency]
As $n_k\longrightarrow\infty$,
\begin{eqnarray*}
\begin{aligned}
&\text{ for every }k\in \mathcal{W}_1,\quad \mathbb P_{P^{(k)}}\left(\argmin_{\substack{\text{undirected graph }\widetilde{\G}
\text{ over }S_k}}
\operatorname{Loss}_k(\widetilde{\G};\mathcal D_k)=\{\G^{\star(k)}_{S_k}\}
\right)
\longrightarrow 1 ,\ 
\\
&\text{ for every }k\in \mathcal{W}_2,\quad \mathbb{P}_{P^{(k)}}(\widehat{\G}_{S_k}^{(k)} = \G^{\star(k)}_{S_k}) \to 1.\ 
\end{aligned} 
\end{eqnarray*} 
Here, $\G^{\star(k)}_{S_k}$ is defined in 
Definition~\ref{defn:source-wise}.  
\label{ass:consistent}
\end{assumption}
In Section \ref{sec:disccusion_consistent}, we discuss settings where Assumption~\ref{ass:consistent} is satisfied. 

\subsection{Estimator}To learn a global undirected graph, we propose the following estimator for $\gamma \geq 0$ and $\omega_k\ge0$:

\begin{eqnarray}
\begin{aligned}
\widehat{\G} \in \argmin_{\substack{\text{undirected graph }\G \text{ over } V}}\quad \sum_{k \in \mathcal{W}_1} \omega_k\mathrm{Loss}_k(\pi_{S_k}(\G);\mathcal{D}_k)+ \gamma\sum_{k \in \mathcal{W}_2}\omega_k\mathrm{SHD}(\pi_{S_k}(\G), \widehat{\G}^{(k)}_{S_k}).
\end{aligned}
\label{eqn:main}
\end{eqnarray}
 Here, for two undirected graphs \(\G_1\) and \(\G_2\) on the same vertex set, \(\mathrm{SHD}(\G_1,\G_2)\) denotes their structural Hamming distance, namely, the number of edges on which they differ. For each candidate global graph \(\G\), the objective of \eqref{eqn:main} evaluates how well \(\pi_{S_k}(\G)\) fits the data from each source \(k\in \mathcal{W}_1\) and how closely it agrees with the estimated local graph provided by sources \(k\in\mathcal{W}_2\). The estimator therefore selects a global graph that minimizes the weighted discrepancies across sources. 
 
 The weights $\omega_k$ determine the relative contributions of the sources. One possible choice is $\omega_k=n_k/\sum_{k}n_k$, so that sources with more observations receive greater weight. More generally, the weights may reflect differences in source reliability or account for differences in the scales of the source-specific loss functions. The parameter \(\gamma\) additionally controls the relative contribution of sources providing graph estimates rather than raw data; we describe strategies to choose $\gamma$ in Remark~\ref{remark:choosing_gamma}.

 Notice that since the optimization problem in \ref{eqn:main} depends on $\G$ only through its latent projections $\{{\pi_{S_k}(\G)}\}_{k=1}^K$, any ${\G} \in \mathrm{EC}(\widehat{\G};\{S_k\}_{k=1}^K)$ yields exactly the same objective value. Consequently, the estimator cannot distinguish between graphs within the same projection-equivalence class. The following proposition establishes that, under suitable conditions, $\mathrm{EC}(\widehat{\G};\{S_k\}_{k=1}^K)$ consistently estimates the population equivalence class $\mathrm{EC}({\G}^\star;\{S_k\}_{k=1}^K)$. The proof is in Appendix~\ref{sec:proof_equiv}. Throughout, when a source does not contribute to the objective in eqref~\ref{eqn:main} because $\omega_k=0$ or, for $k\in\mathcal{W}_2$, $\gamma=0$, the equivalence classes are understood to be defined only with respect to the remaining sources.
 
 \begin{proposition}Suppose Assumptions~\ref{ass:markov_faithfulness}-\ref{ass:consistent} hold. Then, we have that:
$$
\mathbb P_{P^{(1)},P^{(2)},\dots,P^{(K)}}\left(
\mathrm{EC}(\widehat{\G},\{S_k\}_{k=1}^K)=\mathrm{EC}(\G^\star,\{S_k\}_{k=1}^K)
\right)
\longrightarrow 1
\text{ as }n_k\longrightarrow\infty~\forall k\in[K].
$$
That is, the estimator \eqref{eqn:main} consistently estimates the population equivalence class of $\G^\star$.
\label{prop:equiv}
\end{proposition}

\begin{remark}
Notice that Proposition~\ref{prop:equiv} assumes the size of the vertex set $V$ does not grow with the sample size. In Section \ref{sec:theory}, we provide high-dimensional consistency guarantees that allow for the the number of vertices to grow with the sample size. 
\end{remark}

\begin{remark}[Choosing $\gamma$]\label{remark:choosing_gamma} The user may specify \(\gamma\) based on domain knowledge, assigning greater weight to one type of local information than to the other. Alternatively, \(\gamma\) may be selected in a data-driven manner using hold-out validation. Specifically, we partition the data into a training set \(\mathcal D_k^{\mathrm{train}}\) and a validation set \(\mathcal D_k^{\mathrm{valid}}\). For each candidate \(\gamma\), we solve \eqref{eqn:main} using the training data to obtain $\widehat{\G}_\gamma$; this also gives the associated continuous model parameters $\widehat{\theta}_{\gamma}^{(k)}$. The corresponding validation score is $\sum_{k \in \mathcal{W}_1}\ell_k(\widehat{\theta}_{\gamma}^{(k)};\mathcal{D}_k^\mathrm{valid})$. We can select the value of $\gamma$ that minimizes this validation score. Another possibility, which we adopt in our experiments, is to calibrate $\gamma$ to the sparsity penalties in the data-source losses. When these penalties are proportional to the number of edges, we set $\gamma$ equal to the median penalty per undirected edge across the data-providing sources. Since one unit of SHD represents one edge disagreement, this places its cost on the scale of a typical data-source edge penalty.
\end{remark}

\subsection{Discussion of Assumption~\ref{ass:consistent}}
\label{sec:disccusion_consistent}
Below, we describe choices of loss functions and parameter spaces for the distributional families considered in Examples~\ref{example:Gaussian}--\ref{example:discrete}, for which Assumption~\ref{ass:consistent} can hold under suitable regularity and tuning conditions. In the following examples, $\widetilde{\G}$ is an undirected graph over $S_k$.

\medskip
\noindent\textbf{Gaussian graphical models in Example~\ref{example:Gaussian}.} Let $\ell_k$ be the negative Gaussian log-likelihood, defined for $\Theta_k({\widetilde{\G}}):= \{\theta \in \mathbb{R}^{|S_k| \times |S_k|}: \theta \succ 0, \theta_{ij}=0 \text{ if }(i,j)\not\in E(\widetilde{\G})\}$:
\begin{equation}\ell_k(\theta^{(k)};\mathcal{D}_k):= -\log\det(\theta^{(k)}) + \mathrm{tr}(\widehat{\Sigma}_k\theta^{(k)}).
\label{eqn:log_det}
\end{equation}
Here, $\widehat{\Sigma}_k \in \mathbb{R}^{|S_k| \times |S_k|}$ is the sample covariance matrix of the data in source $k$. Further, $\mathrm{dim}(\Theta_k(\widetilde{\G}))= |E(\widetilde{\G})|+|S_k|$. 

Another possible loss, defined for $\Theta_k({\widetilde{\G}}):= \{\theta \in \mathbb{R}^{|S_k| \times |S_k|}: \theta_{ii}=0~\forall i,\theta_{ij}=0 \text{ if }(i,j)\not\in E(\widetilde{\G})\}$, is based on pseudo-likelihood approximation \citep{Besag1974SpatialIA} and takes on the quadratic form:
\begin{equation}\ell_k(\theta^{(k)};\mathcal{D}_k):= \mathrm{tr}((I-\theta^{(k)})^\top \widehat{\Sigma}_k(I-\theta^{(k)})). 
\label{eqn:convex_quad}
\end{equation}
Here, $\mathrm{dim}(\Theta_k(\widetilde{\G}))=2|E(\widetilde{\G})|$. In our experiments, we favor the quadratic loss \eqref{eqn:convex_quad} over its non-quadratic counterpart \eqref{eqn:log_det} since it yields a mixed-integer quadratic program, which is amenable to computationally faster algorithms. See Section \ref{MIP-form} for a detailed discussion. Note that the loss \eqref{eqn:convex_quad}, augmented with the $L_1$ penalty on the coefficients $\theta^{(k)}$, is the \emph{neighborhood selection} procedure \citep{Meinshausen2006HighdimensionalGA}. Since our formulation already introduces binary variables to encode the latent projections, an $L_0$ penalty on the coefficients, which is equivalently $\mathrm{dim}(\Theta_k(\widetilde{\G}))$, can be represented as a linear function of these binary variables and therefore does not require any additional integer variables.

\medskip
\noindent\textbf{Pearson correlation graphical models in Example~\ref{example:pearson}.} The loss \eqref{eqn:log_det} or \eqref{eqn:convex_quad} can be used directly, with the same choices for $\Theta_k(\widetilde{\G})$ and its dimension as in the Gaussian example.

\medskip
\noindent\textbf{Gaussian copula and Elliptical graphical models in Examples~\ref{example:copula}-\ref{example:elliptical}.} 
The losses \eqref{eqn:log_det} or \eqref{eqn:convex_quad} can again be used after replacing $\widehat{\Sigma}_k$ with the rank-based estimator $\widehat{R}$, obtained from the sample Kendall correlations via $\widehat{R}_{ij}=\sin(\pi\widehat{\tau}_{ij}/2)$. Further, $\Theta_k(\widetilde{\G})$ and its dimension are the same as in the Gaussian example.

\medskip
\noindent\textbf{Extremal graphical models in Example~\ref{example:extremal}.} Let $\ell_k$ be the following loss, defined for $\Theta_k({\widetilde{\G}}):=\{\theta\in\mathbb{R}^{|S_k| \times |S_k|}: \theta\mathbf{1}=0, \theta \succeq 0, \theta_{ij}=0 \text{ if }(i,j)\not\in E(\widetilde{\G})\}$ with $\mathbf{1}$ denoting all-ones vector \citep{Hentschel2022StatisticalIF,Engelke2024ExtremalGM}:
\begin{equation*}\ell_k(\theta^{(k)};\mathcal{D}_k):= -\log\det(U^T\theta^{(k)}{U}) + \mathrm{tr}(-\widehat{\Gamma}/2\theta^{(k)}). 
\label{eqn:extremal}
\end{equation*}
Here, $U \in \mathbb{R}^{|S_k| \times |S_k|-1}$ is an orthonormal matrix such that $UU^\top$ is the projection matrix onto the subspace orthogonal to the all-ones vector. Further, $\widehat{\Gamma}$ is an empirical estimate of the true variogram matrix, and $\mathrm{dim}(\Theta_k(\widetilde{\G})) = |E(\widetilde{\G})|$. 

\medskip
\noindent\textbf{Discrete graphical models in Example~\ref{example:discrete}.} Let $\widehat{P}^{(k)}_{S_k}$ denote the empirical distribution of the data in source $k$, over the finite discrete space $\mathcal{X}:=\prod_{i\in S_k}\mathcal{X}_i$. We parameterize distributions Markov with respect to $\widetilde{\G}$ as $p_{\theta^{(k)}}(x) = \exp\{\sum_{C\in\mathcal{C}(\widetilde{\G})} \theta^{(k)}_C(x_C)-A(\theta^{(k)})\}$, where $\mathcal{C}(\widetilde{\G})$ is the collection of nonempty cliques in $\widetilde{\G}$, $\theta^{(k)}$ is a collection of functions $\theta^{(k)}_C:\mathcal{X}_C \to \mathbb{R}$, $A(\theta^{(k)})$ is a normalization term, and $\Theta_k(\widetilde{\G})$ denotes the corresponding parameter space. The negative log-likelihood under an undirected graph $\widetilde{\G}$ can then be written as
\begin{equation}
\ell_k(\theta^{(k)};\mathcal{D}_k) = -\sum_{x\in\mathcal{X}}\widehat{P}^{(k)}_{S_k}(x)\log{p}_{\theta^{(k)}}(x).
\label{eqn:discrete_loss_1}
\end{equation}
Since for any $C \in \mathcal{C}(\widetilde{\G})$, $\mathrm{dim}(\theta_c) = \prod_{i \in C}\left\{|\mathcal{X}_i|-1\right\}$, we have that $\mathrm{dim}(\Theta_k(\widetilde{\G})) = \sum_{C \in \mathcal{C}(\widetilde{\G})}\prod_{i \in C}\left\{|\mathcal{X}_i|-1\right\}$.

Alternatively, we parameterize the conditional distribution of each variable $X_i$ given its neighbors $X_{\mathrm{N}(i,{\widetilde{\G}})}$ as $p_{\theta_i^{(k)}}(x_i\mid x_{\mathrm{N}(i,\widetilde{\G})})=\exp\{\theta_i^{(k)}(x_i,x_{\mathrm{N}(i,\widetilde{\G})})-A_i(\theta_i^{(k)})\}$, where $\theta_i^{(k)}:\mathcal{X}_i\times\mathcal{X}_{\mathrm{N}(i,\widetilde{\G})}\to\mathbb{R}$ and $\theta^{(k)}=\{\theta_i^{(k)}\}_{i\in S_k}$. We let $\Theta_k(\widetilde{\G})$ denote the resulting parameter space. The negative pseudo-log-likelihood under an undirected graph $\widetilde{\G}$ is then
\begin{equation}
\ell_k(\theta^{(k)};\mathcal{D}_k) = -\sum_{i\in S_k}\sum_{x \in \mathcal{X}} \widehat{P}^{(k)}_{S_k}(x)\log{p}_{\theta^{(k)}_i}(x_i|x_{\mathrm{N}(i,\widetilde{\G})}).
\label{eqn:discrete_loss_2}
\end{equation}
Compared to \eqref{eqn:discrete_loss_1}, the loss in \eqref{eqn:discrete_loss_2} avoids the global normalization constant, requiring only nodewise normalization of the conditional distributions. Here, $\mathrm{dim}(\Theta_k(\widetilde{\G})) = \sum_{i\in S_k}(|\mathcal{X}_i|-1)\prod_{j \in \mathrm{N}(i,\widetilde{\G})}|\mathcal{X}_j|$.

\medskip

Under suitable regularity and identifiability conditions, consistent graph selection in the fixed-dimensional regime, where $|S_k|$ does not grow with $n_k$, can be obtained by choosing a penalty parameter that vanishes sufficiently slowly. In particular, when the loss is normalized by sample size and the parameter estimators are $\sqrt{n_k}$-consistent, sufficient conditions are $\lambda_k\to 0$ and $n_k\lambda_k\to\infty$, as satisfied by $\lambda_k=\log(n_k)/n_k$.

Note that other model classes, including pairwise exponential-family graphical models, also admit loss functions that yield consistent graph estimators under suitable regularity conditions \citep{Yang2015GraphicalMV}. We focus on the examples above as (i) their underlying distributional families are closed under marginalization, and (ii) they admit an algebraic characterization of Assumption~\ref{ass:markov_faithfulness} (see Section~\ref{sec:setup}). The first property provides a coherent model for both the full distributions \(P^{(k)}\), which define the global graph \(\G^\star\), and the observed marginal distributions \(P_{S_k}^{(k)}\). Additionally, the second property allows us to better understand Assumption~\ref{ass:markov_faithfulness}.

\section{Exact solvers: casting estimator as a convex mixed-integer program}
\label{sec:exact_solvers}
\subsection{Encoding latent projections via mixed-integer constraints}
\label{sec:encoding_latent}
A candidate global graph $\G=(V,E)$, together with its latent projection onto the observed vertex set $S_k$ in source $k$, is represented by the binary decision variables
$$
g_{ij}\in\{0,1\}\qquad \forall\,1\le i<j\le p,\ \quad  h^{(k)}_{ij}\in\{0,1\}\qquad \forall\, i<j,\ i,j\in S_k.
$$
Here, \(g_{ij}=1\) indicates that \((i,j)\in E\), while the variables \(h^{(k)}_{ij}\) encode the edges of the latent projection graph \(\pi_{S_k}(\G)\). Accordingly, for each \(k\in[K]\) and each pair \(i<j\) with \(i,j\in S_k\), we require \(h_{ij}^{(k)}=1\) if and only if $i$ and $j$ are adjacent in $\pi_{S_k}(\G)$. Equivalently, we enforce the relation

\begin{eqnarray}
h_{ij}^{(k)}=1
~\Longleftrightarrow~
\begin{aligned}[t]
&\text{\(i\) and \(j\) are adjacent in \(\G\), or}\\
&\text{there exists a path between \(i\) and \(j\) in \(\G\)
whose internal vertices lie in \(T_k\),}
\end{aligned}
\label{eqn:projection_constraint}
\end{eqnarray}
where $T_k:= V\setminus S_k$ represents hidden vertices, i.e., vertices that are not observed in source $k$. By introducing additional continuous variables, we will show that relation \eqref{eqn:projection_constraint} can be expressed as a collection of linear constraints involving the binary variables $g_{ij}$ and $h_{ij}^{(k)}$ and the auxiliary continuous variables. This formulation is key to recasting the estimator in \eqref{eqn:main} as a convex MIP, as well as to designing linear MIPs for enumerating graphs in an equivalence class and computing summary statistics of that equivalence class.

For the direction $\Rightarrow$ in \eqref{eqn:projection_constraint}, we use a \emph{network-flow formulation}. Fix an origin $i \in S_k$, which serves as the source of the flow, and let $J_k(i) := \{j\in S_k, j>i\}$ denote the set of possible destinations. Each variable $h_{ij}^{(k)}=1$, with $j\in J_k(i)$, creates one unit of demand at destination $j$. The flow may be sent directly from $i$ to a destination or along a path whose internal vertices lie in $T_k$. Let $A_k(i)$ denote the set of directed arcs available to this flow. It consists of arcs from $i$ to vertices in \(T_k\) and to destinations in \(J_k(i)\), arcs in both directions between distinct vertices of \(T_k\), and arcs from vertices in \(T_k\) to destinations in \(J_k(i)\). For each \((u,v)\in A_k(i)\), let \(f_{uv}^{(k),i}\geq 0\) denote the amount of flow sent along the arc \((u,v)\). A positive flow \(f_{uv}^{(k),i}>0\) requires the corresponding undirected edge \((u,v)\) to be selected in \(\G\), that is, \(g_{uv}=1\), with the indices ordered appropriately. Figure~\ref{fig:aggregated_flow} (left) illustrates this flow network.

\begin{figure}
    \centering
    \begin{minipage}[t]{0.48\linewidth}
    \includegraphics[width=0.7\linewidth]{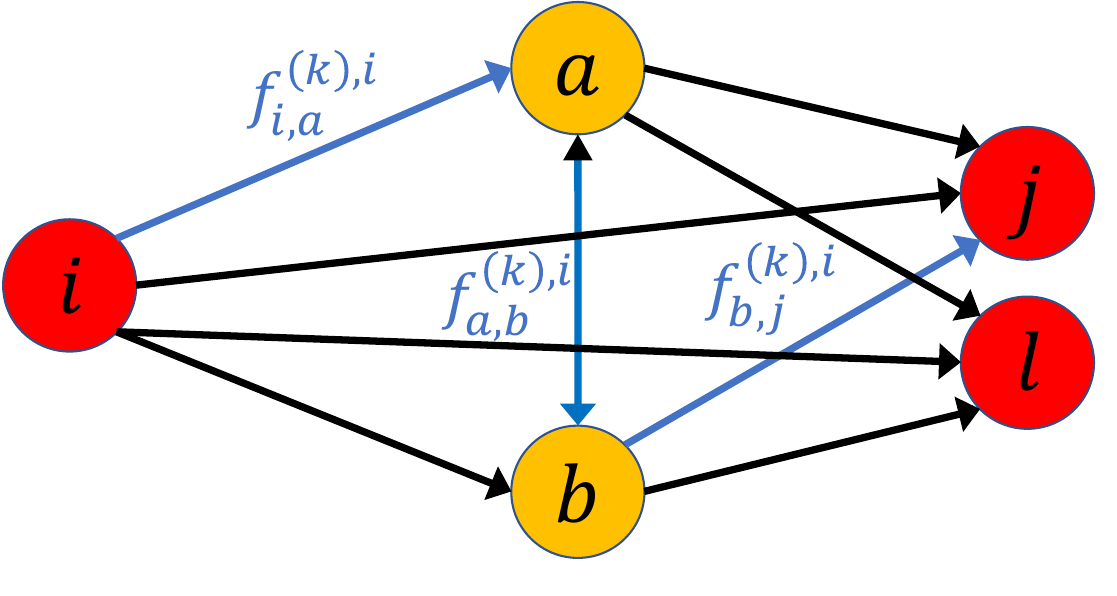}
\end{minipage}
    \begin{minipage}[t]{0.48\linewidth}
\includegraphics[width=0.9\linewidth]{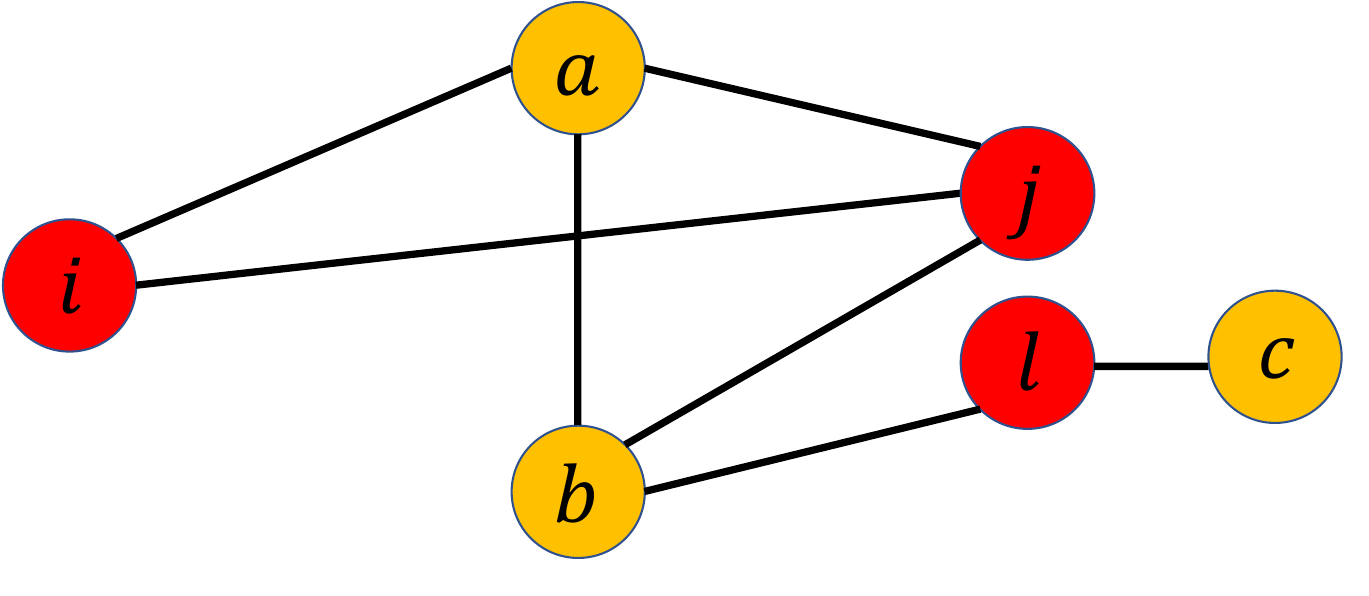}
\end{minipage}
    \caption{\small \textbf{Left}. A flow network with $S_k = \{i,j,l\}$ with origin $i$ and possible destinations $J_k(i) = \{j,l\}$. Here, $T_k= \{a,b\}$ represents the set of hidden vertices. The directed arcs are specified by \(A_k(i)\), which determines the possible directions of flow; arcs between vertices in \(T_k\) are available in both directions. Suppose $h_{ij}^{(k)}=1$. The blue arcs illustrate one possible flow from $i$ to $j$, with each displayed arc carrying a strictly positive amount of flow. \textbf{Right.} A global graph $\G$ over six vertices. One additional vertex is included relative to the left panel to improve the illustration. Here $S_k = \{i,j,l\}$, $T_k = \{a,b,c\}$, and the connected components of the subgraph induced by \(T_k\) are $\{a,b\}$ and $\{c\}$. Because the component \(\{a,b\}\) is attached to \(i\), we have \(m_a^{(k),i},m_b^{(k),i}>0\), whereas \(m_c^{(k),i}=0\). Moreover, since \(m_b^{(k),i}>0\) and \(l\) is adjacent to \(b\), it follows that \(h_{il}^{(k)}=1\). Finally, because \(i\) and \(j\) are directly adjacent, we have \(h_{ij}^{(k)}=1\).}
    
    \label{fig:aggregated_flow}
\end{figure}
Formally, for each $i\in S_k$ and $k\in [K]$, we impose the following linear constraints, connecting the flow variables to $g$ and $h^{(k)}$:
\begin{eqnarray}
\begin{aligned}
\sum_{v:(i,v)\in {A}_k(i)} f^{(k),i}_{iv} = \sum_{j\in J_k(i)} h_{ij}^{(k)},\quad\quad  
\sum_{u\in T_k} f^{(k),i}_{uj}+f^{(k),i}_{ij} = h^{(k)}_{ij}\quad  \forall\, j\in J_k(i),\\
f^{(k),i}_{it}+\sum_{u\in T_k\setminus\{t\}} f^{(k),i}_{ut}=
\sum_{v\in T_k\setminus\{t\}} f^{(k),i}_{tv}+\sum_{j\in J_k(i)} f^{(k),i}_{tj}\quad \forall\, t\in T_k,\\
f_{uv}^{(k),i} \leq [1+(|J_k(i)|-1)\mathbb{I}\{v\in T_k\}]g_{uv},\ \forall (u,v) \in {A}_k(i).
\end{aligned}
\label{eqn:flow_constraints}
\end{eqnarray}

The first constraint requires the origin to emit one unit of flow for each selected destination. The second requires each destination $j$ to receive exactly $h_{ij}^{(k)}$ units, either directly from $i$ or through a hidden vertex. The third imposes flow conservation at every hidden vertex. Finally, the fourth constraint allows flow only along edges selected in the global graph: arcs entering a destination have capacity one, whereas arcs entering a hidden vertex have capacity $|J_k(i)|$. Note that our formulation aggregates all destinations associated with the same origin to avoid a separate flow problem for every pair $(i,j)$ and reduces the number of flow variables for source $k$ from $O(s_k^2t_k^2)$ to $O(s_k^2t_k+s_kt_k^2)$.

To enforce the direction $\Leftarrow$ in \eqref{eqn:projection_constraint}, we introduce an auxillary continuous variable $m^{(k),i}_{u} \in [0,1]$ for each $i \in S_k$ and $u \in T_k$. This variable indicates whether the connected component containing $u$ in the subgraph induced by $T_k$ is attached to $i$, meaning that $i$ is adjacent to at least one vertex in the component. In particular, $m^{(k),i}_{u}>0$ if and only if the connected component containing $u$ is attached to $i$. If a vertex in such a component is also attached to some $j \in J_k(i)$, then \(\G\) contains a path from $i$ to $j$ whose internal vertices lie in $T_k$, and the constraints force \(h_{ij}^{(k)}=1\). The case where $i$ and $j$ are adjacent in $\G$ is handled separately in which $h_{ij}^{(k)}=1$. See Figure~\ref{fig:aggregated_flow}(right) for an illustration.

Formally, for each $i \in S_k$ and $k \in [K]$, we impose the following linear constraints, connecting the auxiliary variables to $g$ and $h^{(k)}$:
\begin{eqnarray}
\begin{aligned}
m_u^{(k),i}-m_v^{(k),i} \leq 1-g_{uv}, \quad m_v^{(k),i}-m_u^{(k),i} \leq 1-g_{uv},\quad m_{u}^{(k),i} \geq g_{iu},\ \forall\, u,v\in T_k,\ u<v, \\
m_{u}^{(k),i} \leq 1-g_{uj}+h_{ij}^{(k)},\quad g_{ij}\leq h_{ij}^{(k)} ,\ \forall\, u \in T_k, j \in J_k(i).
\end{aligned}
\label{eqn:hidden_component}
\end{eqnarray}
If $g_{uv}=1$ for hidden vertices $u,v \in T_k$, the first two inequalities force $m_u^{(k),i}=m_v^{(k),i}$. Thus, the same value is propagated throughout each connected component of the subgraph induced by \(T_k\). If the origin $i$ is adjacent to a hidden vertex $u$, the third constraint forces $m_u^{(k),i}=1$, and this value propagates throughout the hidden component containing $u$. The fourth constraint ensures that if a hidden component attached to $i$ contains a vertex adjacent to $j$, then $h^{(k)}_{ij}=1$. Finally, the last constraint handles the case in which $i$ and $j$ are adjacent. This construction avoids introducing separate reachability variables for every pair $(i,j)$. Instead, the hidden-component variables are shared across all destinations associated with the same origin $i$. As a result, the formulation uses $O(s_kt_k)$ continuous variables and $O(s_kt_k^2+s_k^2t_k)$ constraints, rather than $O(s_k^2t_k)$ variables and $O(s_k^2t_k^2)$ constraints. 

The linear constraints \eqref{eqn:flow_constraints} and \eqref{eqn:hidden_component} provide an exact characterization of a latent projection: if the graph encoded by $h^{(k)}$ is the latent projection onto $S_k$ of the global graph encoded by $g$, then $(g,h^{(k)})$ can be extended to a feasible solution of the constraint set; conversely, any feasible pair $(g,h^{(k)})$ satisfies this latent-projection relation.

\subsection{Convex mixed-integer programming reformulation of estimator}
\label{sec:MIP_reformulation}
Building on the constraints in \eqref{eqn:flow_constraints} and \eqref{eqn:hidden_component} for encoding latent projections, for the graphical models in Examples~\ref{example:Gaussian}-\ref{example:discrete}, we reformulate the estimator \eqref{eqn:main} as a convex MIP. 

\medskip 

\subsubsection{MIP formulation for Examples~\ref{example:Gaussian}-\ref{example:extremal}.}\label{MIP-form} These examples share two common attributes: 
\begin{enumerate}
\item[(i)] The loss $\ell_k(\theta^{(k)};\mathcal{D}_k)$ is parameterized by the matrix parameter $\theta^{(k)} \in \mathbb{R}^{|S_k| \times |S_k|}$ and is convex with respect to $\theta^{(k)}$;
\item[(ii)] The condition $\theta^{(k)} \in \Theta_k({\widetilde{\G}})$ can be expressed as $\theta^{(k)}_{ij} =0 \text{ if }(i,j)\not\in E(\widetilde{\G})$, and $\theta^{(k)} \in \Xi_{k}$, where $\Xi_{k}$ is a convex set that does not depend on $\widetilde{\G}$. 
\end{enumerate}
Recall that, up to a graph-independent additive constant, $\mathrm{dim}(\Theta_k(\widetilde{\G})) = c\cdot(\# \text{ edges  in }\widetilde{\G})$ where $c = 2$ for the quadratic loss \eqref{eqn:convex_quad}; otherwise $c = 1$. As a result of the above two properties, for all these problems, \eqref{eqn:main} can be formulated as follows:
\begin{eqnarray*}
\begin{aligned}
\argmin_{\G,\{\theta^{(k)}\}_{k\in\mathcal{W}_1}} &\quad \sum_{k \in \mathcal{W}_1} \omega_k\Big[ \ell_k(\theta^{(k)};\mathcal{D}_k) + c\lambda_k|E\left\{\pi_{S_k}(\G)\right\}|\Big]+ \gamma\sum_{k \in \mathcal{W}_2} \omega_k\mathrm{SHD}(\pi_{S_k}(\G),\widehat\G^{(k)}_{S_k}),\\
\text{subject-to:}&\quad \G \text{ undirected graph on }V,\\
&\quad \theta^{(k)} \in \Xi_k \quad \forall k\in\mathcal{W}_1,\\[-0.05in]
&\quad  (i,j) \text{ not adjacent in }\pi_{S_k}(\G) \Rightarrow \theta^{(k)}_{ij}=0 \quad \forall i,j \in S_k, k\in\mathcal{W}_1.
\end{aligned}
\end{eqnarray*}
Here and throughout, for notational simplicity, we index the rows and columns of $\theta^{(k)}$ by the elements of $S_k$; thus, for $i,j \in S_k$, $\theta^{(k)}_{ij}$ denotes the entry corresponding to the pair $(i,j)$. 

The optimization problem above is nonconvex and computationally challenging, primarily because of the latent-projection operator $\pi_{S_k}(\cdot)$. We now use the binary variables $g$ and $h$ and the auxiliary continuous variables $f$ and $m$ introduced in Section~\ref{sec:encoding_latent} to reformulate this problem using the linear constraints in \eqref{eqn:flow_constraints} and \eqref{eqn:hidden_component}. To simplify notation, we define their joint domain by:
\begin{eqnarray*}
\begin{aligned}
\mathcal{F}(\{S_k\}_{k=1}^K):= \Bigg\{&(g,h,f,m): g_{i,j} \in \{0,1\}~\forall 1\leq i<j \leq p;~ h_{ij}^{(k)} \in \{0,1\} ~\forall i \in S_k, j\in J_k(i), k \in[K],  \\ 
&f^{(k),i}_{u,v}\geq 0 ~\forall i\in S_k, (u,v)\in A_k(i), k \in [K];~m_u^{(k),i} \in [0,1] ~\forall  i \in S_k, u \in T_k,k \in [K]\Bigg\}. 
\end{aligned}
\end{eqnarray*}
Furthermore, for $k \in \mathcal{W}_2$, let $\widehat h^{(k)}$ encode the estimated graph $\widehat{\G}_{S_k}^{(k)}$, so that $\widehat h_{ij}^{(k)}=1$ if and only if $i$ and $j$ are adjacent in $\widehat{\G}_{S_k}^{(k)}$. The values $\widehat h^{(k)}$ are fixed and are therefore not decision variables in the optimization problem. A straightforward calculation shows that for every $k \in \mathcal{W}_2$:
$$\mathrm{SHD}(\pi_{S_k}(\G),\widehat\G^{(k)}_{S_k}) = \sum_{i,j \in S_k, i<j} |h^{(k)}_{ij}-\widehat h^{(k)}_{ij}|= \sum_{i,j \in S_k, i<j} \widehat h^{(k)}_{ij} + (1-2\widehat h^{(k)}_{ij})h^{(k)}_{ij}.$$
With this simplification, we can obtain the following reformulation of \eqref{eqn:main}:
\begin{equation}
\begin{aligned}
\argmin_{\substack{(g,h,f,m)\\ \{\theta^{(k)}\}_{k\in\mathcal{W}_1}}}
&\quad \sum_{k \in \mathcal{W}_1} \omega_k
\Big[
\ell_k(\theta^{(k)};\mathcal{D}_k)
+ c\lambda_k
\sum_{\substack{i,j\in S_k\\ i<j}}h_{ij}^{(k)}
\Big]
+ \gamma\sum_{k \in \mathcal{W}_2} 
\sum_{\substack{i,j\in S_k\\ i<j}}
\omega_k(1-2\widehat h^{(k)}_{ij})h_{ij}^{(k)},\\
\text{subject-to:}
&\quad (g,h,f,m) \in\mathcal{F}(\{S_k\}_{k=1}^K),\\
&\quad \theta^{(k)} \in\Xi_k \quad \forall k \in\mathcal{W}_1,\\
&\quad (g,h,f,m) \text{ satisfy the linear constraints in }
\eqref{eqn:flow_constraints} \text{ and }\eqref{eqn:hidden_component}
\quad \forall k\in[K],\\
&\quad h_{ij}^{(k)}=0
\Rightarrow \theta^{(k)}_{ij}=\theta_{ji}^{(k)}=0
\quad \forall i,j\in S_k,\ i<j,\ \forall k\in\mathcal{W}_1.
\end{aligned}
\label{eqn:reformulation_main}
\end{equation}
We have omitted additive terms that do not depend on the decision variables. The indicator constraint ensures that whenever an edge is absent from the latent projection, the corresponding model parameters are zero. These indicator constraints can be handled directly by standard mixed-integer optimization solvers such as \emph{Gurobi}. Alternatively, if a valid finite bound $M$ is available, they can be expressed as the linear big-$M$ inequalities:   $-M h_{ij}^{(k)} \leq \theta^{(k)}_{ij} \leq Mh_{ij}^{(k)}$ and $-M h_{ij}^{(k)} \leq \theta^{(k)}_{ji} \leq Mh_{ij}^{(k)}$.  Since the loss $\ell(\theta^{(k)};\mathcal D_k)$ is convex over $\Xi_k$, the continuous optimization problem is convex for any fixed assignment of the binary variables. Thus, \eqref{eqn:reformulation_main}
is a convex MIP.

\noindent\textbf{Perspective strengthening and outer approximation for quadratic losses.}  Formulation \eqref{eqn:reformulation_main} provides a general convex MIP formulation of the estimator \eqref{eqn:main}. When $\ell_k(\theta^{(k)};\mathcal{D}_k)$ is the quadratic loss \eqref{eqn:convex_quad}, we can exploit additional problem structure to develop more efficient convex mixed-integer formulations using perspective strengthening \citep{Cui2012ConvexRA} and outer approximation \mbox{\citep{Duran1986AnOA}}. In our simulations and real data experiments, we apply this faster algorithm. See Appendix~\ref{sec:perspective} for details.

\subsubsection{MIP formulation for Example~\ref{example:discrete}} We describe in Appendix~\ref{sec:discrete_MIP} how the estimator \eqref{eqn:main}, for the loss $\ell_k(\theta^{(k)};\mathcal{D}_k)$ given by either \eqref{eqn:discrete_loss_1} or \eqref{eqn:discrete_loss_2}, can be cast as a convex MIP. For the loss \eqref{eqn:discrete_loss_1}, after restricting $\G$ to be chordal (so that every cycle of length at least four has a chord), we use the fact that chordal graphs admit perfect elimination orderings \citep{Fulkerson1965IncidenceMA}. This allows us to obtain a convex MIP formulation. For the loss \eqref{eqn:discrete_loss_2}, we exploit the decomposition of the minimized pseudo-likelihood into nodewise conditional entropies, expressed in terms of precomputed empirical entropies.

\subsection{Enumerating graphs in an equivalence class} Suppose we obtain a global graph $\widehat{\G}$ after solving \eqref{eqn:reformulation_main}. We next present an integer programming approach to enumerate members of the equivalence class $\mathrm{EC}(\widehat{\G};\{S_k\}_{k=1}^K)$. Recall that this equivalence class consists of global graphs whose latent projections onto $S_k$ match those of $\widehat{\G}$.

Let $\widehat{g}$ be binary variables encoding the graph $\widehat{\G}$. we also obtain binary variables $\widehat h^{(k)}$ encoding the latent projections of $\widehat{\G}$ onto $S_k$, and hence treat these as known. Remark~\ref{remark:latent_graphs} describes how the $\widehat h^{(k)}$ can alternatively be obtained via breadth-first-search algorithm. To enumerate the equivalence class, we initialize the procedure by setting $\widehat g^1:=\widehat g$. Suppose that, after $L$ iterations, we have obtained distinct global graphs $\widehat g^1,\ldots,\widehat g^L$ whose latent projections agree with the estimated local graphs $\widehat h^{(k)}$, $k\in[K]$. To obtain another member of the equivalence class, we solve the following linear MIP:
\begin{eqnarray}
\begin{aligned}
\argmin_{(g,h,f,m)\in\mathcal{F}(\{S_k\}_{k=1}^K)}&\quad 1,\\
\text{subject-to:}&\quad(g,h,f,m)\text{ satisfy the linear constraints in } \eqref{eqn:flow_constraints} \text{ and }\eqref{eqn:hidden_component} ~\forall k\in[K],\\&\quad h^{(k)} = \widehat{h}^{(k)} \quad \forall k\in[K],\\
&\quad \sum_{\substack{1\leq i<j\leq p\\\widehat{g}^\ell_{ij} = 1}}(1-g_{ij})+ \sum_{\substack{1\leq i<j\leq p\\ \widehat{g}^\ell_{ij} = 0}}g_{ij}\geq 1 \quad \forall \ell \in [L].
\end{aligned}
\label{eqn:enumerating}
\end{eqnarray}
Here, the first constraint ensures that $h^{(k)}$ encodes the latent projection of the global graph encoded by $g$ onto $S_k$. The second constraint ensures that these latent projections match those of $\widehat{\G}$. The final constraint ensures that the global graph encoded by $g$ is distinct from each of the previously identified graphs $\widehat g^1,\ldots,\widehat g^L$. 

If \eqref{eqn:enumerating} has a solution $\widehat g^{L+1}$, we add it to the collection of identified graphs and repeat the procedure. If \eqref{eqn:enumerating} is infeasible, we terminate and return $\widehat g^1,\ldots,\widehat g^L$ as the members of the equivalence class.
\begin{remark}
\label{remark:latent_graphs}
If the latent projected graphs of $\widehat\G$ onto $S_k$ are not already available, e.g., from the solution of \eqref{eqn:reformulation_main}, they can be obtained via breadth-first-search (BFS). Specifically, for each node $i \in S_k$, start a BFS traversing on the graph $\widehat\G$. If a node $j \in S_k$, $j \neq i$ is reached, then add edge $i-j$ in the latent projected graph, and do not continue the BFS through $j$. 

\end{remark}

\subsection{Computing summary statistics of an equivalence class}
\label{sec:summary_integer}
In Section \ref{eq:opt_equiv}, we described a general optimization procedure for computing various summary statistics of an equivalence class without enumerating its graphs. In this section, we describe how this optimization can be reformulated as a linear MIP.

As with the previous section, let $\widehat h^{(k)}$ encode the latent projections of $\widehat{\G}$ onto $S_k$; the binary variables $\widehat{h}^{(k)}$ are assumed to be known, otherwise they can be learned (see Remark~\ref{remark:latent_graphs}). 

\noindent\emph{Edge identifiability.}. For a pair of nodes $(i,j)$, we can verify whether this edge is present in every member of the equivalence class by solving the linear MIP:
\begin{eqnarray}
\begin{aligned}
\min_{(g,h,f,m)\in\mathcal{F}(\{S_k\}_{k=1}^K)}&\quad g_{ij},\\
\text{subject-to:}&\quad(g,h,f,m)\text{ satisfy the linear constraints in } \eqref{eqn:flow_constraints} \text{ and }\eqref{eqn:hidden_component} ~\forall k\in[K],\\&\quad h^{(k)} = \widehat{h}^{(k)} \quad \forall k\in[K].
\end{aligned}
\label{eqn:edge_identifiability}
\end{eqnarray}
A solution of one means that $i$ and $j$ are adjacent in every equivalence class member.

\noindent\emph{Logical edge implications.} For two pairs of nodes $(i,j)$ and $(a,b)$, we can verify whether the presence of edge $(a,b)$ implies the presence of edge $(i,j)$ for all equivalent members by solving the linear MIP \eqref{eqn:edge_identifiability} with the added constraint that $g_{a,b}=1$. A solution of one means that the implication holds for every member of the equivalence class.

\noindent\emph{Separation statements.} For disjoint variable sets $A,B,C \subseteq V$, we can verify whether in every graph in the equivalence class, $C$ separates $A$ and $B$ by solving:
\begin{eqnarray*}
\begin{aligned}
\max_{(g,h,f,m)\in\mathcal{F}(\{S_k\}_{k=1}^{K+1})}&\quad \sum_{i \in A}\sum_{j \in B}h_{ij}^{(K+1)},\\
\text{subject-to:}&\quad(g,h,f,m)\text{ satisfy the linear constraints in } \eqref{eqn:flow_constraints} \text{ and }\eqref{eqn:hidden_component} ~\forall k\in[K+1],\\&\quad h^{(k)} = \widehat{h}^{(k)} \quad \forall k\in[K].
\end{aligned}
\end{eqnarray*}
Here, $S_{K+1}:= \{A,B,C\}$. A value of zero means that $C$ separates $A$ and $B$ in every equivalence class member. Under faithfulness (i.e.,  Assumption~\ref{ass:standard_faithfulness}), we would have the additional interpretation that all graphs in the equivalence class specify the independence relationship $X_A \perp X_B \mid X_C$.

\emph{Bounds
on graph attributes.} To find the minimum number of edges in any graph in the equivalence class, we can solve the linear mixed-integer program \eqref{eqn:edge_identifiability} with the objective changed to $\sum_{1\leq i<j \leq p}g_{ij}$. The maximum number of edges would solve a similar problem, but with a maximization. To find a minimum degree of a node $i$ among all equivalence class members, we can solve the linear MIP \eqref{eqn:edge_identifiability} with the objective changed to $\sum_{j:1\le j<i}g_{ji} + \sum_{j:i<j\leq p}g_{ij}$. As before, the maximum degree of a node swaps minimization with maximization.

\section{Coordinate descent solvers for increased scalability} \label{sec:coordinate_descent}

The mixed-integer formulations above provide exact approaches to estimator (4), but may become computationally demanding as the number
of nodes increases. We therefore also develop a coordinate-descent heuristic. In our experiments, its solution is supplied as a warm start to the mixed-integer solver, which is then run with a fixed time budget to improve the solution or certify its optimality.

For a graph $\G$ on $V$, let $F(\G)$ denote its objective value in~\eqref{eqn:main}, and let $\G^{(i,j)}$ denote the graph obtained by flipping the adjacency of nodes $i$ and $j$, namely, replacing $g_{ij}$ with $1-g_{ij}$. Algorithm~\ref{alg:coordinate_descent} repeatedly scans the node pairs and accepts objective-improving flips. Since each accepted update strictly decreases the objective and there are finitely many graphs, the procedure terminates at a graph admitting no improving
single-edge flip.

\begin{algorithm}[H]
\caption{Coordinate descent for global graph estimation}
\label{alg:coordinate_descent}
\begin{algorithmic}[1]
\Require Objective $F$, initial graph $\widehat{\G}$ with $F(\widehat{\G})<\infty$, ordering $\mathcal{O}$ of all node pairs.
\Repeat
    \State $\mathrm{updated}\gets\mathrm{false}$
    \For{$(i,j)\in\mathcal{O}$}
        \If{$F(\widehat{\G}^{(i,j)})<F(\widehat{\G})$}
            \State $\widehat{\G}\gets\widehat{\G}^{(i,j)}$
            \State $\mathrm{updated}\gets\mathrm{true}$
        \EndIf
    \EndFor
\Until{$\mathrm{updated}=\mathrm{false}$}
\State \Return $\widehat{\G}$
\end{algorithmic}
\end{algorithm}

Evaluating $F(\G)$ first requires computing the latent projections, which can be obtained efficiently using BFS (see Remark~\ref{remark:latent_graphs}). For Examples~\ref{example:Gaussian}--\ref{example:extremal}, the continuous parameters can then be optimized by solving convex programs. When the loss decouples across nodes, as in~\eqref{eqn:convex_quad}, only the local scores of nodes whose projected neighborhoods change need to be recomputed.

For Example~\ref{example:discrete} with loss~\eqref{eqn:discrete_loss_1}, we restrict the global graph to be chordal by setting $F(\G)=+\infty$ otherwise. Chordality of a graph $\G$ can be checked in $O(p+|E(\G)|)$ time \citep{tarjan1984simple}; nonchordal proposals are rejected before evaluating the remaining objective. Latent projections of chordal graphs remain chordal. Thus, for each projected graph $\widetilde{\G}$, the minimized loss is
$$
\min_{\theta^{(k)}\in\Theta_k(\widetilde{\G})} \ell_k(\theta^{(k)};D_k) = \sum_{C\in\mathcal{C}_{\max}(\widetilde{\G})} \widehat{H}_k(X_C) - \sum_{M\in\mathcal{M}(\widetilde{\G})} \widehat{H}_k(X_M),
$$
where $\mathcal{C}_{\max}(\widetilde{\G})$ contains the maximal cliques of $\widetilde{\G}$, and $\mathcal{M}(\widetilde{\G})$ is the multiset of their intersections along the edges of a junction tree. A junction tree connects the maximal cliques so that, for every
variable, the cliques containing that variable form a connected subtree. Here, $\widehat H(X_A) = -\sum_{x_A}\widehat{P}^{(k)}(x_A)\log(\widehat{P}^{(k)}(x_A))$ denotes the empirical entropy of $X_A$. Both the maximal cliques and a junction tree can be computed in polynomial time for chordal graphs \citep{tarjan1984simple}, allowing efficient evaluation of the loss when the maximal cliques are not too large. Although one global edge flip may change multiple projected edges, entropy terms for unchanged variable subsets can be reused. Evaluation is efficient when the resulting cliques are small. The space of chordal graphs on $V$ is connected under chordality-preserving single-edge flips \citep{giudici1999decomposable}.

For loss \eqref{eqn:discrete_loss_2}, the minimized loss is 
$$
\min_{\theta^{(k)}\in\Theta_k({\widetilde{\G}})}\ell(\theta^{(k)};\mathcal{D}_k)= \sum_{i \in S_k} \widehat{H}(X_i|X_{\operatorname{N}(i,{\widetilde{\G}})})
$$ 
where $\widehat{H}(X_i|X_A) =\widehat{H}(X_{\{i\}\cup A})-\widehat{H}(X_A)$ is the empirical conditional entropy. Both this minimized loss and its dimension penalty decouple across nodes. Consequently, as with loss~\eqref{eqn:convex_quad}, only contributions from nodes whose projected neighborhoods change need to be recomputed. Computation is efficient when these neighborhoods are small.

\begin{remark} The initial graph and the ordering of pairs in $\mathcal{O}$ determine the initialization of the algorithm. As each run of coordinate descent is relatively fast, we recommend running the algorithm with several different initializations and selecting the solution with the lowest objective value. In our simulations and real-data experiments, we consider two initial graphs: the empty graph and coordinate descent solution of \eqref{eqn:main} with $\mathcal{W}_2=[K]$, obtained by first estimating a local graph separately for each source. For this preliminary  coordinate descent procedure, we initialize at the empty graph and use the canonical ordering $\mathcal{O}$, in which all pairs involving node 1 are updated first, followed by those involving node 2, and so on. For the main procedure, we consider both the canonical ordering and a random ordering of the pairs in $\mathcal{O}$.
\end{remark}

\section{Theoretical Guarantees}
\label{sec:theory}
\subsection{Setup and definitions}We suppose all sources provide data, i.e. $|\mathcal{W}_1|=K$ and $\mathcal{W}_2=\emptyset$. The data $\mathcal{D}_k$ in each source $k$ are assumed to be independently sampled from a zero-mean multivariate Gaussian distribution over the observed variable subset $S_k$. We denote the covariance matrix of this marginal distribution by $\Sigma_k^\star$. Motivated by both computational and statistical considerations (see Sections \ref{sec:disccusion_consistent} and \ref{sec:MIP_reformulation}), we use the loss in \eqref{eqn:convex_quad} for learning a global graph from this data. 

We restrict the class of underlying graphs by controlling their maximum degree and the complexity of their latent projections. Specifically, we assume that the maximum degree is at most $r_{\max,n}$ and that, for every source $k$, each connected component of $\G[V\setminus S_k]$ has at most $b_n$ neighbors in $S_k$. We denote by $\mathfrak{G}$ the resulting class of global graphs and restrict the search in estimator \eqref{eqn:main} to $\mathfrak{G}$. These conditions imply that, for every $\G\in\mathfrak{G}$, the maximum degree of each source-specific latent projection $\pi_{S_k}(\G)$ is at most $d_n:=r_{\max,n}(b_n-1)$. We allow both $r_{\max,n}$ and $b_n$ to grow with the sample size and, for concreteness, take $r_{\max,n}$ and $b_n$ so that $d_n \leq \min\{\tilde{C}(\min_kn_k/\log(p))^{1/4},p/2\}$. We assume that $\G^\star \in \mathfrak{G}$.

For every $\G \in \mathfrak{G}$, $k\in[K]$, we define the population loss
\begin{equation*}
\mathcal{L}(\G)
:= \sum_{k\in[K]}
\sum_{j\in S_k}\omega_k\mathcal{L}_{kj}(\G)\quad \text{ with }\quad \mathcal{L}_{kj}(\G):=\min_{\beta, \text{support}(\beta)\subseteq \{j\}\cup N(j,\pi_{S_k}(\G)), \beta_j=1}
\beta^\top\Sigma_k^\star\beta.
\end{equation*}
Here, $N(j,\G)$ denotes the nodes adjacent to node $j$ in the graph $\G$. Minimizing $\mathcal{L}(\G)$ over $\G\in\mathfrak{G}$
is the population analogue of the estimator \eqref{eqn:main}. By definition, $\mathcal{L}_{kj}(\G) \geq \mathcal{L}_{kj}(\G^\star)$ for every $\G \in \mathfrak{G}$ and $k \in [K], j\in S_k$. Thus, $\mathcal{L}(\G) \geq \mathcal{L}(\G^\star)$. Define:
$$\mathcal{A}^-(\G,\G^\star)
:=
\left\{
(k,j):
N(j,\pi_{S_k}(\G^\star))
\not\subseteq
N(j,\pi_{S_k}(\G))
\right\}.$$
Thus, $\mathcal{A}^-(\G,\G^\star)$ consists of source–node pairs for which the projected graph of $\G$ omits at least one neighbor of $j$ present in the projected graph of $\G^\star$.  

We next define two population combinatorial quantities that characterize the difficulty of distinguishing the true global graph from competing graphs. The first measures the smallest average population loss incurred by a graph that omits true neighbors in more than $r$ source/node pairs:
\begin{equation}
\Gamma^-(r):=
\inf_{\G\in\mathfrak{G}:\,|\mathcal{A}^-(\G,\G^\star)|\geq r}
\frac{
\mathcal{L}(\G)-\mathcal{L}(\G^\star)
}{
|\mathcal{A}^-(\G,\G^\star)|
}.
\end{equation}
We use the convention $\inf\varnothing=+\infty$. The second quantity characterizes the difficulty of recovering a particular source/node neighborhood after restricting attention to graphs that are sufficiently close to $\G^\star$. For a localization radius $q\geq 0$, define
\begin{equation}
\delta_{kj}(q):= \inf_{\substack{\G\in\mathfrak{G}, |\mathcal{A}^-(\G,\G^\star)|\leq q\\N(j,\pi_{S_k}(\G^\star))\not\subseteq N(j,\pi_{S_k}(\G))}}\mathcal{L}(\G)-\mathcal{L}(\G^\star).
\end{equation}
The quantity $\delta_{kj}(q)$ is the smallest population loss incurred by any graph within this localized class that omits at least one true neighbor of node $j$ in source $k$. Larger values of $\delta_{kj}(q)$ indicate a greater separation in population loss between $\G^\star$ and competing graphs that fail to recover all true neighbors of node $j$ in source $k$. Notice that $\delta_{kj}(q)$ is a decreasing function of $q$. 

\begin{remark}We emphasize the global nature of the quantity $\delta_{kj}(q)$, even though it is defined for a particular source and node. In particular, if $N(j,\pi_{S_k}(\G^\star))\not\subseteq N(j,\pi_{S_k}(\G))$, then $\mathcal{A}^-(\G,\G^\star)$ necessarily contains the pair $(k,j)$. However, it may also contain other pairs because, due to the global dependence structure induced by latent projections, omitting a true neighbor of node $j$ in one source can trigger the omission of true neighbors in other source–node pairs. 

The following stylized example illustrates the phenomena described above. Suppose $\G^\star$ contains $m$ internally vertex-disjoint paths $i-h_\ell-j$, $\ell=1,\dots,m$, together with an isolated node $z$. Let source 1 observe only $\{i,j\}$ and source 2 observe all variables except $z$. Then $(i,j)$ is present in the latent projection for source 1. If a competing graph $\G$ omits this projected edge, then for every $\ell$, at least one of the true source-2 edges $(i,h_\ell)$ or $(j,h_\ell)$ must also be omitted; otherwise, $\G$ has a path from $i$ to $j$. Hence $(2,h_\ell)\in\mathcal A^{-}(\G,\G^\star)$ for all $\ell=1,2,\dots,m$, and together with the omission at $(1,i)$ and $(1,j)$, this yields $|\mathcal A^-(\G,\G^\star)|\ge m+2$.
\end{remark}

\begin{remark}Note that:
\begin{equation*}
\delta_{kj}(q)= \inf_{\substack{\G\in\mathfrak{G}, |\mathcal{A}^-(\G,\G^\star)|\leq q\\N(j,\pi_{S_k}(\G^\star))\not\subseteq N(j,\pi_{S_k}(\G))}}\sum_{(l,m)\in \mathcal{A}^-(\G,\G^\star)}\omega_l\left\{\mathcal{L}_{lm}(\G)-\mathcal{L}_{lm}(\G^\star)\right\},
\end{equation*}
where by definition, $\mathcal{L}_{lm}(\G)-\mathcal{L}_{lm}(\G^\star)>0$ for every $(l,m)\in \mathcal{A}^-(\G,\G^\star)$. To better understand how small the gap $\mathcal{L}_{lm}(\G)-\mathcal{L}_{lm}(\G^\star)$ can be, let $T^\star =N(m,\pi_{S_l}(\G^\star))$ and $T:= N(m,\pi_{S_l}(\G))$. Define $\beta^\star:= \left([\Sigma^\star_l]_{T^\star,T^\star}\right)^{-1}[\Sigma^\star_l]_{T^\star,m}$ which encode the strength of true edges adjacent to node $m$ in source $l$. Then, a straightforward calculation shows $$\mathcal{L}_{lm}(\G)-\mathcal{L}_{lm}(\G^\star) = [\beta^\star]_{M}^\top [\Sigma_l^\star]_{M|T}[\beta^\star]_M,$$ where $M= T^\star\setminus T$, and $[\Sigma_k^\star]_{M|T}$ is the conditional covariance of the variables in $M$ given those in $T$ in source $l$. Thus, $\delta_{kj}(q)$ aggregates the signal strengths of all omissions that must accompany an error at node $j$ in source $k$, highlighting the benefit of integrating information across sources.
\label{remark:diff}
\end{remark}
\begin{remark}Consider a source-specific neighborhood selection estimator for recovering the neighborhood of node $j$ in source $k$. Let $\ell_{kj}(T)$ denote the population loss obtained by regressing node $j$ on a candidate neighborhood $T\subseteq S_k\setminus\{j\}$. The corresponding population separation is
\[
\delta_{kj}^{\mathrm{local}}
:=
\inf_{\substack{
T\subseteq S_k\setminus\{j\}:\\
N(j,\pi_{S_k}(\G^\star))\not\subseteq T
}}
\ell_{kj}(T)
-
\ell_{kj}\bigl(N(j,\pi_{S_k}(\G^\star))\bigr).
\]
Unlike $\delta_{kj}(q)$, this quantity depends only on the population neighborhood-selection problem for node $j$ in source $k$: it does not involve the neighborhoods of any other nodes or the distributions in any other sources. In contrast, $\delta_{kj}(q)$ is defined through a common global graph and therefore incorporates the cross-source constraints induced by its latent projections.

\end{remark}

\subsection{Theoretical results}
Our results rely on a few additional notations. For any disjoint sets $A,B \subseteq S_k$, let $[\Sigma_k^\star]_{A|B} := [\Sigma_k^\star]_{A,A}-[\Sigma_k^\star]_{A,B}([\Sigma_{k}^\star]_{B,B})^{-1}[\Sigma_{k}^\star]_{B,A}$. Define:
\begin{eqnarray*}
\begin{aligned}
\kappa_{\text{min}}:= \min_{k\in[K]}\min_{\substack{A\subseteq S_k, |A| \leq 2d_n\\B\subseteq S_k, |B| \leq 2d_n, A\setminus B \neq \emptyset}}\sigma_\text{min}([{\Sigma}_k^\star]_{A|B}), \quad
\kappa_{\text{max}}:= \max_{k\in[K]}\max_{A\subseteq S_k, |A| \leq 2d_n}\sigma_\text{max}([{\Sigma}_{k}^\star]_{A,A}).
\end{aligned}
\end{eqnarray*}
Here, for a positive definite matrix $M$, $\sigma_\text{min}(M)$ and $\sigma_\text{max}(M)$ represent its smallest and largest singular values, respectively. Let $\bar{C}:= 4+\frac{8\kappa_{\text{max}}}{\kappa_{\text{min}}}+\frac{4\kappa_{\text{max}}^2}{\kappa_{\text{min}}^2}$. Define
$$t_k:= {\kappa_\text{max}}
\sqrt{\frac{256(d_n+1)\log(2Kp)}{{n_k}}}.$$
Let $C \geq 1$ be any constant satisfying:
$$\frac{\max_{k\in[K]} \omega_kt_k}{\min_{k\in[K]} \omega_kt_k}\leq C.$$
Define 
$$
q_\text{min}:=
\min\left\{
r\in\mathbb{N}:
\Gamma^-(r)>
(d_n+1/2)\max_{k\in[K]} \omega_k\lambda_k
\right\}
(1+C(2d_n+1)).
$$

We then have the following result, with proof in Appendix~\ref{proof:main}. 
\begin{theorem}For any $k\in[K]$, define $\delta_{k,\text{min}}:= \min_{i\in S_k}\delta_{ki}(q_{\min})$. Suppose for every $k\in[K]$, 
\begin{equation*}
\begin{aligned}
n_k &\geq \frac{1024(d_n+1)\kappa_{\text{max}}^2}{\kappa_{\text{min}}^2}\log(2Kp),\\
n_k &\geq \frac{4096C^2q_\text{min}^2\bar{C}^2\omega_k^2\kappa_{\text{max}}^2(d_n+1)^3}{\delta_{k,\text{min}}^2} \log(2Kp).
\end{aligned}
\end{equation*}
Suppose for every $k\in[K]$, $\lambda_k$ is chosen as to satisfy:
\begin{equation*}
2\bar{C}t_k\leq \lambda_k
\leq
\frac{\delta_{k,\text{min}}}{2Cq_\text{min}(d_n+1/2)\omega_k},
\end{equation*}
while satisfying $\max_{k\in[K]} \omega_k\lambda_k/\min_{k\in[K]}\omega_k\lambda_k \leq C$.
Then, with probability greater than $1-1/p$, $\widehat{\G}\in \mathrm{EC}(\G^\star, \{S_k\}_{k=1}^K)$. 
\label{thm:main}
\end{theorem}
Theorem~\ref{thm:main} highlights four main drivers of the sample size requirement. First, the logarithmic factor $\log(Kp)$ is the expected high-dimensional dependence arising from uniform control over all source-node pairs. Second, $d_n$ controls the local complexity of the projected graphs. Larger projected neighborhoods make estimation harder and thus require a larger sample size. Third, the conditional number of the covariance matrices matters through $\kappa_\text{max}/\kappa_\text{min}$. Thus, poorly conditioned covariance matrices require more samples. Finally, and most importantly, the recovery difficulty is governed by the ratio $q_\text{min}/\delta_{k,\text{min}}$. Here, $q_\text{min}$ measures how large a neighborhood of competing graphs must be considered after localization, while $\delta_{k,\text{min}}$ measures the smallest population separation from a graph that makes an error in some source. Larger $\delta_{k,\text{min}}$ reduces the required sample size, while larger $q_\text{min}$ increases it. 

The recovery guarantee in Theorem~\ref{thm:main} is expressed in terms of the population separation quantity $\delta_{k,\text{min}}$. We next provide sufficient conditions that relate this abstract quantity to the strengths of the true local regression coefficients. Let $T^\star_{kj}:= N(j,\pi_{S_k}(\G^\star))$. Define $\beta^\star_{kj}:= \left([\Sigma_{k}^\star]_{T^\star_{kj},T^\star_{kj}}\right)^{-1}[\Sigma_{k}^\star]_{T^\star_{kj},j}$ and define $\beta^\star_{kji}:= e_i^\top \beta^\star_{kj}$ with $e_i$ denoting a standard basis element. Suppose $\mathcal{W}\subseteq \{(k,j): k\in[K], j\in S_k\}$ is the set of weak source-node pairs and let 
$$\beta_{k,\text{min}} := \min_{(i,j) \in E(\pi_{S_k}(\G^\star)), (k,j)\not\in \mathcal{W}}|\beta^\star_{kji}|.$$ 
Thus, outside $\mathcal{W}$, every true edge in source $k$ has strength at least $\beta_{k,\text{min}}$. The next proposition shows that if the strong signals dominate the stochastic error level, then $q_\text{min}$ can be controlled explicitly, and $\delta_{k,\text{min}}$ is lower bounded by the cumulative contribution of the strong omissions that necessarily accompany an error in source $k$. The proof is in Appendix~\ref{proof:prop_delta}. 
\begin{proposition}\label{prop_delta}
If 
$$\kappa_\text{min}\min_{k\in[K]}\{\omega_k\beta_{k,\text{min}}^2\} >2(d_n+1/2)\max_{k\in[K]} \omega_k\lambda_k,$$then:
$$\delta_{k,\text{min}} \geq \min_{j\in S_k} \inf_{\substack{\G\in\mathfrak{G}, |\mathcal{A}^-(\G,\G^\star)|\leq q_0\\N(j,\pi_{S_k}(\G^\star) \not\subseteq N(j,\pi_{S_k}(\G))}}\sum_{\substack{(l, m)\in\mathcal{A}^-(\G,\G^\star)\\(l,m)\not\in\mathcal{W}}}\kappa_\text{min}\omega_l\beta_{l,\text{min}}^2,$$
where $q_0:= \left(2|\mathcal{W}| +1\right)(1+C(2d_n+1))$.
\end{proposition}
The bound makes explicit the benefit of integration: $\delta_{k,\text{min}}$ need not be determined solely by the weakest signal at the target source. If an error there forces additional omissions elsewhere, the corresponding strong signals also contribute to the separation. In particular, a locally weak source-node pair can still be recoverable when the shared global graph structures link that error to sufficiently strong signals in other sources.


\section{Synthetic Experiments} \label{sec:simulations}
\subsection{Data generating process}
We generate $\G^{\star}$ from an Erd\H{o}s--R\'{e}nyi model with $\mathbb{P}\{(i,j)\in E^{\star}\}=\frac{d}{p-1}$, 
where $d$ is the expected node degree. Conditional on $\G^{\star}$, each source $\ell$ has a Gaussian distribution with precision matrix
$\Omega^{(k)}$. The off-diagonal support of every $\Omega^{(k)}$ is $E^{\star}$. For each $(i,j)\in E^{\star}$, independently across edges and sources, we draw $
U^{(\ell)}_{ij}\sim\operatorname{Unif}(0.5,0.75)$ and
$\mathbb{P}(\xi^{(\ell)}_{ij}=1) =\mathbb{P}(\xi^{(\ell)}_{ij}=-1) =\frac{1}{2}$. 
We then set $
\Omega^{(k)}_{ij} =\Omega^{(k)}_{ji}
=\xi^{(k)}_{ij}U^{(k)}_{ij}$ 
for $(i,j)\in E^{\star}$ and set all other off-diagonal entries to zero. The diagonal entries are
$
\Omega^{(k)}_{ii}
=0.5+\sum_{j\neq i}|\Omega^{(k)}_{ij}|$. 
This makes $\Omega^{(k)}$ strictly diagonally dominant and positive definite. The sources share the graph $\G^{\star}$, while their edge weights and covariance matrices $\Sigma^{(k)}=(\Omega^{(k)})^{-1}$ may differ.

Source $k$ observes a subset $S_{k}\subseteq V$ containing $q=(1-m)p$ variables, where $m$ is the missingness fraction. For a set
$T\subseteq V$ denote $\operatorname{cc}(\G^{\star}[T])$ as the collection of connected components with nodes in $T$. For a candidate set
$S\subseteq V$ and a connected component
$C\in\operatorname{cc}(\G^{\star}[V\setminus S])$, define the subset $
\partial_{S}(C)
:=\{i\in S:(i,j)\in E^{\star}\text{ for some }j\in C\}$. This subset collects vertices in $S$ that are connected to some node in the connected component $C$. The collection of admissible observed sets is
$$
\mathcal{H}_{q}(\G^{\star},l)
:=\left\{S\subseteq V:|S|=q,
\max_{C\in\operatorname{cc}(\G^{\star}[V\setminus S])}
|\partial_{S}(C)|\leq l\right\}.
$$
This restriction limits the edges introduced by marginalization: each connected component of unobserved vertices makes its observed boundary a clique on at most $l$ vertices.

We generate each $S_{k}$ using a Markov chain on
$\mathcal{H}_{q}(G^{\star},l)$. A randomized search first provides a feasible initial state $S^{(0)}$. At iteration $t$, the current state is $S^{(t)}$. With probability $0.05$, the chain stays at $S^{(t)}$. Otherwise, it selects a radius $r\in\{1,\ldots,\min(q,p-q)\}$. Conditional on making a proposal, radius one has probability $0.80$, radius two has probability $0.15$, and total probability $0.05$ is distributed uniformly over the remaining radii.  Given $r$, the chain samples $r$ nodes uniformly from $S^{(t)}$ and $r$ nodes uniformly from $V\setminus S^{(t)}$. Denote these sets by $A$ and $B$. The
proposed state is
$
S^{\prime}=\left(S^{(t)}\setminus A\right)\cup B$. If $S^{\prime}\in\mathcal{H}_{q}(\G^{\star},l)$, the proposal is accepted and $S^{(t+1)}=S^{\prime}$. Otherwise, the proposal is rejected and $S^{(t+1)}=S^{(t)}$. It is straightforward to show that the chain is reversible, irreducible, and aperiodic, and thus the procedure samples uniformly from $\mathcal{H}_{q}(G^{\star},l)$.
We run one chain for each source for $10{,}000$ transitions and retain the
terminal state as $S_{k}$. Given $S_{k}$, source $k$ provides observations 
$X^{(k)}_{1},\ldots,X^{(k)}_{n}
\mathrel{\overset{\mathrm{iid}}{\sim}}
N_{q}(0,\Sigma^{(k)}_{S_{k},S_{k}})$. 

In the simulation settings we consider $p\in\{10,25,50\}$, $K\in\{5,10\}$, $n\in\{100,500,1000\}$, $m\in\{0.2,0.4\}$, and $(d,l)\in\{(2,2),(2,3),(3,3)\}$, where $d$ is the expected node degree and $l$ is the maximum allowed boundary size in $\mathcal{H}_{q}(\G^{\star},l)$. Each setting contains 50 replications. 
\subsection{Performance evaluation}

We evaluate the estimated global graph $\widehat{\G}$ using the projected-SHD: $$\tau(\widehat{\G})
:=\underbrace{\frac{1}{K}\sum_{k=1}^{K}
\operatorname{SHD}\left(\pi_{S_k}(\widehat{\G}),\pi_{S_k}(\G^\star)\right)}_{\text{projected-SHD}}.$$
This metric measures how accurately the joint procedure recovers the marginal graphs. Notice that $\widehat{\G}$ is in the population equivalence class, i.e., $\widehat{\G} \in \mathrm{EC}(\G^\star)$, if and only if $\tau(\widehat{\G}) = 0$. We say that the projection-equivalence class is recovered exactly when $\tau(\widehat{\G}) = 0$.

We compare four procedures: (i) \emph{data-integrated coordinate descent} where an approximate solution to \eqref{eqn:main} is obtained (see Section~\ref{sec:coordinate_descent}); (ii) \emph{data-integrated MIP} that solves \eqref{eqn:main} initialized with the coordinate descent solution and run until solution is found or the time limit of 300 seconds is reached; (iii) \emph{graph-integrated coordinate descent} that first estimates a local graph for each source and then combines these local graphs by approximately solving \eqref{eqn:main} via coordinate descent; and (iv) \emph{separate-source baseline} that does not combine the local graphs. For the separate-source baseline, projected-SHD is computed as $K^{-1}\sum_{k=1}^K \mathrm{SHD}(\widehat{\G}_{S_k}^{(k)},\pi_{S_k}(\G^\star))$, where $\widehat{\G}_{S_k}^{(k)}$ denotes the estimated local graph for source $k$. Thus, a lower projected-SHD for any of our integration approaches (i)–(iii) than for the separate-source baseline indicates a benefit from integration.

For the two data-integrated procedures, we use the quadratic loss in \eqref{eqn:convex_quad} from Example~\ref{example:Gaussian}, with the Bayesian Information Criterion (BIC)-based penalty $\lambda_1=\cdots=\lambda_K=\log(n)/n$.
For both graph integration and the separate-source baseline, source graphs are estimated by neighborhood selection \citep{Meinshausen2006HighdimensionalGA}. Specifically, in each source, every observed variable is regressed on all other observed variables using the lasso, with the regularization parameter selected by the BIC. An undirected edge is retained only when both corresponding directed regressions select it, corresponding to the AND rule. We also compared this estimator with neighborhood selection using the OR rule, with regularization parameters selected by cross-validation or stability-based tuning and either symmetrization rule, as well as graphical lasso with BIC or cross-validation.  Figures~\ref{fig:baseline-oracle-d2-b2}--\ref{fig:baseline-oracle-d3-b3} in Appendix~\ref{sec:additional_simulations} compare the chosen estimator with an oracle-best alternative selected separately in each setting from the other seven procedures. Despite this retrospective advantage, the oracle alternative is usually less accurate than the neighborhood selection with the AND rule.

The exact MIP-based procedures were given a time limit of 300 seconds, measured after initialization. The coordinate descent procedures used to construct warm starts were not subject to this time limit and ran until their convergence criterion was satisfied or a maximum of 40 iterations (in our simulations they always converged within 40 iterations).

\subsection{Results with data available from all sources}

Figures~\ref{fig:all-data-d2-b2} and \ref{fig:all-data-d3-b3} present the results for $(d,l)=(2,2)$ and $(3,3)$, respectively. Rows correspond to sample size, columns to the number of variables, and groups within each panel to combinations of source count and missingness. The boxplots summarize projected-SHD $\tau(\widehat{\G})$ across the 50 replications, with exact projection-equivalence-class recovery rates reported above them. Results for $(d,l)=(2,3)$ appear in Figure~\ref{fig:all-data-d2-b3} in Appendix~\ref{sec:additional_simulations}.

The figures show a consistent ordering: the data-integrated solutions generally have the lowest $\tau(\widehat{\G})$, followed by graph integration and then separate-source estimation. Among the two data-integrated approaches, the coordinate descent and MIP results are often nearly indistinguishable, in part because the coordinate descent solution initializes the MIP. Nevertheless, the MIP incumbent achieves lower projected-SHD in 55 of the 108 settings and ties coordinate descent in another 36; its exact recovery rate is higher in 37 settings and equal in 67. Both data-integrated solutions have lower projected-SHD than graph integration in all 108 settings and generally attain equal or higher exact recovery rates. The separate-source baseline also tends to exhibit greater variability, particularly for larger $p$, greater missingness, or greater structural complexity.

Increasing the sample size shifts the projected-SHD distributions toward zero in nearly every panel, with the largest gains generally occurring between $n=100$ and $n=500$. For example, in Figure~\ref{fig:all-data-d2-b2}, when $p=50$, $K=10$, and $m=0.2$, exact recovery for both data-integrated solutions increases from $4\%$ at $n=100$ to $100\%$ at $n=500$ and remains at $100\%$ at $n=1000$. Over the same range, graph integration improves from $0\%$ to $60\%$ and then to $90\%$, whereas separate-source estimation remains at $0\%$. The same pattern appears in Figure~\ref{fig:all-data-d3-b3}, although recovery is generally more difficult under $(d,l)=(3,3)$.

Within each panel, the effect of increasing the number of sources can be assessed by comparing the first and third groups at $20\%$ missingness and the second and fourth groups at $40\%$ missingness. Additional sources most consistently benefit the data-integrated solutions, with smaller gains for graph integration and, obviously, no benefit for separate-source estimation. For example, in Figure~\ref{fig:all-data-d2-b2}, when $n=1000$, $p=50$, and $m=0.4$, exact recovery increases from $16\%$ to $56\%$ for coordinate descent, from $24\%$ to $60\%$ for the MIP incumbent, and from $8\%$ to $18\%$ for graph integration, while remaining at $0\%$ for separate-source estimation. This improvement is notable because exact recovery becomes more stringent as the number of sources increases.

Increasing the missingness fraction from $20\%$ to $40\%$ substantially degrades recovery, as reflected by higher projected-SHD and lower exact recovery across panels. In Figure~\ref{fig:all-data-d3-b3}, for example, when $n=1000$, $p=25$, and $K=10$, exact recovery decreases from $98\%$ to $18\%$ for coordinate descent, from $98\%$ to $32\%$ for the MIP incumbent, and from $20\%$ to $0\%$ for graph integration. Greater missingness both reduces the information available within each source and decreases the amount of overlap across sources.

Recovery also becomes more difficult as $p$ increases, with projected-SHD shifting upward and becoming more variable across the columns of Figures~\ref{fig:all-data-d2-b2} and \ref{fig:all-data-d3-b3}. For example, in Figure~\ref{fig:all-data-d2-b2}, with $n=1000$, $K=5$, and $m=0.4$, exact recovery for coordinate descent decreases from $76\%$ at $p=10$ to $48\%$ at $p=25$ and $16\%$ at $p=50$; the corresponding rates are $76\%$, $56\%$, and $24\%$ for the MIP incumbent and $62\%$, $26\%$, and $8\%$ for graph integration.

Comparing Figures~\ref{fig:all-data-d2-b2} with \ref{fig:all-data-d2-b3} in Appendix~\ref{sec:additional_simulations} isolates the effect of increasing the maximum boundary size while holding the expected degree at $d=2$. The values of projected-SHD $\tau(\widehat{\G})$ are generally higher under $l=3$, and exact recovery is lower. A larger boundary allows more edges to be induced through unobserved variables. These dependences can be weaker than direct relationships, making recovery more difficult even when the total number of edges in the projected graphs changes little. For example, when $n=500$, $p=50$, $K=10$, and $m=0.4$, projected-SHD for coordinate descent increases from $0.50$ under $l=2$ to $1.98$ under $l=3$, while exact recovery decreases from $44\%$ to $12\%$. For the MIP incumbent, projected-SHD increases from $0.49$ to $1.96$, and exact recovery decreases from $46\%$ to $10\%$. Graph integration is affected more strongly, with projected-SHD increasing from $3.01$ to $9.61$ and exact recovery decreasing from $2\%$ to $0\%$. 

Comparing Appendix Figure~\ref{fig:all-data-d2-b3} with Figure~\ref{fig:all-data-d3-b3}, we observe that increasing the expected degree from $d=2$ to $d=3$ substantially worsens recovery across methods, although the relative ordering of the approaches generally remains unchanged. Higher degrees increase the number of relationships to estimate. Under our precision-matrix construction, they also produce larger diagonal entries, which can reduce absolute partial correlations even when the off-diagonal coefficients remain comparable. Recovery therefore becomes harder both because there are more edges to identify and because their conditional dependences can be weaker.

\begin{figure} 
    \centering
    \includegraphics[width=1.1\linewidth]{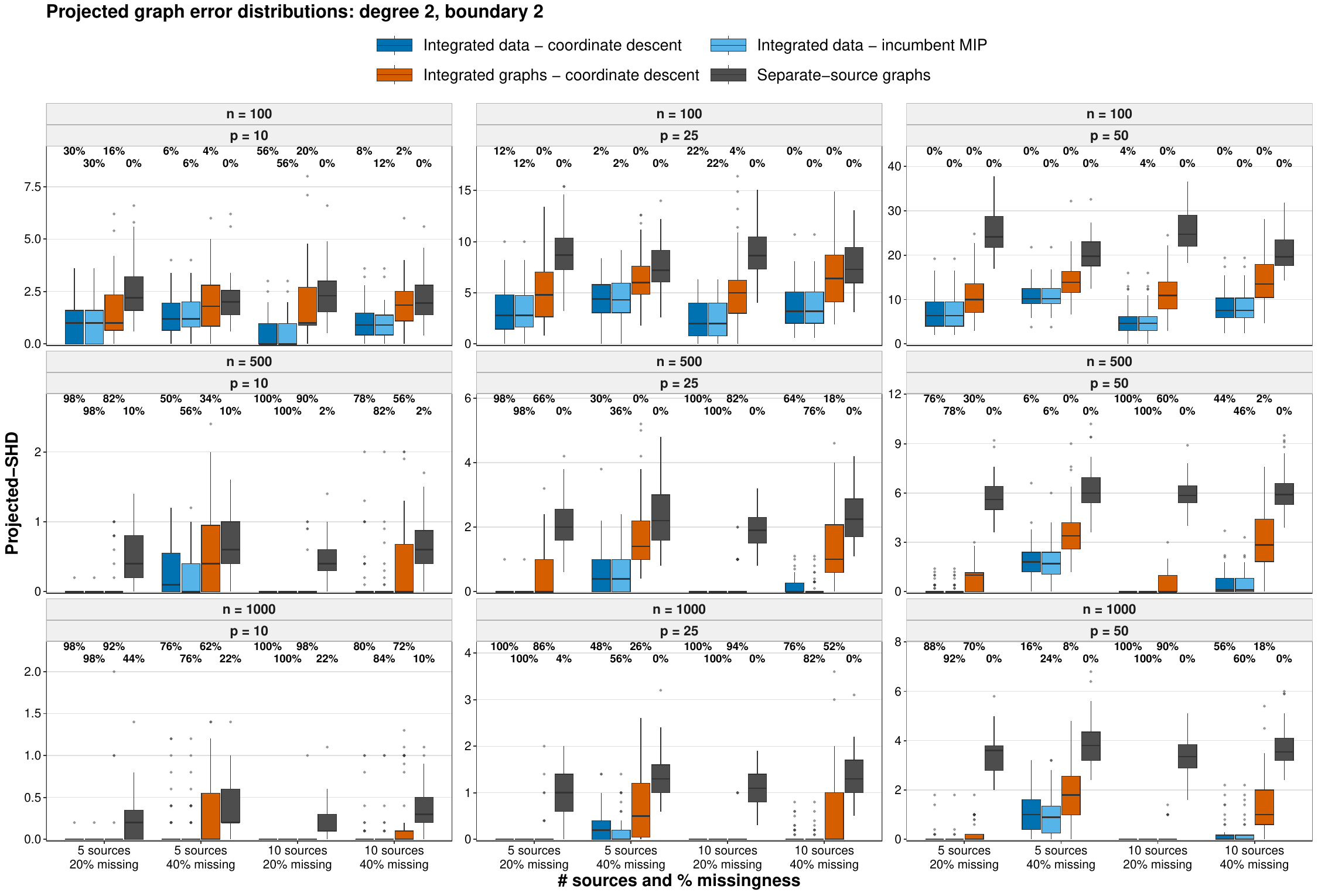}
     \caption{\small Projected-SHD for expected degree $d=2$ and maximum boundary size $l=2$. Boxes summarize the 50 replications. Percentages at the top of each panel report exact recovery of the projection-equivalence class.}
     \label{fig:all-data-d2-b2}
\end{figure}

\begin{figure}
    \centering
    \includegraphics[width=1.1\linewidth]{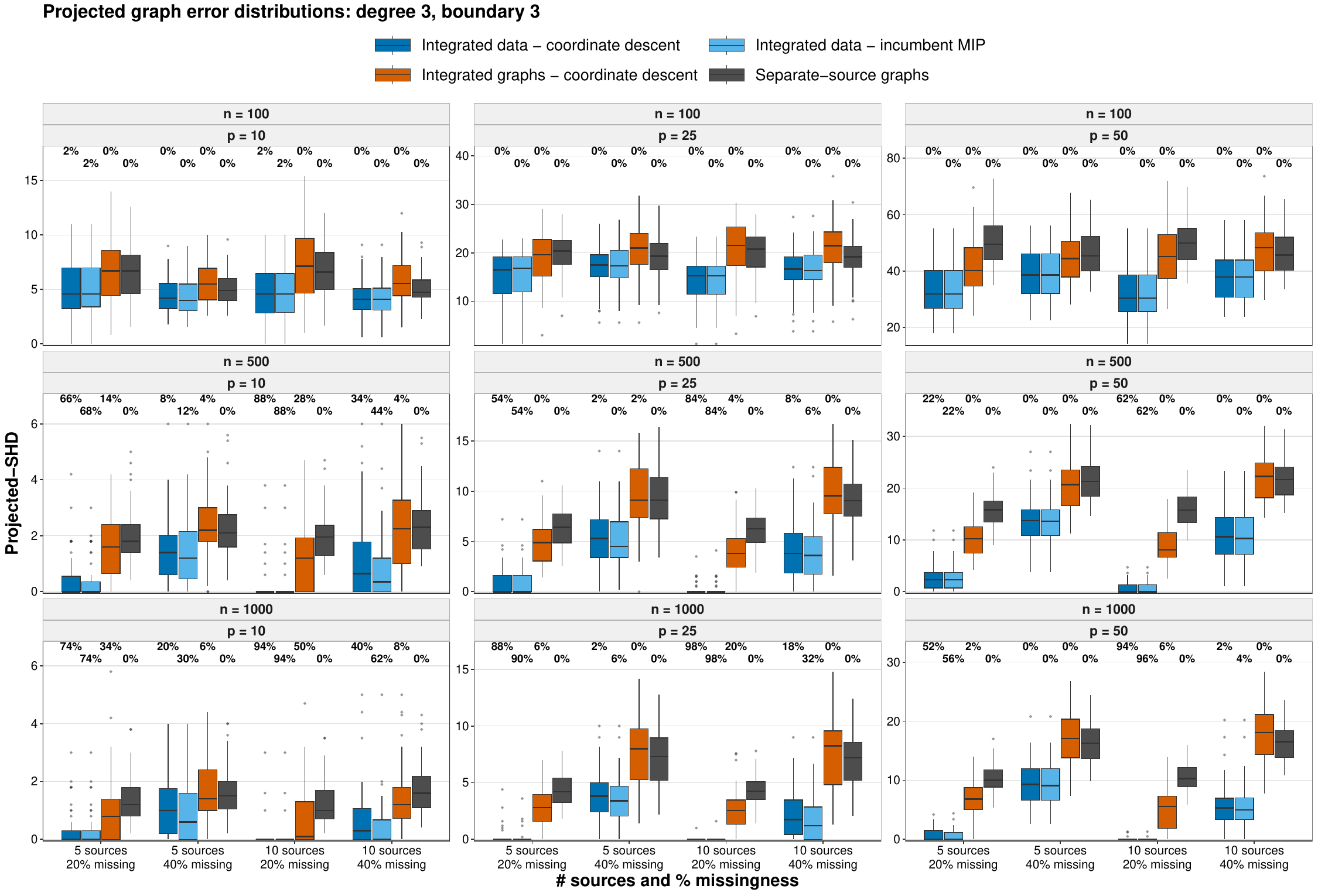}
    \caption{\small Projected-SHD for expected degree $d=3$ and maximum boundary size $l=3$. Boxes summarize the 50 replications. Percentages at the top of each panel report exact recovery of the projection-equivalence class.}
     \label{fig:all-data-d3-b3}
\end{figure}

\subsection{Results when some sources provide data and others just graphs}

We consider a hybrid setting with five sources with an expected degree $d=2$, in which three sources provide data and two provide estimated graphs. That is, $\mathcal{W}_1 = \{1,2,3\}$ and $\mathcal{W}_2 = \{4,5\}$. We apply the coordinate descent procedure and compare the prespecified graph-source weight $c=1$ with a value selected by cross-validation (see Remark \ref{remark:choosing_gamma}). Fixed $c=1$ generally provides better graph recovery. It attains lower projected-SHD in 32 of the 36 settings, while the cross-validated choice performs better in four. Its exact recovery rate is higher in 19 settings, equal in 13, and lower in four. Relative to the previously reported all-data and all-graph results, the hybrid procedure with $c=1$ lies between the two information endpoints. In particular, averaging projected-SHD over the 50 replications within each setting, we obtain
$\text{all data}
<
\text{three data and two graph sources}
<
\text{all graph sources}
$
in 35 of the 36 settings. Similarly, in 33 settings, the hybrid exact recovery rate is no higher than that of the all-data procedure and no lower than that of the all-graph procedure.

\subsection{Comparison with graph quilting procedures}

We compare the proposed method with the procedures in \cite{Vinci2019GraphQG,Chang2023NonparanormalGQ}. These procedures, called graph quilting (GQ), reconstruct a global graph assuming all the sources provide data, and the data are generated from the same distribution. As we show in the following experiment, their performances deteriorate under distribution shifts. 

 We generate Gaussian data with $p=10$ variables, $K=5$ data sources, $20\%$ missingness in each source, expected degree $d=2$, maximum boundary size $l=2$, and  sample sizes $n\in\{100,500,1000\}$. We draw a matrix $\Gamma \in \mathbb{R}^{p \times p}$ in a similar fashion as Section~\ref{sec:simulations}. For each source $k\in[K]$ and true edge $(i,j)$, independently draw $B_{ij}^{(k)}\sim\operatorname{Unif}(-\log 2,\log 2)$ and $U_{ij}^{(k)}\sim\operatorname{Unif}(0,1)$. Then, for $\alpha\in[0,1]$, we generate the precision matrix $\Omega^{(k)}(\alpha)$ for source $k$ as follows. For each true edge $(i,j)$, let
$$
\Omega_{ij}^{(k)}(\alpha)
=
\Gamma_{ij}\exp\!\left(\alpha B_{ij}^{(k)}\right)
(-1)^{\mathbb{I}\{U_{ij}^{(k)}\leq\alpha/2\}}.
$$
All other off-diagonal entries remain zero, and symmetry is enforced. The diagonal entries are set to $0.5+\sum_{j\neq i}|\Omega_{ij}^{(k)}(\alpha)|$, ensuring positive definiteness. We then rescale each source distribution so that every variable has variance one. Thus, when $\alpha=0$, the source distributions are identical. Increasing $\alpha$ introduces greater magnitude variation and more sign differences across sources, while preserving the common graph support.

The GQ methods in \cite{Vinci2019GraphQG,Chang2023NonparanormalGQ} involve three or more tuning parameters. We consider tuning these either via cross-validation or via an oracle that minimizes the projected-SHD. For pairs of nodes that are never jointly observed, GQ returns a candidate superset $\widehat U$ rather than a uniquely identified edge set. We therefore give GQ an additional oracle selection within this superset: for every $H\subseteq\widehat U$, we form the graph with edge set $\widehat E_{\mathrm{obs}}\cup H$, where $\widehat E_{\mathrm{obs}}$ is the edge set returned by GQ among co-observed nodes, and retain the choice that minimizes the projected-SHD. 

Figure~\ref{fig:gq-heterogeneity} shows that GQ performs well, particularly with oracle tuning, when the source distributions are nearly identical ($\alpha\approx 0$). Its error nevertheless increases rapidly as the source distributions diverge, even after oracle selection within its returned superset. In contrast, the proposed method remains stable because it assumes only a common graph and allows source-specific precision matrices.

\begin{figure}[t!]
    \centering
    \includegraphics[width=1.1\linewidth]{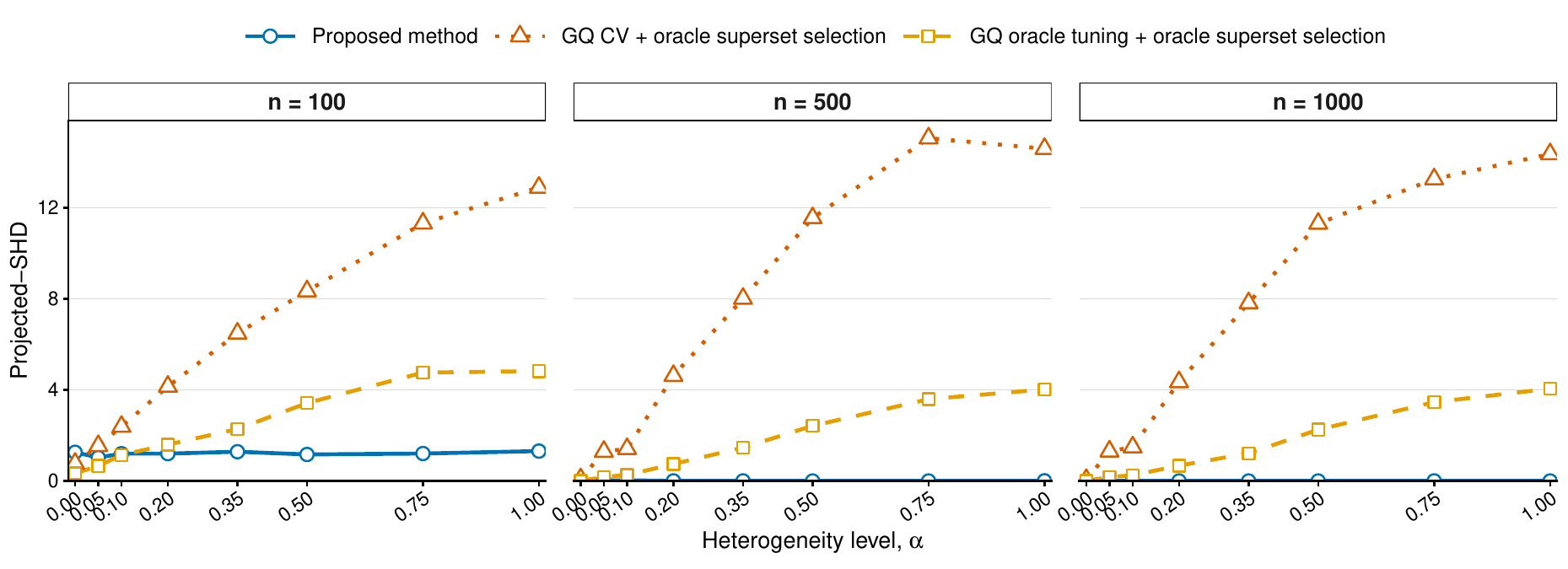}
    \caption{\small Projected-SHD across 50 replications for different heterogeneity levels of our approach compared with graph quilting (GQ). At $\alpha=0$, all sources share a common precision matrix; For $\alpha>0$, each source has a source-specific precision matrix, where the degree of heterogeneity increases with larger $\alpha$. The tuning parameters of GQ are selected either by cross-validation or via an oracle.}
    \label{fig:gq-heterogeneity}
\end{figure}

\section{Real data analysis} \label{sec:realdata}
NHANES \citep{NCHS2024NHANES} is a nationally representative study of the U.S. population that combines interviews, physical examinations, and laboratory measurements. We use NHANES fatty-acid data to study conditional relationships among circulating fatty acids, whose concentrations are strongly interdependent because they share dietary sources, transport mechanisms, and metabolic pathways \citep{Hodson2008,Calder2015}. A conditional dependence graph can help distinguish direct dependences from those explained by the broader fatty-acid profile.

Fatty-acid measurements were available from four NHANES cycles, with panel sizes $|S_1|=24, |S_2|=|S_3|=30$ and $|S_4| = 21$. The first three measured absolute fatty-acid concentrations in plasma or serum in 2003--2004, 2011--2012, and 2013--2014, whereas the fourth measured relative fatty-acid composition in red blood cells in 2021--2023. Together, the partially overlapping panels covered $p=31$ distinct fatty acids. The 2011–2012 and 2013–2014 serum panels were identical. Of the 31 distinct fatty acids, 21 were measured in all four cycles, two in the first three cycles only, seven in the two serum cycles only, and one in 2003–2004 cycle only. After restricting to nonpregnant participants aged 20 years or older with complete analysis variables, the respective sample sizes were $n_1=1{,}333$, $n_2=1{,}215$, $n_3=1{,}545$, and $n_4=5{,}115$. Although the measurement platforms differ, both circulating plasma/serum and red-cell fatty acids reflect systemic fatty-acid status, and controlled dietary studies have found broadly similar changes across these compartments \mbox{\citep{Skeaff2006,Harris2010,Hodson2014}}. We therefore use the four sources to learn a common conditional-independence graph.

Prior work has typically analyzed selected fatty acids, predefined biological groups, or low-dimensional summaries rather than estimating a joint conditional-dependence structure. Network-based metabolomics studies have recovered biologically plausible fatty-acid relationships \citep{Krumsiek2011GGM}, while NHANES studies have focused on selected fatty acids or principal-component summaries in relation to specific outcomes \citep{Yu2018DNL,Petersen2020Essential,Murphy2021Omega3,Chen2024Depression}. To our knowledge, no previous NHANES analysis has learned a fatty-acid conditional-independence graph by integrating its partially overlapping measurement panels.

NHANES provides survey weights tailored to each examination or laboratory subsample to account for unequal selection probabilities and survey nonresponse. For sources 1--3, we regressed each log-transformed fatty-acid concentration on age, modeled using a restricted cubic spline with knots at the weighted 5th, 35th, 65th, and 95th percentiles, sex, race and ethnicity, and fasting duration. We centered and standardized the residuals using the same weights and retained these weights in the graph-estimation loss.

Source 4 requires additional care because its fatty-acid measurements are compositional, reported as percentages constrained to sum to one. Since directly incorporating an appropriate compositional graph-learning procedure into our joint optimization is less straightforward, we first estimate the source-4 graph separately and then incorporate this estimated graph $\widehat{\G}_{S_4}^{(4)}$ into the integration procedure. Specifically, we use a survey-weighted adaptation of CARE \citep{Zhang2025CARE}, after adjusting for the same covariates in centered log-ratio space using the NHANES phlebotomy weights. Further details are provided in Appendix~\ref{sec:CARE_adaptation}.

Letting $\mathcal{W}_1 = \{1,2,3\}$ and $\mathcal{W}_2 = \{4\}$, we obtain estimator \eqref{eqn:main} using a survey-weighted version of the loss in \eqref{eqn:convex_quad}, replacing the usual sample covariance matrix with a weighted covariance matrix in which each observation is weighted by its NHANES survey weight. We set $\gamma = \operatorname{median}_{k\in\mathcal W_1}
\left\{2\log(n_k)/n_k\right\}
= 0.0107$.

The MIP was run for 10 hours and returned an 80-edge incumbent graph $\widehat{\G}$ with a relative optimality gap of 15.8\%. Conditional on its fitted source projections, the equivalence class $\mathrm{EC}(\widehat{\G})$ contains two graphs, which differ only in the inclusion of a single edge. Within its equivalence class all 80 selected edges were always present (AP), 384 unselected pairs were always absent (AA), and one unselected pair was ambiguous. To assess the sensitivity of these classifications to sampling variability, we generated 100 replicates by independently subsampling 80\% of the participants in each of sources 1--3 while holding the fourth-source graph fixed. For tractability, each replicate used multistart coordinate descent, including the full-data MIP graph as a warm start, without rerunning the MIP, followed by a complete equivalence-class analysis. An edge was considered stable-AP or stable-AA if it received that classification in at least 90\% of the replicates. Figure~\ref{fig:nhanes_network} presents the full-data and stability classifications jointly.
\begin{figure}[t!]
    \centering
    \includegraphics[width=1.1\linewidth]{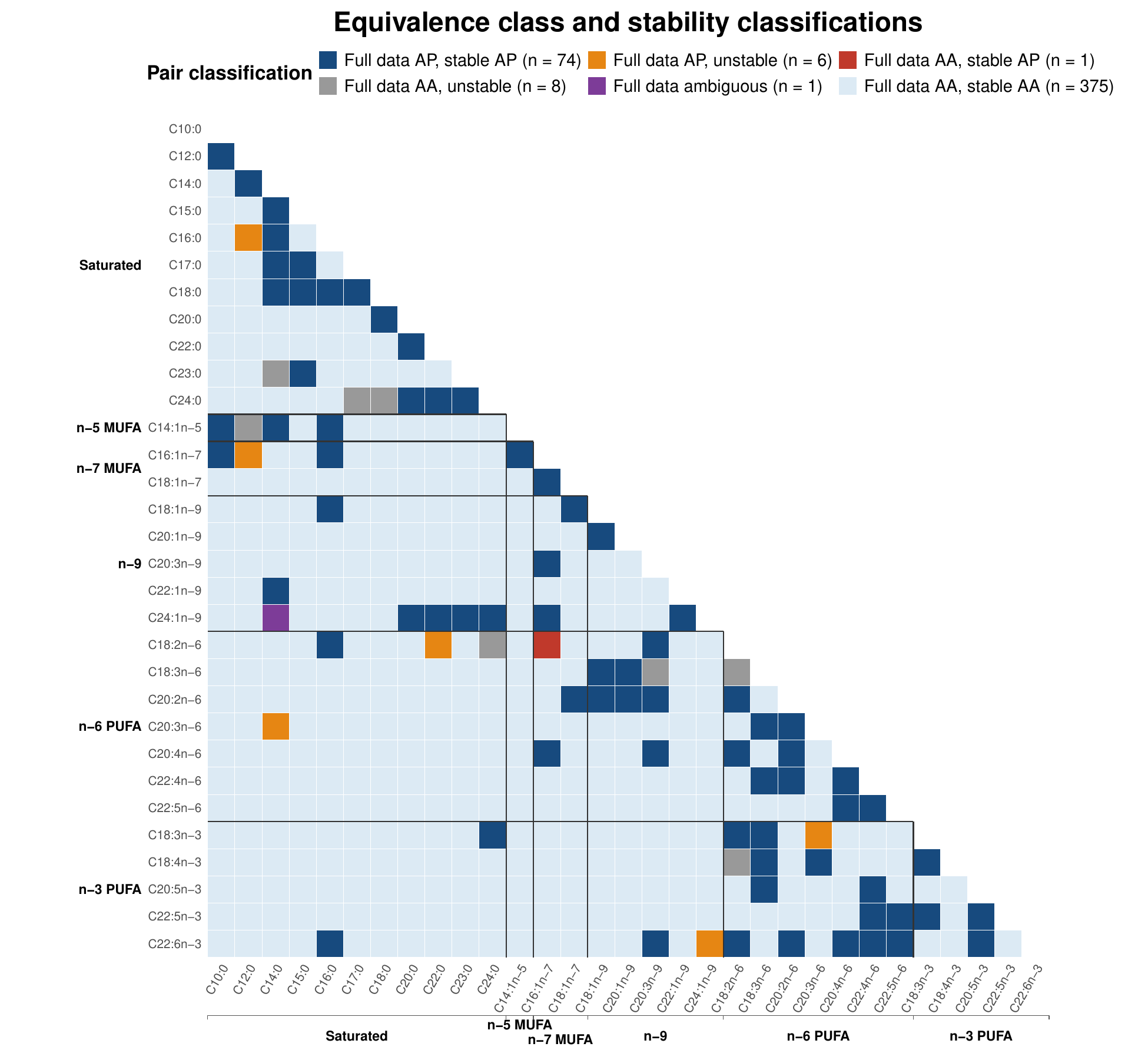}
    \caption{\small Full-data equivalence and subsampling-stability classifications for all 465 pairs of the 31 fatty acids in the integrated analysis of the four NHANES cycle data. Fatty acids are ordered by family and chain length, with black boundaries separating families.}
    \label{fig:nhanes_network}
\end{figure}

The fatty acids were grouped into saturated, $n$-5/$n$-7 Monounsaturated Fatty Acids (MUFAs), $n$-9 fatty acids, and $n$-6/$n$-3 Polyunsaturated Fatty Acids (PUFAs). The network showed clear organization by family and chain length: 34.0\% of within-family pairs were connected, compared with 11.1\% of between-family pairs, with particularly strong connectivity within the $n$-6 and $n$-3 PUFAs and between these two families. Connectivity was also higher among fatty acids with the same chain length or differing by two carbons than among more distant pairs. This pattern is consistent with fatty-acid metabolism, where desaturation preserves chain length, elongation typically adds two carbons, and $n$-3 and $n$-6 fatty acids share desaturase and elongase pathways \citep{Schaeffer2006FADS,Guillou2010Elongases}.

Eicosadienoic acid (C20:2$n$-6) had the highest degree, with nine stable connections, followed by DHA (C22:6$n$-3), with eight. C20:2$n$-6 is an intermediate in long-chain PUFA synthesis \citep{Park2009AlternativeFADS2}, whereas DHA is a terminal product of the $n$-3 pathway and an important membrane component \citep{Sprecher1995PUFA}. Notably, DHA was connected to five $n$-6 fatty acids, and Mead acid to both $n$-3 and $n$-6 fatty acids, highlighting cross-family relationships that may reflect shared dietary, transport, or FADS/ELOVL pathways. Lauric and erucic acids had the lowest degrees, with two connections each.

Several established reaction pairs, including stearic acid--oleic acid, dihomo-gamma-linolenic acid--arachidonic acid, and C20:1$n$-9--C22:1$n$-9, were consistently classified as AA \citep{Schaeffer2006FADS,Paton2009SCD}. Biochemical reactions serve here as biological annotations rather than as ground-truth conditional edges. Their stable absence indicates that these pairs did not retain distinct dependences after conditioning on the other measured fatty acids, not that the underlying reactions are absent. Contrasting stable AP and AA relationships may therefore help distinguish pathway relationships that structure population-level variation from those represented indirectly through the surrounding network.

These interpretations are subject to estimation uncertainty and incomplete measurement. After adjustment for the specified covariates, the analysis estimates the conditional-dependence structure of the marginal distribution of the 31 measured fatty acids. Other unmeasured variables may therefore induce or obscure dependences among them. Finite-sample and laboratory measurement variability may also affect edge selection, although the subsampling analysis assesses sensitivity to sampling variation. Finally, the MIP returned an incumbent rather than a certified globally optimal solution, so the full-data equivalence classifications are conditional on the fitted projections of that incumbent.

\section{Discussion}

We developed a general framework for reconstructing a shared global graph from heterogeneous sources that observe different subsets of variables and could contribute raw data, estimated graphs, or both. The framework links each source-specific graph to a common global graph through latent projections, allowing information to be integrated coherently despite incomplete variable coverage and
data-access restrictions. Our theoretical and empirical results showed that this integration recovered the global projection-equivalence class and could substantially improve the estimation of source-specific graph structure relative to analyzing each source separately. 

Our work opens several directions for future research. First, we assumed that all sources arise from a common global graph. This assumption could be relaxed by allowing the underlying graphs to differ slightly across sources while encouraging them to share substantial structure, building on ideas from \cite{danaher2014joint}. Second, when a source does not observe a hub node of the global graph, its latent projection can become substantially denser, potentially degrading both statistical and computational performance. In such settings, one could instead use a sparse-plus-low-rank decomposition \citep{LSR}. The SHD term would then compare the graph encoded by the estimated sparse component with the induced subgraph of a candidate global graph, rather than with its latent projection. Third, in federated settings, it would be interesting to consider an iterative procedure over several rounds. The central estimator would send each source its current projected graph estimate. Sources would refine these estimates using local data and return the updated graphs for reintegration.	 Finally, while our exact MIP formulations allow certification of global optimality, they do not scale to very large graphs. Tailored cutting-plane methods that exploit the latent-projection constraints could improve scalability.

\section{Acknowledgements}
We acknowledge funding from NIH grants RF1-AG090462 and U01 HG013177 (Nir Keret and Ali Shojaie) and NSF grant DMS-2413074 (Armeen Taeb).

\appendix

\section{Proofs}
\subsection{Proof of Proposition~\ref{prop:identifiability}}
\label{sec:proof_identify}
\begin{proof}
By Assumption~\ref{ass:markov_faithfulness}, we have that $(i,j)\in E(\pi_{S_k}(\G^\star)) \iff X_i \not\perp X_j | X_{S_k\setminus\{i,j\}} \text{ under }P_{S_k}^{(k)}$. Since the pairwise conditional independences in $P_{S_k}^{(k)}$ are known, $\pi_{S_k}(\G^\star)$ is uniquely identifiable from $P_{S_k}^{(k)}$. Suppose $\G$ is indistinguishable from $\G^\star$ given pairwise conditional independences of \(P^{(k)}_{S_k}, k\in[K]\). From the earlier analysis, it must be that $\pi_{S_k}(\G)=\pi_{S_k}(\G^\star)$ for every $k\in[K]$ and thus $\G \in \mathrm{EC}(\G^\star;\{S_k\}_{k=1}^K)$. On the other hand, if $\G \in \mathrm{EC}(\G^\star;\{S_k\}_{k=1}^K)$, then $\pi_{S_k}(\G)=\pi_{S_k}(\G^\star)$ for every $k\in[K]$ and so the pairwise conditional independences from the marginal distributions cannot distinguish $\G$ from $\G^\star$. 
\end{proof}

\subsection{Proof of Proposition~\ref{prop:equiv}}
\label{sec:proof_equiv}

\begin{proof}
For each $k\in[K]$, define
\[
\mathcal E_k :=
\begin{cases}
\left\{
\displaystyle
\arg\min_{\widetilde{\G}}
\operatorname{Loss}_k(\widetilde{\G};\mathcal D_k)
=
\{\G_{S_k}^{\star(k)}\}
\right\},
& k\in\mathcal W_1,\\[1.2em]
\left\{
\widehat{\G}_{S_k}^{(k)}
=
\G_{S_k}^{\star(k)}
\right\},
& k\in\mathcal W_2.
\end{cases}
\]
By Assumption~\ref{ass:consistent} $\mathbb P(\mathcal E_k)\longrightarrow 1
\text{ for every }k\in[K]$. Since $K$ is fixed,
\[
\mathbb P\left(\bigcap_{k=1}^K \mathcal E_k\right)
\ge
1-\sum_{k=1}^K \mathbb P(\mathcal E_k^c)
\longrightarrow 1.
\]

On the event $\bigcap_{k=1}^K \mathcal E_k$, the true graph $\G^\star$
attains the minimum possible value of every source-specific term in
\eqref{eqn:main}. Hence any minimizer $\widehat{\G}$ of \eqref{eqn:main}
must satisfy
\[
\pi_{S_k}(\widehat{\G})
=
\G_{S_k}^{\star(k)}
=
\pi_{S_k}(\G^\star),
\qquad
\forall k\in[K].
\]
Therefore,
\[
\mathrm{EC}\left(
\widehat{\G};\{S_k\}_{k=1}^K
\right)
=
\mathrm{EC}\left(
\G^\star;\{S_k\}_{k=1}^K
\right).
\]
Consequently,
\[
\mathbb P\left(
\mathrm{EC}(\widehat{\G};\{S_k\}_{k=1}^K)
=
\mathrm{EC}(\G^\star;\{S_k\}_{k=1}^K)
\right)
\longrightarrow 1.
\]
\end{proof}

\subsection{Proof of Theorem~\ref{thm:main}}
\label{proof:main}
The proof relies on two propositions. We state these first, prove Theorem~\ref{thm:main}, and then prove the propositions.
\begin{proposition}Let $$\varepsilon_{n,k}:= \max_{A,B\subseteq S_k, |A|\leq d_n,|B|\leq d_n} \|[\widehat{\Sigma}_k]_{A,B}-[\Sigma^\star_k]_{A,B}\|_2.$$
Suppose for every $k\in[K]$: \[
 \omega_k\lambda_k
\geq 2
\bar{C}\omega_k\varepsilon_{n,k},
\]
and 
$$\max_k \omega_k\lambda_k/\min_k \omega_k\lambda_k \leq C.$$
Define:
\[
q_-:=
\min\left\{
r\in\mathbb{N}:
\Gamma^-(r)
>
\max_k \omega_k\lambda_k (d_n+1/2)
\right\}.
\]
Then,
\[
|\mathcal{A}(\widehat{\G},\G^\star)|
\leq
(1+C(2d_n+1))
q_-.
\]
\label{prop:helper_1}
\end{proposition}

\begin{proposition}\label{prop:helper_2}Suppose $|\mathcal{A}(\widehat{\G},\G^\star)|\leq q$. Suppose $\lambda_k$ satisfies the condition described in Proposition~\ref{prop:helper_1} and that additionally, for every $k\in[K]$ and $i\in S_k$ 
$$(d_n+1/2) q\max_{j}\omega_j\lambda_j\leq \delta_{ki}(q)/2.$$
Suppose also that
\begin{eqnarray*}
\begin{aligned}
\varepsilon_{n,k} &\leq \kappa_\text{min}/2. 
\end{aligned}
\end{eqnarray*}
Then, $\pi_{S_k}(\widehat{\G})=\pi_{S_k}(\G^\star)$ for every $k\in[K]$. 
\end{proposition}

\begin{proof}Define:
$$t_k = \kappa_\text{max}\sqrt{\frac{256(d_n+1)\log(2Kp)}{n_k}}.$$
For any $A,B \subseteq S_k$, $|A|,|B| \leq d_n$, writing $U = A \cup B$, we have:
$$\|[\widehat{\Sigma}_k]_{A,B} - [\Sigma^\star_k]_{A,B}\|_2 \leq \|[\widehat{\Sigma}_k]_{U,U} - [\Sigma^\star_k]_{U,U}\|_2, \quad |U|\leq 2d_n.$$
Since $t_k \leq 8\kappa_\text{max}$ and $n_k \geq 128d_n\kappa_\text{max}^2/t_k^2$, standard Gaussian concentration results mean that:
\begin{eqnarray*}
\begin{aligned}
\mathbb{P}\left(\|[\widehat{\Sigma}_k]_{U,U} - [\Sigma^\star_k]_{U,U}\|_2 > t_k \right) \leq 2\mathrm{exp}\left(-\frac{n_kt_k^2}{128\kappa_\text{max}^2}\right) = 2(2Kp)^{-2(d_n+1)}.
\end{aligned}
\end{eqnarray*}
Since the number of subsets $U \subseteq S_k$ of size at most $2d_n$ is bounded by $p^{2d_n+1}$, a union bound gives:
$$\mathbb{P}(\exists k: \varepsilon_{n,k}>t_k) \leq 2Kp^{2d_n+1}(2Kp)^{-2(d_n+1)}\leq \frac{1}{p}.$$
Thus, with probability at least $1-1/p$, simultaneously for all $k$, $\varepsilon_{n,k} \leq t_k \leq \kappa_\text{min}/2$. On this event, since $\lambda_k \geq 2\bar{C}t_k$, we have for every $k\in[K]$:
$$\omega_k\lambda_k \geq 2\bar{C}\omega_kt_k \geq 2\bar{C}\omega_k\varepsilon_{n,k}.$$
Moreover, by assumption
$$\frac{\max_k\omega_k\lambda_k}{\min_k\omega_k\lambda_k} \leq C.$$
Hence, Proposition~\ref{prop:helper_1} applies and yields $|\mathcal{A}(\widehat{\G},\G^\star)|\leq q_\text{min}$. Since $\max_j\omega_j\lambda_j \leq C\omega_k\lambda_k$, we have that:
$$\max_j\omega_jt_j \leq C\omega_kt_k\quad \text{ for every }k\in[K].$$
So we have for every $k\in[K]$
$$(d_n+1/2)q_\text{min}\max_j\omega_j\lambda_j \leq 2C\bar{C}(d_n+1/2)q_\text{min}\omega_kt_k \leq \frac{\delta_{k,\text{min}}}{2}.$$
Here, the last inequality comes from the sample size requirement. Thus, all the conditions of Proposition~\ref{prop:helper_2} are satisfied on the event $\varepsilon_{n,k} \leq t_k$ for all $k$. 

\end{proof}

\subsubsection{Proof of Propopsition~\ref{prop:helper_1}}
The proof relies on the following lemma. Let $\widehat{\mathcal{L}}$ be the empirical version of $\mathcal{L}$. We let:
$$\mathcal{A}(\G,\G^\star):= \{(k,j): N(j,\pi_{S_k}(\G))\neq N(j,\pi_{S_k}(\G^\star))\}.$$
We also let $\mathcal{A}^+(\G,\G^\star) = \mathcal{A}(\G,\G^\star)\setminus \mathcal{A}^-(\G,\G^\star)$. 

\begin{lemma}\label{lemma:intermed} We have:
$$\max_{\G\in\mathfrak{G}}\max_{k\in[K],j\in S_k}\widehat{\mathcal{L}}_{kj}(\G)-\mathcal{L}_{kj}(\G) \leq \max_k\bar{C}\omega_k\varepsilon_{n,k}$$
and
\[
\max_{\G\in\mathfrak{G}}
\frac{
[\widehat{\mathcal{L}}{}(\G)-\mathcal{L}{}(\G)]
-
[\widehat{\mathcal{L}}{}(\G^\star)-\mathcal{L}{}(\G^\star)]
}{
|\mathcal{A}(\G,\G^\star)|
}
\leq
\max_k \bar{C}\omega_k\varepsilon_{n,k}.
\]
\end{lemma}
\begin{proof}[Proof of Lemma~\ref{lemma:intermed}]
Note that:
$$\widehat{\mathcal{L}}_{kj}(\G)-\mathcal{L}_{kj}(\G):= \omega_k\left[[\widehat{\Sigma}_k]_{jj}-[\Sigma^\star_k]_{jj} - [\widehat{\Sigma}_k]_{jT}([\widehat{\Sigma}_k]_{TT})^{-1}[\widehat{\Sigma}_k]_{Tj}+[{\Sigma}_k^\star]_{jT}([{\Sigma}_k^\star]_{TT})^{-1}[{\Sigma}^\star_k]_{Tj}\right].$$
Then it follows that:
\begin{eqnarray*}
\begin{aligned}
\max_{\G\in \mathfrak{G}} \frac{|[\widehat{\mathcal{L}}(\G)-\widehat{\mathcal{L}}(\G^\star)]-[\mathcal{L}(\G)-\mathcal{L}(\G^\star)]|}{\mathcal{A}(\G,\G^\star)} &= 
\max_{\G\in\mathfrak{G}} \frac{\left|\sum_{(k,j)\in\mathcal{A}(\G,\G^\star)} \left[ \widehat{\mathcal{L}}_{kj}(\G)-\mathcal{L}_{kj}(\G)-\widehat{\mathcal{L}}_{kj}(\G^\star)+\mathcal{L}_{kj}(\G^\star)\right]\right|}{|\mathcal A(\G,\G^\star)|}\\
&\le 2\max_{\substack{\G\in\mathfrak{G}\\k\in[K],\,j\in S_k}} \left|\widehat{\mathcal{L}}_{kj}(\G)-\mathcal{L}_{kj}(\G)\right|\\
&\le 2\left[\max_{k\in[K],j\in S_k}\omega_k\left|[\widehat{\Sigma}_k]_{jj}-[\Sigma^\star_k]_{jj}\right|\right]\\&+2\left[\max_{k\in[K],j\in S_k}\max_{\substack{T\subseteq S_k\setminus\{j\}\\|T|\leq d_n}} \omega_k\left|[\widehat{\Sigma}_k]_{jT}([\widehat{\Sigma}_k]_{TT})^{-1}[\widehat{\Sigma}_k]_{Tj}-[{\Sigma}_k^\star]_{jT}([{\Sigma}_k^\star]_{TT})^{-1}[{\Sigma}_k^\star]_{Tj}\right|\right]\\
&\leq \max_k \bar{C}\omega_k\varepsilon_{n,k}.
\end{aligned}
\end{eqnarray*}
The last inequality follows from $\varepsilon_{k,n} \leq \kappa_{\text{min}}/2$ and $\bar{C}:= 4+\frac{8\kappa_{\text{max}}}{\kappa_{\text{min}}}+\frac{4\kappa_{\text{max}}^2}{\kappa_{\text{min}}^2}$. 
\end{proof}
\begin{proof}[Proposition~\ref{prop:helper_1}]
We have from the optimality condition that:
$$\widehat{\mathcal{L}}(\widehat{\G})+2\sum_{k}\omega_k\lambda_k|E(\pi_{S_k}(\widehat{\G}))| \leq \widehat{\mathcal{L}}(\G^\star)+2\sum_{k}\omega_k\lambda_k|E(\pi_{S_k}({\G}^\star))|.$$
The optimality condition above implies:
\[
\begin{aligned}
0\geq\;&
\widehat{\mathcal{L}}(\widehat{\G})
-
\widehat{\mathcal{L}}{}(\G^\star)-\sum_{(k,j)\in\mathcal{A}^-(\widehat{\G},\G^\star)}\omega_k\lambda_kd_n+ 
\sum_{(k,j)\in\mathcal{A}^+(\widehat{\G},\G^\star)}
\omega_k\lambda_k.
\end{aligned}
\]
{\noindent}Therefore, by Lemma~\ref{lemma:intermed},
\begin{equation}
\begin{aligned}
\mathcal{L}{}(\widehat{\G})
-
\mathcal{L}{}(\G^\star)
&-\sum_{(k,j)\in\mathcal{A}^-(\widehat{\G},\G^\star)}\left[\omega_k\lambda_k d_n
+
\bar{C}\omega_k\varepsilon_{n,k}
\right]
\\
&+
\sum_{(k,j)\in\mathcal{A}^+(\widehat{\G},\G^\star)}\left[
\omega_k\lambda_k
-
\bar{C}\omega_k\varepsilon_{n,k}
\right]
\leq 0.
\end{aligned}
\label{eqn:temp_first}
\end{equation}
Since
$$\omega_k\lambda_k \geq 2\bar{C}\omega_k\varepsilon_{n,k},$$
we have that:
\begin{eqnarray}
\begin{aligned}
0 &\geq \mathcal{L}(\widehat{\G})-\mathcal{L}(\G^\star)- |\mathcal{A}^-(\widehat{\G},\G^\star)|\max_k \omega_k\lambda_k(d_n+1/2)
+ \frac{1}{2}|\mathcal{A}^+(\widehat{\G},\G^\star)|\min_k \omega_k\lambda_k.
\end{aligned}
\label{tem:inequality}
\end{eqnarray}

\noindent Take any $r$ satisfying

\[
\Gamma^-(r)
>
\max_k \omega_k\lambda_k (d_n+1/2).
\]
{\noindent}Suppose, as a point of contradiction, that

\[
|\mathcal{A}^-(\widehat{\G},\G^\star)|\geq r.
\]

{\noindent}Then,

\[
\frac{
\mathcal{L}{}(\widehat{\G})-\mathcal{L}{}(\G^\star)
}{
|\mathcal{A}^-(\widehat{\G},\G^\star)|
}
\geq
\Gamma^-(r)
>
\max_k \omega_k\lambda_k (d_n+1/2),
\]
implying that
$$\mathcal{L}{}(\widehat{\G})-\mathcal{L}{}(\G^\star) > |\mathcal{A}^-(\widehat{\G},\G^\star)|\max_k \omega_k\lambda_k (d_n+1/2).$$
This however contradicts the inequality  \eqref{tem:inequality}. Hence, defining
\[
q_-:=
\min\left\{
r\in\mathbb{N}:
\Gamma^-(r)
>
\max_k \omega_k\lambda_k (d_n+1/2)
\right\},
\]

{\noindent}we have

\[
|\mathcal{A}^-(\widehat{\G},\G^\star)|
\leq q_-.
\]

{\noindent}Moreover, since $\mathcal{L}{}(\widehat{\G})
-
\mathcal{L}{}(\G^\star)
\geq 0$, inequality \eqref{eqn:temp_first} implies:

\[
\frac{1}{2}
\min_k \omega_k\lambda_k
|\mathcal{A}^+(\widehat{\G},\G^\star)|
\leq
\max_k \omega_k\lambda_k (d_n
+1/2)
|\mathcal{A}^-(\widehat{\G},\G^\star)|.
\]
Therefore,
\[
|\mathcal{A}^+(\widehat{\G},\G^\star)|
\leq
(2d_n+1)\frac{
\max_k \omega_k\lambda_k 
}{
\min_k \omega_k\lambda_k
}
q_- \leq C(2d_n+1)q_-.
\]

{\noindent}Hence,

\[
|\mathcal{A}(\widehat{\G},\G^\star)|
\leq
(1+C(2d_n+1))
q_-.
\]

\end{proof}

\begin{proof}[Proof of Proposition~\ref{prop:helper_2}]

Let $T_{kj} := N(j,\pi_{S_k}(\G))$, $T_{kj}^\star:= N(j,\pi_{S_k}(\G^\star))$, $M_{kj} = T_{kj}^\star \setminus T_{kj}$ and $\Sigma^{(k)}_{M_{kj}|T_{kj}}$ be the conditional covariance of $X_{M_{kj}}^{(k)}$ conditioned on $X_{T_{kj}}^{(k)}$.

Suppose for any $k\in [K]$ and $i\in S_k$, $N(i,\pi_{S_k}({\G}^\star))\not\subseteq N(i,\pi_{S_k}(\widehat{\G}))$. Then, we have by Lemma~\ref{lemma:intermed}
\begin{eqnarray*}
\begin{aligned}
\delta_{ki}(q) &\leq \mathcal{L}(\widehat{\G})-\mathcal{L}(\G^\star) \\
&\leq \widehat{\mathcal{L}}(\widehat{\G})-\widehat{\mathcal{L}}(\G^\star) + \sum_{(l,m)\in\mathcal{A}(\widehat{\G},\G^\star)}\bar{C}\omega_l\varepsilon_{n,l} \\
&= [\widehat{\mathcal{L}}(\widehat{\G})+2\sum_{k}\lambda_k\omega_k|E(\pi_{S_k}(\widehat{\G}))|]-[\widehat{\mathcal{L}}(\G^\star) + 2\sum_{k}\lambda_k\omega_k|E(\pi_{S_k}({\G}^\star))|]\\ &- 2\sum_{k}\lambda_k\omega_k(|E(\pi_{S_k}(\widehat{\G}))|-|E(\pi_{S_k}({\G}^\star))|) +\sum_{(l,m)\in\mathcal{A}(\widehat{\G},\G^\star)}\bar{C}\omega_l\varepsilon_{n,l}\\
&\leq \sum_{(l,m)\in\mathcal{A}^-(\widehat{\G},\G^\star)}d_n\omega_l\lambda_l - \sum_{(l,m)\in\mathcal{A}^+(\widehat{\G},\G^\star)} \omega_l\lambda_l + \sum_{(l,m)\in\mathcal{A}(\widehat{\G},\G^\star)}\bar{C}\omega_l\varepsilon_{n,l}.
\end{aligned}
\end{eqnarray*}
The second inequality follows from Lemma~\ref{lemma:intermed}, and the final inequality follows from the optimality of $\widehat{\G}$. Since, by the condition of Proposition~\ref{prop:helper_1}:
$$\omega_l\lambda_l \geq 2\bar{C}\omega_l\varepsilon_{n,l},$$
we have from the ineqality above:
$$\delta_{ki}(q) \leq (d_n+1/2)q\max_l\omega_l\lambda_l.$$
This however contradicts the assumption of Proposition~\ref{prop:helper_2}. So we conclude that $\pi_{S_k}(\G^\star)$ is a subgraph of $\pi_{S_k}(\widehat{\G})$.

Finally, we need to argue that $\pi_{S_k}(\widehat{\G})=\pi_{S_k}(\G^\star)$. We have that $\mathcal{A}^-(\widehat{\G},\G^\star)=\varnothing$. Thus, inequality~\eqref{eqn:temp_first} gives
$$
\begin{aligned}
0 \ge {}& \mathcal{L}(\widehat{\G})-\mathcal{L(\G^\star)} +\sum_{(l,m)\in\mathcal{A}^+(\widehat{\G},\G^\star)}(\omega_l\lambda_l-\bar{C}\omega_l\varepsilon_{n,l}) 
\ge  \frac{1}{2}\sum_{(l,m)\in\mathcal{A}^+(\widehat{\G},\G^\star)}\omega_l\lambda_l \ge 0,
\end{aligned}
$$
where the second inequality uses $\mathcal{L}(\widehat{\G})-\mathcal{L}(\G^\star)\ge0$ and $\omega_l\lambda_l\ge 2\bar{C}\omega_l\varepsilon_{n,l}$. Since $\omega_l\lambda_l>0$ for every $l$, this implies $\mathcal{A}^+(\widehat{\G},\G^\star) = \varnothing$. We conclude that $\pi_{S_k}(\widehat{\G})=\pi_{S_k}(\G^\star)$ for every $k \in [K]$.

\end{proof}
\subsection{Proof of Proposition~\ref{prop_delta}}
\label{proof:prop_delta}
\begin{proof}
Let
$$
s_{\min}:=
\kappa_{\min}\min_k\{\omega_k\beta_{k,\min}^2\}.
$$
For any \(\G\in\mathfrak G\), every pair \((l,m)\in\mathcal A^-(\G,\G^\star)\setminus\mathcal W\) omits at least one true neighbor. Hence, building on the results in Remark~\ref{remark:diff}, by the definition of $\beta_{l,\min}$ and the lower eigenvalue bound,
$$
\omega_l\left\{\mathcal L_{lm}(\G)-\mathcal L_{lm}(\G^\star)\right\}
\ge
\kappa_{\min}\omega_l\beta_{l,\min}^2
\ge s_{\min}.
$$
Since at most $|\mathcal W|$ elements of $\mathcal A^-(\G,\G^\star)$ belong to $\mathcal W$, it follows that
$$
\mathcal L(\G)-\mathcal L(\G^\star)
\ge
s_{\min}
\left(|\mathcal A^-(\G,\G^\star)|-|\mathcal W|\right)_+.
$$
Recall:
\begin{equation*}
\Gamma^-(r):=
\inf_{\G\in\mathfrak{G}:\,|\mathcal{A}^-(\G,\G^\star)|\geq r}
\frac{
\mathcal{L}(\G)-\mathcal{L}(\G^\star)
}{
|\mathcal{A}^-(\G,\G^\star)|
}.
\end{equation*}
Therefore, for every $r>|\mathcal W|$,
$$
\Gamma^-(r)
\ge
s_{\min}\left(1-\frac{|\mathcal W|}{r}\right).
$$
By the assumption of the proposition,
$$ 
s_{\min} > 2(d_n+1/2)\max_{k\in[K]}\omega_k\lambda_k.
$$
Set $r_0:=2|\mathcal W|+1$. Then
$$
\Gamma^-(r_0) \ge s_{\min}\left(1-\frac{|\mathcal{W}|}{r_0}\right) > \frac{s_{\min}}{2} > (d_n+1/2)\max_{k\in[K]}\omega_k\lambda_k.
$$
Thus, $r_0$ satisfies the threshold condition in the definition of $q_{\min}$, and therefore
$$
q_{\min} \le r_0(1+C(2d_n+1)) = (2|\mathcal W|+1) (1+C(2d_n+1)) =:q_0
$$

Since $\delta_{kj}(q)$ is nonincreasing in $q$, we have:
$$
\delta_{k,\min}
=
\min_{j\in S_k}\delta_{kj}(q_{\min})
\ge
\min_{j\in S_k}\delta_{kj}(q_0).$$
Since all loss differences $\mathcal{L}_{lm}(\G)-\mathcal{L}_{lm}(\G^\star)$ are nonnegative for $(l,m)\in\mathcal{A}^-(\G,\G^\star)$, we have:
$$
\delta_{kj}(q_0)
\ge
\inf_{\substack{
\G\in\mathfrak G,\,
|\mathcal A^-(\G,\G^\star)|\le q_0\\
N(j,\pi_{S_k}(\G^\star))
\not\subseteq
N(j,\pi_{S_k}(\G))
}}
\sum_{\substack{
(l,m)\in\mathcal A^-(\G,\G^\star)\\
(l,m)\notin\mathcal W
}}
\kappa_{\min}\omega_l\beta_{l,\min}^2.$$
Taking the minimum over $j\in S_k$ yields the claimed result.   

\end{proof}

\section{Perspective strengthening} \label{sec:perspective}
 Let $\delta^k \in \mathbb{R}^{|S_k|}_{++}$ be a strictly positive vector with the property that $\widehat{\Sigma}_k-D_{\delta^k} \succeq 0$ where $D_{\delta^k} = \mathrm{diag}(\delta^k_1,\dots,\delta^k_{|S_k|})$. By splitting the quadratic term, we have $(I-\Theta^{(k)})^\top \widehat{\Sigma}_k(I-\Theta^{(k)}) = (I-\Theta^{(k)})^\top Q_k(I-\Theta^{(k)}) + (I-\Theta^{(k)})^\top D_{\delta^k}(I-\Theta^{(k)})$ where $Q_k :=  \widehat{\Sigma}_k- D_{\delta^k} \succeq 0$. Note we can replace $\mathrm{tr}((I-\Theta^{(k)})^\top D_{\delta^k}(I-\Theta^{(k)}))$ with $\sum_{\substack{i,j\in S_k \\ i \neq j}}(\Theta_{ij}^{(k)})^2\delta^k_i + \sum_{i\in S_k}\delta^k_i$. We can then obtain the following equivalent formulation:
\begin{eqnarray}
\hspace*{0.2in}
\begin{aligned}
\argmin_{\substack{(g,h,f,m)\in\mathcal{F}(\{S_k\}_{k=1}^K)\\\Theta^{(k)} \in \mathbb{R}^{|S_k|\times|S_k|} ~\forall k\in[K]\\\Theta^{(k)}_{ii}=0 ~\forall k\in[K]i\in S_k}}&\quad\quad\quad  \sum_{k =1}^K \omega_k\Big[ \mathrm{tr}((I-\Theta^{(k)})^\top Q_k(I-\Theta^{(k)})) +\\[-6ex]&\quad\quad \sum_{i,j \in S_k, i\neq j}\frac{(\Theta^{(k)}_{ij})^2\delta_i^k}{h_{\min\{i,j\},\max\{i,j\}}^{(k)}}+ 2\lambda_{k}\sum_{i,j\in S_k,i<j}h_{ij}^{(k)}\Big],\\
\text{subject-to:}
&\quad\quad\quad (g,h,f,m) \text{ satisfy the linear constraints in } \eqref{eqn:flow_constraints} \text{ and }\eqref{eqn:hidden_component} ~\forall k\in[K].
\end{aligned}
\label{eqn:perspective_reformulation}
\end{eqnarray}
We have omitted the graph-independent additive constant $\sum_{k=1}^K\omega_k\sum_{i\in S_k}\delta_i^k$, which does not affect the optimizers. Here, in the ratios of the form ${(\Theta^{(k)}_{ij})^2\delta_i^k}/{h_{ij}^{(k)}}$, we adopt the convention that $0/0 = 0$. Since the perspective of a convex quadratic is convex, each term ${(\Theta^{(k)}_{ij})^2\delta_i^k}/{h_{\min\{i,j\},\max\{i,j\}}^{(k)}}$ is jointly convex over $\Theta^{(k)}$ and $h^{(k)}$. Thus, \eqref{eqn:perspective_reformulation} is a convex mixed-integer program.

To establish that \eqref{eqn:perspective_reformulation} and \eqref{eqn:reformulation_main} have the same optimizers, first note that $h_{ij}^{(k)} = 0$ necessarily implies $\Theta^{(k)}_{ij} = \Theta^{(k)}_{ji}=0$ at the optimum of \eqref{eqn:perspective_reformulation}. Moreover, since $h_{ij}^{(k)} \in \{0,1\}$, for such feasible $h^{(k)}$ and $\Theta^{(k)}$, we have that for $h^{(k)}$ and $\Theta^{(k)}$ satisfying $h_{ij}^{(k)}=0 \Rightarrow \Theta^{(k)}_{ij} = \Theta^{(k)}_{ji}=0$: ${(\Theta^{(k)}_{ij})^2\delta_i^k}/{h_{ij}^{(k)}} = {(\Theta^{(k)}_{ij})^2\delta_i^k}$. Hence, the two formulations are equivalent. However, the continuous relaxation of \eqref{eqn:perspective_reformulation}, obtained by replacing the binary domains of $g$ and $h$, yields a tighter lower bound. Indeed, for $h_{ij}^{(k)} \in (0,1]$, ${(\Theta^{(k)}_{ij})^2\delta_i^k}/{h_{ij}^{(k)}} \geq {(\Theta^{(k)}_{ij})^2\delta_i^k}$. {Since larger $\delta$ results in a tighter lower bound, we choose $\delta$ by maximizing $\sum_{i}\delta_i^{k}$ subject to $\widehat{\Sigma}_k-D_{\delta}\succeq 0$ for $\delta_i^{k} > \epsilon, i=1,\dots,|S_k|$, for some small $\epsilon>0$.} This is a convex semidefinite program that can be solved efficiently.

To reduce the number of continuous variables, we can project out the continuous variables $\Theta^{(k)}$ and obtain a convex MIP over fewer variables. Specifically, minimizing over $\Theta^{(k)},k\in[K]$, we obtain the following MIP:

\begin{eqnarray}
\hspace*{0.2in}
\begin{aligned}
\argmin_{\substack{(g,h,f,m)\in\mathcal{F}(\{S_k\}_{k=1}^K)\\t_{kj} \in \mathbb{R}~\forall k\in[K],\,j\in S_k}}&\quad\quad\quad \sum_{k=1}^{K}\omega_k\Big[\sum_{j\in S_k}t_{kj}+ 2\lambda_{k}\sum_{i,j\in S_k,i<j}h_{ij}^{(k)}\Big],\\
\text{subject-to:}
&\quad\quad\quad (g,h,f,m) \text{ satisfy the linear constraints in } \eqref{eqn:flow_constraints} \text{ and }\eqref{eqn:hidden_component} ~\forall k\in[K],\\& \quad\quad\quad t_{kj}\geq F_{kj}(h^{(k)})~\forall k\in[K],\,j\in S_k,
\end{aligned}
\label{eqn:perspective_reformulation_projection}
\end{eqnarray}
where
$$F_{kj}(h^{(k)}) :=  \left\{\left(Q_k + \mathrm{diag}_{i\in S_k}\left(\begin{cases}\dfrac{\delta_i^k}{h^{(k)}_{\min\{i,j\},\max\{i,j\}}}, & i\neq j,\\0, & i=j
\end{cases} \right)\right)^{-1}_{jj}\right\}^{-1}.$$
Since each $F_{kj}(h^{(k)})$ is obtained by partial minimization of a jointly convex function, it is convex. Thus, \eqref{eqn:perspective_reformulation_projection} is a convex MIP. We use outer approximation to solve \eqref{eqn:perspective_reformulation_projection}, a cutting-plane method that finds the optimal solution by constructing a sequence of piecewise-affine lower bounds of each $F_{kj}(h^{(k)})$ using the first-order condition for convex functions. Specifically, at each candidate $\bar h^{(k)}$, we evaluate the nodewise terms and add the supporting-hyperplane cut
$t_{kj}\ge F_{kj}(\bar h^{(k)}) +\nabla F_{kj}(\bar h^{(k)})^\top (h^{(k)}-\bar h^{(k)})$ for each violated epigraph constraint. These cuts progressively refine the mixed-integer linear relaxation within a single branch-and-bound search, which terminates when the prescribed optimality tolerance or time limit is reached.

\section{Convex MIP for Example~\ref{example:discrete}} \label{sec:discrete_MIP}
\subsection{With the loss \eqref{eqn:discrete_loss_1}}

For the loss \eqref{eqn:discrete_loss_1}, we restrict attention to chordal graphs $\G$, in which every cycle of length at least four has a chord. Since latent projection preserves chordality, $\pi_{S_k}(\G)$ is chordal whenever $\G$ is chordal. The chordal restriction results in improved tractability of solving \eqref{eqn:main} for discrete graphical models. In particular, we use the property that a graph is chordal if and only if it has a perfect elimination ordering \citep{Fulkerson1965IncidenceMA}. A perfect elimination ordering of $\G=(V,E)$ is an injective map $\sigma:V\to\mathbb{R}$ where $\sigma(i)<\sigma(j)$ means that vertex $i$ appears before vertex $j$. The ordering is a perfect elimination ordering if, for every $i\in V$, its later neighbors,
$$
N_i(\G,\sigma) := \{j\in V:\{i,j\}\in E,
\ \sigma(j)>\sigma(i)\},
$$
form a clique. Fix any such perfect elimination ordering $\sigma$.
Then,
$$
\min_{\theta^{(k)}\in\Theta_k({\widetilde{\G}})}\ell_k(\theta^{(k)};\mathcal{D}_k)
=
\sum_{i\in S_k}\widehat H\!(X_i\mid X_{
N_i(\widetilde{\G},\sigma)}),
$$
where for any $A \subseteq S_k$, $\widehat{H}\!\left(X_i\mid X_A\right)$ computes the conditional entropy of $X_i$ conditioned on $X_A$ using the empirical distribution $\widehat{P}^{(k)}_{S_k}$. Thus, \eqref{eqn:main} can be reformulated here as:

\begin{subequations}
\begin{align}
\underset{\G,\{\sigma_j^{(k)}\in[1,|S_k|]\}_{j\in{S}_k,k\in\mathcal{W}_1}}{\argmin}\quad
& \sum_{k\in\mathcal{W}_1}\omega_k
\left[
\sum_{i\in S_k}\widehat H\!(X_i\mid X_{
N_i(\pi_{S_k}(\G),\sigma^{(k)})})
+\frac{\log(n_k)}{2n_k}\mathrm{dim}(\Theta_k(\pi_{S_k}(\G))
\right]\label{eqn:objective_discrete} \\
&\quad
+\gamma\sum_{k\in\mathcal{W}_2}\omega_k
\mathrm{SHD}\!\left(\pi_{S_k}(\G),\widehat\G^{(k)}_{S_k}\right),\nonumber\\
\text{subject to}\quad
& \sigma^{(k)} \text{ induces a perfect elimination ordering of }
\pi_{S_k}(\G),
\qquad k\in\mathcal{W}_1\label{eqn:perfect_elimination}.
\end{align}
\label{eqn:discrete_learning}
\end{subequations}
Where $\mathrm{dim}(\Theta_k(\pi_{S_k}(\G))$ is described in Section~\ref{sec:disccusion_consistent}. To cast \eqref{eqn:discrete_learning} as a convex MIP, on top of reformulating latent projections via \eqref{eqn:flow_constraints} and \eqref{eqn:hidden_component}, we must formulate three combinatorial features: (i) the constraint \eqref{eqn:perfect_elimination} that $\sigma^{(k)}$ is a perfect elimination ordering of $\pi_{S_k}(\G)$, (ii) the entropy function $\widehat H\!(X_i\mid X_{
N_i(\pi_{S_k}(\G),\sigma^{(k)})})$ in the objective \eqref{eqn:objective_discrete}, and (iii) $\mathrm{dim}(\Theta_k(\pi_{S_k}(\G))$ in the objective \eqref{eqn:objective_discrete}. We first handle constraint \eqref{eqn:perfect_elimination}. Let $h^{(k)}$ be the binary variables encoding the graph $\pi_{S_k}(\G)$. Then, \eqref{eqn:perfect_elimination} can be equivalently expressed as:
\begin{equation}
\label{eq:PE-encoding}
\begin{aligned}
\text{Constraint }\eqref{eqn:perfect_elimination}
\quad\Longleftrightarrow\quad
&\exists\ \{d^{(k)}_{ij}\in\{0,1\}\}_{i< j}, \{d^{(k)}_{ji}\in[0,1]\}_{i< j} \, , \, i,j\in S_k , \,
\text{ such that :} \\
&\left\{
\begin{aligned}
&  d_{ji}^{(k)}= h_{ij}^{(k)}-d_{ij}^{(k)},
&& \forall i<j,\\
&  \sigma_j^{(k)} \geq \sigma_i^{(k)}+1-|S_k|(1-d^{(k)}_{ij}), && \forall i\neq j,\\
&  d^{(k)}_{ij}+d^{(k)}_{im}\leq 1+h_{jm}^{(k)},
&& \forall i,j,m\in S_k\text{ distinct}.
\end{aligned}
\right.
\end{aligned}
\end{equation}
The first two constraints orient every edge of $\pi_{S_k}(\G)$ according to the ordering induced by $\sigma^{(k)}$. The final constraint requires that any two later neighbors of a node be adjacent, thereby ensuring that $\sigma^{(k)}$ induces a perfect elimination ordering of $\pi_{S_k}(\G)$. To handle the entropy function $\widehat H\!(X_i\mid X_{
N_i(\pi_{S_k}(\G),\sigma^{(k)})})$ in the constraint, we can replace it with a new continuous decision variable $t_i^{(k)}$ and add the following set of constraints:
\begin{equation}
t_i^{(k)} \geq \widehat H\!(X_i\mid X_A) - U_i\left[\sum_{j\in{A}}(1-d_{ij}^{(k)}) + \sum_{j \in S_k\setminus(A\cup\{i\})}d_{ij}^{(k)}\right] \quad\forall A \subseteq S_k \setminus \{i\}
\label{eq:entropy_constraint}
\end{equation}
where $U_i \geq \max_{A \subseteq S_k \setminus \{i\}} \widehat{H}(X_i\mid X_A)-\min_{A \subseteq S_k \setminus \{i\}}\widehat{H}(X_i\mid X_A)$. For any fixed selection of the variables $d_{ij}^{(k)}$, the bracketed term is zero precisely when $A=N_i(\pi_{S_k}(\G),\sigma^{(k)})$ and is at least one otherwise; hence the corresponding constraint reduces to $t_i^{(k)}\geq \widehat H\!(X_i\mid X_{
N_i(\pi_{S_k}(\G),\sigma^{(k)}})$ for the selected neighborhood, while the remaining constraints are rendered nonbinding by the choice of $U_i$. Note that if $X_i$ could take on $|\mathcal{X}_i|$ values, $U_i = \log|\mathcal{X}_i|$ is a sufficient choice. Finally, it is straightforward to show that $\mathrm{dim}(\Theta_k(\pi_{S_k}(\G)) = \sum_{i\in S_k}(|\mathcal{X}_i| - 1)\prod_{j\in S_k\setminus\{i\}}|\mathcal{X}_j|^{d^{(k)}_{ij}}$,
which induces a convex continuous relaxation. We thus arrive at the following convex mixed integer program:
\begin{equation} \label{eqn:convex-mip} \begin{aligned} \underset{g,h,f,m,t,\sigma,d}{\operatorname{argmin}} \quad & \sum_{k\in\mathcal{W}_1}\omega_k \left[ \sum_{i\in S_k} t_i^{(k)} +\frac{\log(n_k)}{2n_k}\sum_{i\in S_k}(|\mathcal{X}_i| - 1)\prod_{j\in S_k\setminus\{i\}}|\mathcal{X}_j|^{d^{(k)}_{ij}}\right] +\gamma\sum_{k\in\mathcal{W}_2}\omega_k \sum_{\substack{i,j\in S_k\\i<j}} (1-2\widehat h_{ij}^{(k)})h_{ij}^{(k)} \\[0.5em]\text{subject-to:}\quad & (g,h,f,m)\in\mathcal{F}(\{S_k\}_{k=1}^K),\\ & t^{(k)}\in\mathbb{R}^{|S_k|},\ \sigma_j^{(k)}\in[1,|S_k|],\ d_{ij}^{(k)}\in\{0,1\},d_{ji}^{(k)}\in[0,1]\quad \substack{\forall k\in\mathcal{W}_1,\ j\in S_k,\\ i< j\in S_k},\\ & \text{linear constraints } \eqref{eq:PE-encoding}\quad \forall k\in\mathcal{W}_1,\\ & t_i^{(k)} \geq \widehat H(X_i\mid X_A) -U_iD_{i,A}^{(k)} \quad \forall i\in S_k,\ A\subseteq S_k\setminus\{i\},\ k\in\mathcal{W}_1, \end{aligned} \end{equation}
where we use the shorthand notation $D_{i,A}^{(k)}:= \sum_{j\in A}(1-d_{ij}^{(k)}) +\sum_{j\in S_k\setminus(A\cup\{i\})}d_{ij}^{(k)}$. The exponentially many constraints indexed by $A\subseteq S_k\setminus\{i\}$ can be added lazily. At an incumbent solution, $d_i^{(k)}$ determines the active set $A_i^{(k)}:=\{j:d_{ij}^{(k)}=1\}$, 
for which \(D_{i,A_i^{(k)}}^{(k)}=0\). If
$t_i^{(k)}<
\widehat H\!(X_i\mid X_{A_i^{(k)}})$, 
the corresponding constraint is added and the optimization continues.

\subsection{With the loss \eqref{eqn:discrete_loss_2}}
For any candidate neighborhood \(A\subseteq S_k\setminus\{i\}\), define the local score
\[
s_{i,A}^{(k)}
:=
\widehat H(X_i\mid X_A)
+
\frac{\log(n_k)}{2n_k}
\left(|\mathcal X_i|-1\right)
\prod_{j\in A}|\mathcal X_j|.
\]
Introduce variables \(z_{i,A}^{(k)}\) indicating whether \(A\) is selected as the neighborhood of node \(i\), with
\[
\sum_{A\subseteq S_k\setminus\{i\}}z_{i,A}^{(k)}=1,
\qquad
h_{ij}^{(k)}
=
\sum_{\substack{A\subseteq S_k\setminus\{i\}\\j\in A}}z_{i,A}^{(k)}
\]
The pseudo-likelihood contribution of source \(k\) can then be written as
\[
\sum_{i\in S_k}
\sum_{A\subseteq S_k\setminus\{i\}}
s_{i,A}^{(k)}z_{i,A}^{(k)}.
\]
The local scores $s_{i,A}^{(k)}$ can be precomputed for every candidate neighborhood $A\subseteq S_k\setminus\{i\}$. The optimization then only selects among these precomputed scores using binary variables $z_{i,A}^{(k)}$, while the latent-projection constraints enforce compatibility with a common global graph. This yields a mixed-integer linear program, and hence convex MIP. To avoid explicitly introducing all exponentially-many candidate neighborhoods, the variables and associated constraints can be generated only as needed.

\newpage
\section{Additional simulation results} \label{sec:additional_simulations}

\begin{figure}[H]
    \centering
    \includegraphics[width=1.1\linewidth]{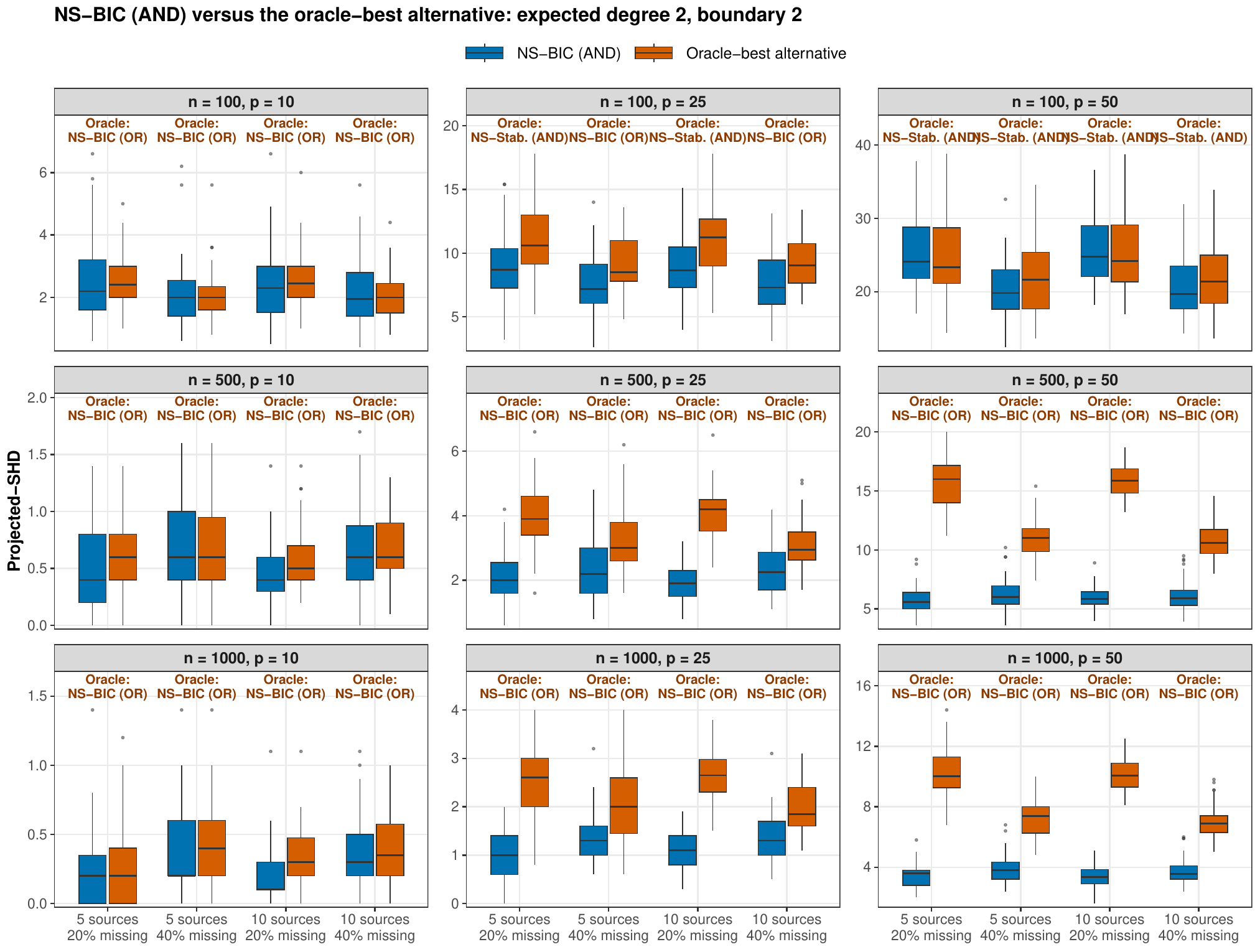}
    \caption{\small Projected SHD per source for neighborhood selection with BIC and the AND rule and the oracle-best alternative when $d=2$ and $l=2$. The oracle alternative is selected separately in each simulation setting using its mean SHD across the 50 paired replications, and its identity is reported above each group of boxplots.}
    \label{fig:baseline-oracle-d2-b2}
\end{figure}

\begin{figure}[t!]
    \centering
    \includegraphics[width=1.1\linewidth]{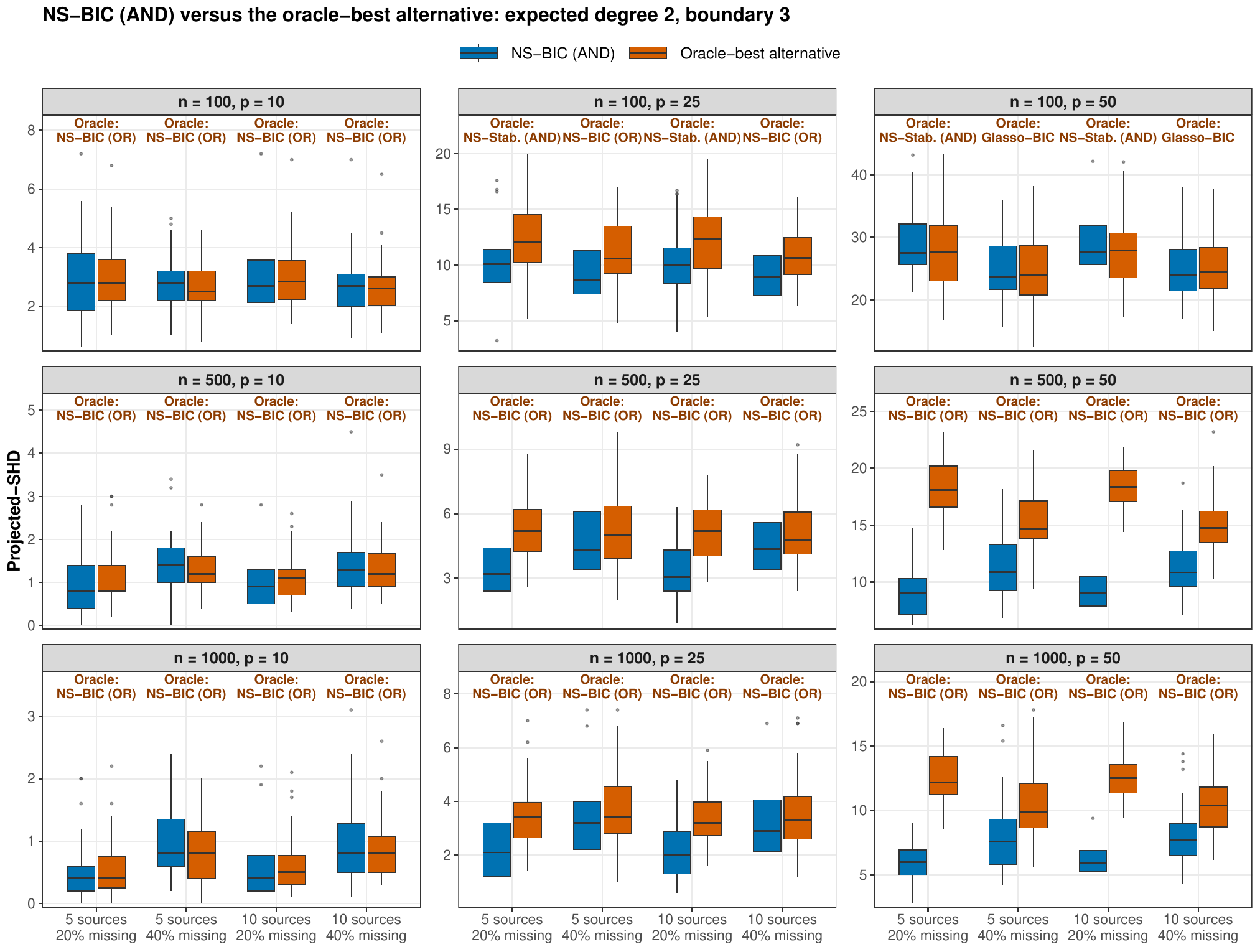}
    \caption{\small Projected SHD per source for neighborhood selection with BIC and the AND rule and the oracle-best alternative when $d=2$ and $l=3$. The oracle alternative is selected separately in each simulation setting using its mean SHD across the 50 paired replications, and its identity is reported above each group of boxplots.}
    \label{fig:baseline-oracle-d2-b3}
\end{figure}

\begin{figure}[t!]
    \centering
    \includegraphics[width=1.1\linewidth]{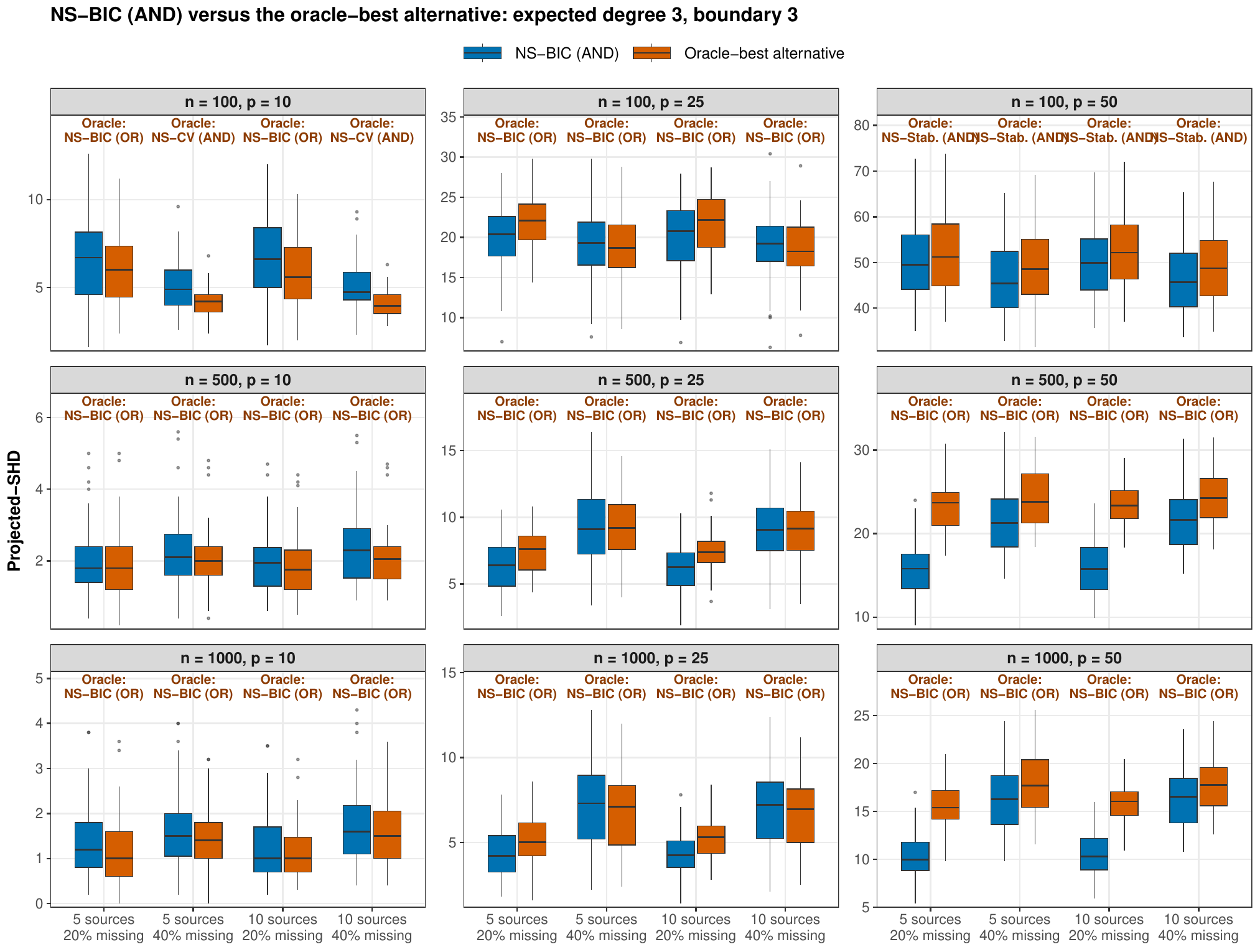}
    \caption{\small Projected SHD per source for neighborhood selection with BIC and the AND rule and the oracle-best alternative when $d=3$ and $l=3$. The oracle alternative is selected separately in each simulation setting using its mean SHD across the 50 paired replications, and its identity is reported above each group of boxplots.}
    \label{fig:baseline-oracle-d3-b3}
\end{figure}

\begin{figure} 
    \centering
    \includegraphics[width=1.1\linewidth]{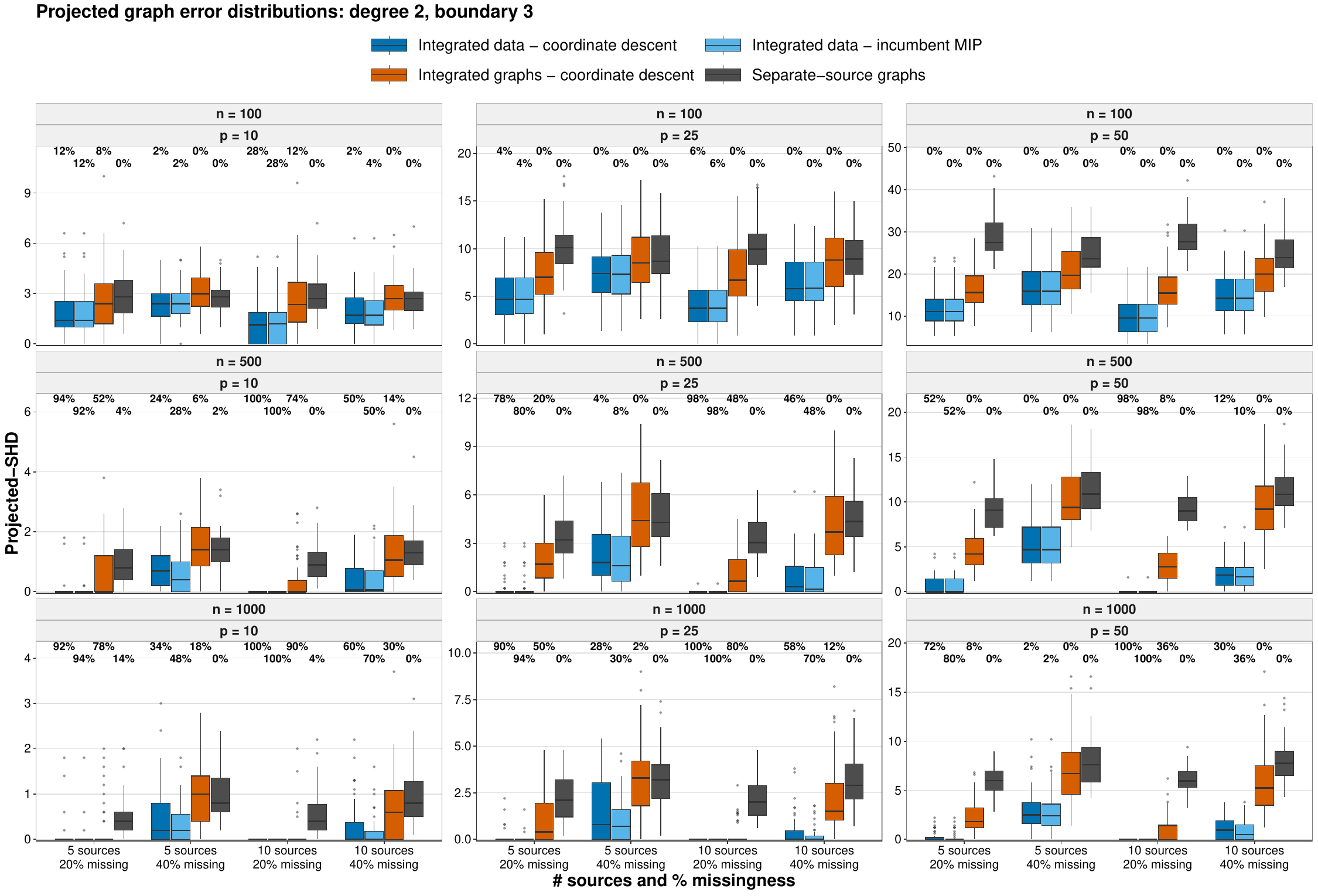}
    \caption{\small Projected SHD per source in the harder regime with expected degree $d=2$ and maximum boundary size $l=3$. Boxes summarize
    the 50 replications. Percentages at the top of each panel report exact recovery of the projection-equivalence class.}
    \label{fig:all-data-d2-b3}
\end{figure}
\FloatBarrier
\begin{figure}
    \centering
    \includegraphics[width=1.1\linewidth]{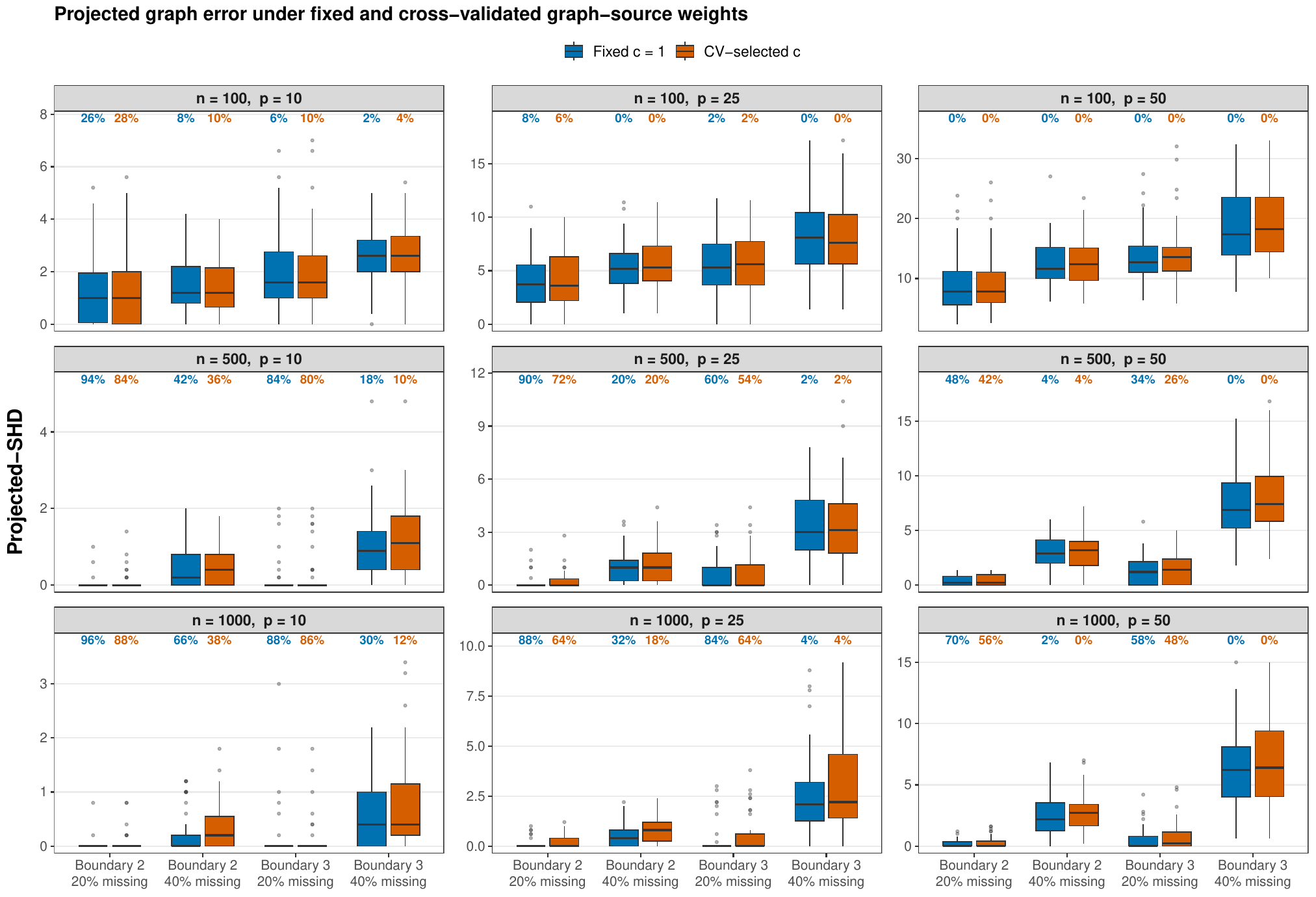}
    \caption{\small Projected SHD per source in the hybrid setting, where three sources provide raw data and two provide estimated graphs. Fixed $c=1$ is compared with $c$ selected by predictive cross-validation. Panels correspond to combinations of sample size $n$ and number of variables $p$, and the groups within each panel correspond to boundary sizes $l\in\{2,3\}$ and missingness levels of $20\%$ and $40\%$. Boxes summarize 50 paired replications. Percentages at the top of each panel report exact recovery of the projection-equivalence class. Results are based on the coordinate-descent procedure with five sources.}
    \label{fig:hybrid-c1-cv}
\end{figure}
\FloatBarrier

\section{Adaptation of CARE} \label{sec:CARE_adaptation}
The graph of the 21 red-blood-cell fatty acids was estimated using the alternative two-stage CARE estimator in Supplement~S.2 of \citet{Zhang2025CARE}, adapted to incorporate NHANES survey weights, covariate adjustment, common-parameter tuning, and stability-based coefficient thresholding. After the eligibility and complete-case restrictions described in the main text, source four contained 5,115 participants. Each centered log-ratio coordinate was regressed on age, sex, race/ethnicity, and fasting duration using phlebotomy weights $w_i=\mathrm{WTPH2YR}_i$. Age was modeled using a restricted cubic spline with knots at its weighted 5th, 35th, 65th, and 95th percentiles. CARE was applied to $S_w$, the survey-weighted covariance matrix of the resulting residuals. Both CARE tuning parameters were set to a common $\lambda$. Following Supplement~S.2, the two columnwise estimates for each off-diagonal pair were symmetrized by assigning the estimate with smaller absolute value to both corresponding entries. The common $\lambda$ was selected over $\{0,0.02,\ldots,4\}$ using five repeated random 50/50 training--validation splits, minimizing the mean validation loss $L(\Omega;S_w)=\tfrac12\operatorname{tr}(\Omega^{\mathsf T}S_w\Omega)-\operatorname{tr}(P^{\mathsf T}\Omega)$, where $S_w$ here denotes the validation residual covariance and $P=I_{21}-\boldsymbol{1}\boldsymbol{1}^{\mathsf T}/21$ is the centering matrix across the 21 fatty-acid coordinates. Within each split, all preprocessing parameters were estimated from the training observations and applied unchanged to the validation observations.

With CARE's regularization parameter fixed at the CV-selected $\widehat\lambda$, a separate coefficient threshold was selected using a StARS-inspired procedure \citep{liu2010stability}. Covariate adjustment and CARE estimation were repeated on 100 independently drawn subsamples, each containing $\lfloor10\sqrt{n}\rfloor=715$ participants sampled uniformly without replacement, retaining their survey weights. For each candidate threshold $t$, an edge was selected in a subsample when its estimated precision-matrix entry exceeded $t$ in absolute value. Letting $\widehat{\pi}_{jk}(t)$ denote the proportion of subsamples selecting pair $(j,k)$, the average instability and selected threshold were
$$
D(t)=\binom{21}{2}^{-1}\sum_{j<k}2\widehat{\pi}_{jk}(t)\{1-\widehat{\pi}_{jk}(t)\},
\qquad \widehat{t}=\min\left\{t:\max_{s\geq t}D(s)\leq0.05\right\}.
$$
Candidate thresholds were evaluated on a grid of 100 equally spaced values from zero to the largest absolute off-diagonal estimate across the subsample fits. For each candidate $t$, the maximum was taken over all grid values $s\geq t$. The final graph retained pairs satisfying $|\widehat{\Omega}_{jk}|>\widehat{t}$ in the full-data CARE estimate.

To evaluate whether the proposed tuning and thresholding strategy recovered graph structure accurately, a simulation study compared it with the original CARE procedure and several alternative choices. Each replicate generated a graph with 21 nodes and 20 edges according to one of three structures. Random graphs selected 20 distinct node pairs uniformly at random. Band graphs arranged the 21 nodes in a random order and connected each node only to its immediate predecessor and successor, forming a single chain. Hub graphs divided the nodes randomly into three groups of seven, each with one central node connected to the other six, and added two edges linking the three centers into a chain.

For each graph, with adjacency matrix $A$, the population precision matrix $\Omega_0$ had off-diagonal entries $(\Omega_0)_{jk}=-\rho A_{jk}$ and diagonal entries $(\Omega_0)_{jj}=0.20+0.40\deg(j)$. Thus, absent edges corresponded to zero precision-matrix entries, while larger $\rho$ produced stronger nonzero entries. The diagonal construction ensured positive definiteness. Each dataset contained $n=400$ independent latent observations $Y_i\sim N(0,\Omega_0^{-1})$. These were converted to positive measurements by exponentiation and then normalized to produce compositions, $X_{ij}=\exp(Y_{ij})/\sum_h\exp(Y_{ih})$. Only these compositions were supplied to the estimation procedures. There were 100 independent replicates per setting, yielding 900 datasets.

The comparison examined regularization selection and subsequent coefficient thresholding. The reference procedure was the original CARE estimator, which uses a single estimation stage and selects a separate regularization parameter for each column by cross-validation. The remaining procedures used the alternative estimator in Supplement~S.2 of \citet{Zhang2025CARE}. This estimator has two estimation stages: a preliminary fit controlled by $\lambda_1$, followed by an adaptive fit controlled by $\lambda_2$ that uses the preliminary estimates to scale its constraints. To simplify tuning, we imposed $\lambda_1=\lambda_2=\lambda$. The common value was selected either by cross-validation using repeated random 50/50 training--validation splits or by minimizing the BIC-type criterion $nL(\widehat\Omega;S)+|\widehat E|\log n$, where $S$ is the sample covariance matrix of the centered log-ratio measurements and $L$ is the quadratic loss defined above. Both selection rules searched the grid $\{0\}\cup\{0.05,0.10,\ldots,0.50\}\cup\{0.75,1.00,\ldots,2.00\}$. An additional comparison fixed $\lambda_1=\lambda_2=0$. Each of these three tuning choices was evaluated with and without subsequent coefficient thresholding, as described next.

For each common-$\lambda$ rule, graph recovery was evaluated in three ways: using the estimated nonzero entries directly, applying bootstrap-based coefficient thresholding, or applying subsampling-based coefficient thresholding. This yielded nine alternatives in addition to original CARE. Bootstrap thresholding used 100 samples of 400 observations drawn with replacement. Subsampling thresholding used 100 subsets of 200 observations drawn without replacement, where $200=10\sqrt{400}$. In both cases, $\lambda$ was selected once from the full dataset and held fixed throughout resampling. The resampled fits were used only to select the coefficient threshold, using 60 candidate cutoffs and the instability bound of 0.05 described above. The selected cutoff was then applied to the full-data estimate.

\begin{figure}[!htb]
\centering
\includegraphics[width=\linewidth]{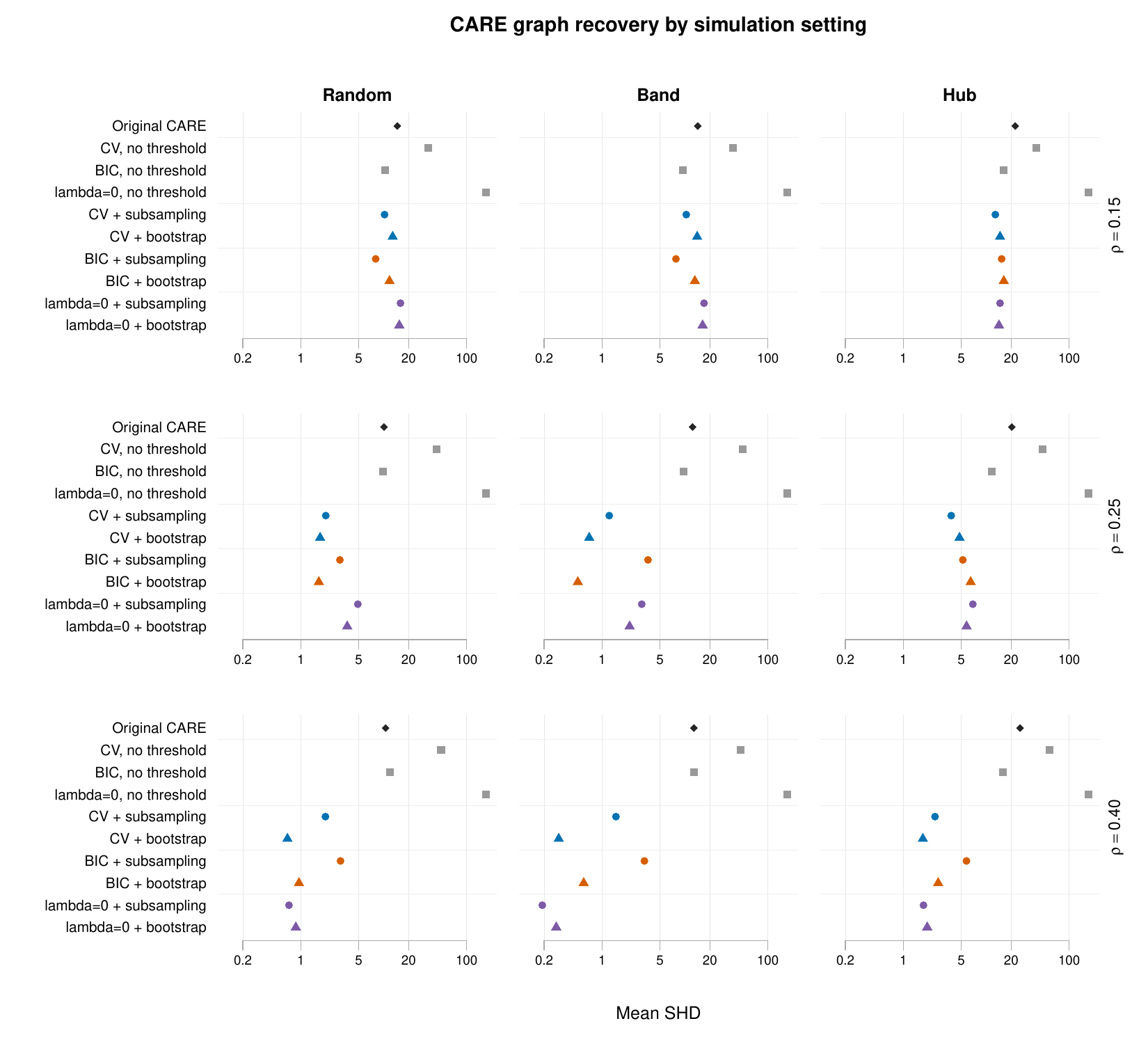}
\caption{Graph recovery for CARE and its tuning and thresholding variants, across graph structures (columns) and signal strengths $\rho$ (rows). Each point represents mean SHD over 100 replicates with $n=400$ observations, 21 nodes, and 20 true edges. Bootstrap and subsampling refer to the resampling scheme used to select the coefficient threshold.}
\label{fig:care-simulation}
\end{figure}
Graph recovery depended on both signal strength and graph structure, as seen in Figure~\ref{fig:care-simulation}. Thresholding markedly improved the common-$\lambda$ CV and fixed-zero estimators, whereas its benefit for BIC depended on the setting. Recovery generally improved as signals strengthened, and hub graphs were consistently harder than random or band graphs for the four CV/BIC thresholded procedures. At the weakest signal, BIC plus subsampling performed best for random and band graphs, while CV plus subsampling performed best for hub graphs. At the strongest signal, CV plus bootstrap outperformed CV plus subsampling in all three structures. Thus, no procedure dominated throughout. Nevertheless, CV plus subsampling improved on original CARE in every setting and achieved the lowest SHD averaged across settings, 5.15 compared with 15.87 for original CARE. These results motivated its use for source four. 

Applying the selected CV-plus-subsampling procedure to source four yielded $\widehat{\lambda}=0$ and a coefficient threshold of $\widehat{t}=25.6313$, with average subsample instability 0.04469. Although CV selected zero for both stage-specific tuning parameters, the final graph was sparsified by the separately selected coefficient threshold. It contained 12 edges among the 21 fatty acids, with eight isolated nodes, see Figure~\ref{fig:source4-care}.

\begin{figure}[!htb]
\centering
\includegraphics[width=0.64\linewidth]{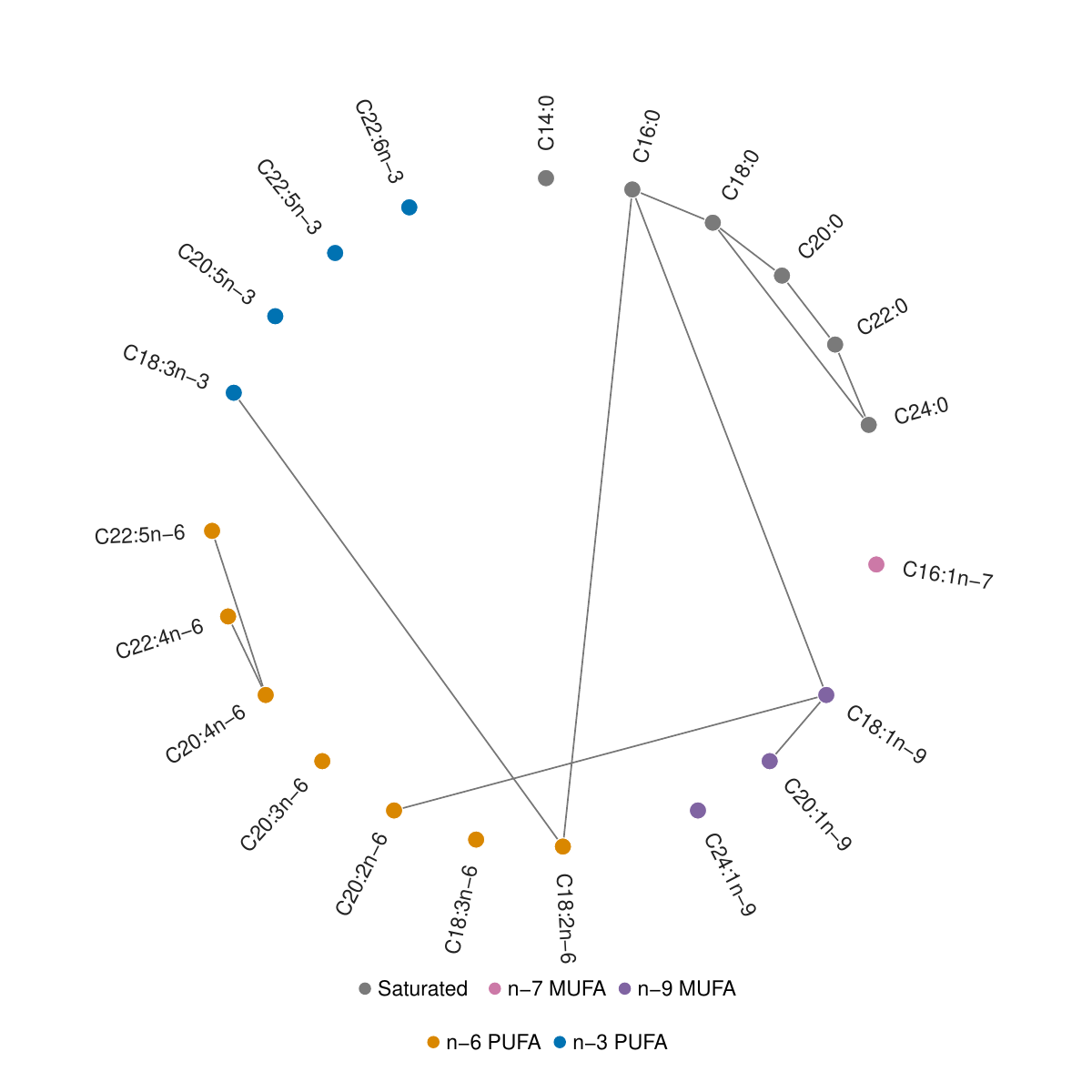}
\caption{Estimated source-four graph obtained from survey-weighted, covariate-adjusted CARE with CV-selected regularization and subsampling-selected coefficient thresholding. All 21 fatty acids are shown, grouped by family and ordered by chain length within each family. Colors identify fatty-acid families, and lines indicate the 12 selected edges.}
\label{fig:source4-care}
\end{figure}
\newpage
\bibliographystyle{agsm}
\bibliography{main.bib}

\end{document}